\PassOptionsToPackage{unicode}{hyperref}
\PassOptionsToPackage{hyphens}{url}
\PassOptionsToPackage{dvipsnames,svgnames,x11names}{xcolor}
\documentclass[12pt]{article}

\usepackage{amsmath,amssymb,amsfonts} 
\usepackage{bm} 
\DeclareMathOperator{\Var}{Var}

\newcommand\independent{\protect\mathpalette{\protect\independenT}{\perp}}
\def\independenT#1#2{\mathrel{\rlap{$#1#2$}\mkern2mu{#1#2}}}

\usepackage{amsthm} 
\newtheorem{theorem}{Theorem}

\newtheorem{Lemma}[theorem]{Lemma}

\theoremstyle{remark}
\newtheorem{assumption}{Assumption}
\newtheorem{remark}{Remark}

\usepackage[authoryear,sort]{natbib}
\usepackage{graphicx} 
\usepackage{caption} 
\usepackage{subcaption} 
\usepackage{booktabs}
\usepackage{pdflscape} 
\usepackage{rotating}
\usepackage{float}
\usepackage{placeins} 
\usepackage{algorithm}
\usepackage{algpseudocode}
\usepackage{pgfplots}
\pgfplotsset{compat=1.17}
\usetikzlibrary{arrows,patterns,positioning,arrows.meta,decorations.markings,fit,shapes}
\usepackage{url}
\usepackage{enumerate} 
\usepackage{colortbl} 
\usepackage{arydshln} 
\usepackage{pifont}

\usepackage{bbm} 
\allowdisplaybreaks[3]

\usepackage{tikz}
\usetikzlibrary{decorations.pathreplacing,calc}

\usepackage{booktabs} 
\usepackage{graphicx}

\usepackage[shortlabels]{enumitem}

\usepackage[utf8]{inputenc}
\usepackage{multirow}
\usepackage{csquotes} 
\theoremstyle{definition}

\usepackage{iftex}
\ifPDFTeX
 \usepackage[T1]{fontenc}
 \usepackage[utf8]{inputenc}
 \usepackage{textcomp} % provide euro and other symbols
\else % if luatex or xetex
 \usepackage{unicode-math}
 \defaultfontfeatures{Scale=MatchLowercase}
 \defaultfontfeatures[\rmfamily]{Ligatures=TeX,Scale=1}
\fi
\usepackage{lmodern}
\ifPDFTeX\else 
\fi
\IfFileExists{upquote.sty}{\usepackage{upquote}}{}
\IfFileExists{microtype.sty}{% use microtype if available
 \usepackage[]{microtype}
 \UseMicrotypeSet[protrusion]{basicmath} % disable protrusion for tt fonts
}{}
\makeatletter
\@ifundefined{KOMAClassName}{% if non-KOMA class
 \IfFileExists{parskip.sty}{%
 \usepackage{parskip}
 }{% else
 \setlength{\parindent}{2em}
 \setlength{\parskip}{6pt plus 2pt minus 1pt}}
}{% if KOMA class
 \KOMAoptions{parskip=half}}
\makeatother
\usepackage{xcolor}
\makeatletter
\ifx\paragraph\undefined\else
 \let\oldparagraph\paragraph
 \renewcommand{\paragraph}{
 \@ifstar
 \xxxParagraphStar
 \xxxParagraphNoStar
 }
 \newcommand{\xxxParagraphStar}[1]{\oldparagraph*{#1}\mbox{}}
 \newcommand{\xxxParagraphNoStar}[1]{\oldparagraph{#1}\mbox{}}
\fi
\ifx\subparagraph\undefined\else
 \let\oldsubparagraph\subparagraph
 \renewcommand{\subparagraph}{
 \@ifstar
 \xxxSubParagraphStar
 \xxxSubParagraphNoStar
 }
 \newcommand{\xxxSubParagraphStar}[1]{\oldsubparagraph*{#1}\mbox{}}
 \newcommand{\xxxSubParagraphNoStar}[1]{\oldsubparagraph{#1}\mbox{}}
\fi
\makeatother

\usepackage{longtable,booktabs,array}
\usepackage{calc} % for calculating minipage widths
\usepackage{etoolbox}
\makeatletter
\patchcmd\longtable{\par}{\if@noskipsec\mbox{}\fi\par}{}{}
\makeatother
\IfFileExists{footnotehyper.sty}{\usepackage{footnotehyper}}{\usepackage{footnote}}
\makesavenoteenv{longtable}
\usepackage{graphicx}
\makeatletter
\def\maxwidth{\ifdim\Gin@nat@width>\linewidth\linewidth\else\Gin@nat@width\fi}
\def\maxheight{\ifdim\Gin@nat@height>\textheight\textheight\else\Gin@nat@height\fi}
\makeatother
\setkeys{Gin}{width=\maxwidth,height=\maxheight,keepaspectratio}
\makeatletter
\def\fps@figure{htbp}
\makeatother

\makeatletter
\@ifpackageloaded{caption}{}{\usepackage{caption}}
\AtBeginDocument{%
\ifdefined\contentsname
 \renewcommand*\contentsname{Table of contents}
\else
 \newcommand\contentsname{Table of contents}
\fi
\ifdefined\listfigurename
 \renewcommand*\listfigurename{List of Figures}
\else
 \newcommand\listfigurename{List of Figures}
\fi
\ifdefined\listtablename
 \renewcommand*\listtablename{List of Tables}
\else
 \newcommand\listtablename{List of Tables}
\fi
\ifdefined\figurename
 \renewcommand*\figurename{Figure}
\else
 \newcommand\figurename{Figure}
\fi
\ifdefined\tablename
 \renewcommand*\tablename{Table}
\else
 \newcommand\tablename{Table}
\fi
}
\@ifpackageloaded{float}{}{\usepackage{float}}
\floatstyle{ruled}
\@ifundefined{c@chapter}{\newfloat{codelisting}{h}{lop}}{\newfloat{codelisting}{h}{lop}[chapter]}
\floatname{codelisting}{Listing}

\makeatother
\makeatletter
\@ifpackageloaded{caption}{}{\usepackage{caption}}
\@ifpackageloaded{subcaption}{}{\usepackage{subcaption}}
\makeatother

\ifLuaTeX
 \usepackage{selnolig} % disable illegal ligatures
\fi
\usepackage[authoryear]{natbib}
\usepackage{bookmark}

\IfFileExists{xurl.sty}{\usepackage{xurl}}{} % add URL line breaks if available
\hypersetup{
 pdftitle={Title},
 pdfauthor={Author 1; Author 2},
 pdfkeywords={3 to 6 keywords, that do not appear in the title},
 colorlinks=true,
 linkcolor={blue},
 filecolor={Maroon},
 citecolor={Blue},
 urlcolor={Blue},
 pdfcreator={LaTeX via pandoc}}

\newcommand{\anon}{1}

\begin{document}

\def\spacingset#1{\renewcommand{\baselinestretch}%
{#1}\small\normalsize} \spacingset{1}

%%%%%%%%%%%%%%%%%%%%%%%%%%%%%%%%%%%%%%%%%%%%%%%%%%%%%%%%%%%%%%%%%%%%%%%%%%%%%%
\if1\anon
{
 \title{\bf Semi-supervised inference for treatment heterogeneity}
 \author{%
 Yilizhati Anniwaer \\ School of Mathematics, Renmin University of China
 \and
 Yuqian Zhang\\
 Institute of Statistics and Big Data, Renmin University of China
 }
 \maketitle
} \fi

\if0\anon
{
 \bigskip
 \bigskip
 \bigskip
 \begin{center}
 {\LARGE\bf Title}
\end{center}
 \medskip
} \fi

\bigskip
\begin{abstract}
In causal inference, measuring treatment heterogeneity is crucial as it provides scientific insights into how treatments influence outcomes and guides personalized decision-making. In this work, we study semi-supervised settings where a labeled dataset is accompanied by a large unlabeled dataset, and develop semi-supervised estimators for two measures of treatment heterogeneity: the total treatment heterogeneity (TTH) and the explained treatment heterogeneity (ETH) of a simplified working model. We propose semi-supervised estimators for both quantities and demonstrate their improved robustness and efficiency compared with supervised methods. For ETH estimation, we show that direct semi-supervised approaches may result in efficiency loss relative to supervised counterparts. To address this, we introduce a re-weighting strategy that assigns data-dependent weights to labeled and unlabeled samples to optimize efficiency. The proposed approach guarantees an asymptotic variance no larger than that of the supervised method, ensuring its safe use. We evaluate the performance of the proposed estimators through simulation studies and a real-data application based on an AIDS clinical trial.
\end{abstract}

\noindent%

\vfill

\newpage
\spacingset{1.8} % DON'T change the spacing!

\section{Introduction}\label{sec-intro}

Heterogeneity in treatment effects is a central topic in causal inference. When treatment effects vary across individuals, assigning all individuals to the same treatment based on the population-level average may be suboptimal. Instead, tailoring treatment decisions to individual characteristics can lead to improved outcomes. Let \( A \in \{0,1\} \) denote a binary treatment variable, \( \bm{X} \in \mathbb{R}^d \) a vector of covariates, and \( Y = Y(A) \in \mathbb{R} \) the observed outcome. We adopt the potential outcomes framework and assume the existence of potential outcomes \( Y(a) \in \mathbb{R} \), representing the outcome an individual would receive if assigned treatment \( a \in \{0,1\} \). In this setting, accurate estimation of the conditional average treatment effect (CATE), defined as \( \tau(\bm{x}) := E[Y(1) - Y(0) \mid \bm{X} = \bm{x}] \), is critical for developing individualized treatment strategies and advancing precision medicine. Improved CATE estimation leads to more effective and targeted decision-making. Common methods for estimating the CATE include the S-learner, T-learner, X-learner, and DR-learner \citep{foster2023,Künzel2019,Kennedy2020,abrevaya2015,nie2021}.

In this work, we aim to understand treatment heterogeneity by estimating two key parameters of interest. The first parameter is the total treatment heterogeneity (TTH), defined as the variance of the CATE, $\theta_\mathrm{TTH}:=\Var[\tau(\bm{X})]$; see also \cite{Levy2021,hines2022variable}. The magnitude of the TTH indicates whether individualized treatment rules are necessary or if a simpler, one-size-fits-all approach is sufficient. While the TTH captures the total heterogeneity based on the true CATE, in practice, making decisions based on the true but potentially complex CATE can pose practical challenges. Instead, personalized treatment rules using simpler models, such as linear or decision tree models, with readily available or low-cost covariates, may improve both interpretability and feasibility. Let $\bm{W}\in\mathbb{R}^p$ be a sub-vector of $\bm{X}$. For any simplified working model $\tau^*(\bm{W})$, we also aim to evaluate the explained treatment heterogeneity (ETH), $\theta_\mathrm{ETH}:=\Var[\tau^*(\bm{W})]$, to determine whether the simplified model captures sufficient heterogeneity. Estimating the ETH (or the ratio between the ETH and the TTH) is crucial for striking a better balance between practicality and decision-making accuracy in follow-up treatments.

The accuracy of treatment heterogeneity estimation impacts the subsequent decision-making. However, in many applications, collecting outcome variables after treatment assignment is both costly and time consuming, especially when long-term outcomes are involved. These practical limitations reduce the number of labeled samples available, potentially leading to insufficient estimation accuracy. In contrast, baseline covariates are typically easy and inexpensive to collect in large volumes. Typical examples include biomedical studies using electronic health records and genome-wide association studies \citep{2,chakrabortty2018efficient}. In these settings, it is beneficial to adopt semi-supervised learning methods that leverage additional unlabeled data to improve estimation accuracy. Semi-supervised methods are also valuable for studying treatment heterogeneity in randomized trials with multiple treatment arms. When interest lies in comparing only two treatment groups, relying solely on samples from those groups may result in limited sample sizes. Rather than discarding all samples assigned to non-target treatment groups, one can improve efficiency by incorporating baseline covariate information from those additional groups using semi-supervised approaches; see our real-data analysis in Section \ref{sec:real}.

Semi-supervised learning has received considerable attention in recent years. For example, \cite{zhang2019semi,1} propose semi-supervised estimators for the mean response, \cite{Chen2023,chakrabortty2018efficient,azriel2022semi,deng2024optimal} study linear regression problems, and \cite{song2023general,angelopoulos2023prediction,zrnic2024cross,cai2025semi} investigate more general M-estimation frameworks. Moreover, \cite{2} develop semi-supervised estimators for the explained variance, but in a non-causal setting where the outcome is fully observed and the analysis is limited to correct linear models. In the context of causal inference, \cite{1,cheng2021robust,chakrabortty2022general,hou2025efficient,kallus2025role} introduce semi-supervised estimators for the average treatment effect, \cite{chakrabortty2022general} consider the estimation of quantile treatment effects, and \cite{sonabend2023semi} explore semi-supervised off-policy reinforcement learning. However, to our knowledge, no semi-supervised methods are currently available for estimating treatment effect heterogeneity, despite its practical importance in many applications.

\paragraph*{Main contribution} In this work, we propose semi-supervised estimators for both the total treatment heterogeneity (TTH) and the explained treatment heterogeneity (ETH). For TTH, we develop a semi-supervised estimator where all nuisance components can be flexibly estimated using non-parametric or machine learning methods. Our approach builds on the DR-learner \citep{foster2023,Kennedy2020} for CATE estimation and incorporates an additional debiasing step to estimate the quadratic-form parameter \( \Var[\tau(\bm{X})] \), leveraging both labeled and unlabeled samples. For ETH, we first demonstrate that direct semi-supervised estimation can suffer from efficiency loss compared to supervised methods. A similar issue arises in variance estimation under degenerate non-causal settings \citep{1,2,kim2024semi}, where the efficiency gain of semi-supervised methods depends on correct model specification. To overcome this limitation, we propose an optimally weighted semi-supervised estimator that assigns data-dependent weights to labeled and unlabeled samples, thereby ensuring safe and efficient use of the additional unlabeled data.

By incorporating additional unlabeled data, the proposed estimators improve upon existing supervised methods \citep{Levy2021,hines2022variable} in the following respects: (a) our methods achieve asymptotic normality under weaker sparsity conditions and remain valid even if the outcome regression models $\mu_a(\bm{x}):=E[Y(a)\mid \bm{X}=\bm{x}]$ (with $a\in\{0,1\}$) are misspecified, a scenario where supervised methods fail; (b) by exploiting additional unlabeled covariates, the asymptotic variance can be further reduced once asymptotic normality holds. Moreover, while existing supervised approaches focus only on estimating the TTH, our framework also extends to ETH estimation, allowing the working model to differ from the true CATE. The proposed optimally weighted method guarantees an asymptotic variance no larger than that of the supervised counterpart, due to its optimal weighting strategy.

\paragraph*{Organization}
The remainder of the paper is organized as follows. Section~\ref{sec2} introduces a semi-supervised framework for estimating the TTH and establishes the corresponding theoretical properties. Section~\ref{sec3} extends the analysis to ETH estimation, outlines the limitations of direct semi-supervised approaches, develops optimally weighted semi-supervised estimators, and establishes their theoretical guarantees. Section~\ref{sec4} presents simulation studies to evaluate the finite-sample performance under various data-generating scenarios, followed by an application to real-world clinical trial data. Additional discussions are provided in Section~\ref{sec5}. The Supplementary Material contains all proofs of the main results.

\paragraph*{Notation}
For any vector $\bm{v}\in\mathbb R^p$, we define $\|\bm{v}\|_{2}=\sqrt{\sum_{j=1}^{p}b_j^2}$, $\|\bm{v}\|_{\infty}=\max_{j}|v_j|$, and $\|\bm{v}\|_0=|\{j:v_j\neq0\}|$. %For any matrix $A\in \mathbbm{R}^{p\times p}$, we denote $\|\bm{a}\|_{2}=\sup_{x\neq0}\|Ax\|_{2}/\|x\|_2$ and $\|v\|_A^2=v^\top Av$. 
For any $\alpha>0$, let $\psi_{\alpha}(x)=\exp(x^{\alpha})-1$ for any $x>0$. The $\psi_{\alpha}$-Orlicz norm of a random variable $X$ is defined as $\|X\|_{\psi_{\alpha}}=\inf[c>0:E[\psi_{\alpha}(|X|/c)]\le1]$. For any measurable function \(f\), denote $\mathbb{E}_{\bm{X}}[f(X)]=\int f(x)d P_{\bm{X}}(x)$ and $\Var_{\bm{X}}[f(\bm{X})]=\mathbb{E}_{\bm{X}}[f^2(\bm{X})]-\{\mathbb{E}_{\bm{X}}[f(\bm{X})]\}^2$, where $P_{\bm{X}}$ is the marginal distribution of a random vector $\bm{X}$. %Moreover, for any $r>0$, let $\mu_r(f) = E[f - \mathbb{E}(f)]^r$ and $\mu_{r,X}(f) = \mathbb{E}_{\bm{X}}[f - \mathbb{E}_{\bm{X}}(f)]^r.$

\section{The total treatment heterogeneity (TTH)}\label{sec2}

\subsection{A semi-supervised TTH estimator}

Consider a semi-supervised setting where we have access to a labeled dataset \( (Z_i)_{i=1}^n = (\bm{X}_i, A_i, Y_i)_{i=1}^n \), along with a large unlabeled dataset \( (\bm{X}_i, A_i)_{i=n+1}^N \). Here, \( n \) denotes the labeled sample size, \( N \) is the total sample size, and \( m := N - n \) is the number of unlabeled samples. Semi-supervised learning is typically studied in the regime where \( N \gg n \), although our framework also accommodates scenarios where \( N \asymp n \). Following the standard missing completely at random (MCAR) assumption commonly adopted in the semi-supervised literature \citep{zhang2019semi,1,2,cheng2021robust,chakrabortty2018efficient,azriel2022semi,deng2024optimal,song2023general}, we assume that \( (\bm{X}_i, A_i) \) have the same joint distribution across the labeled and unlabeled groups. Moreover, we denote $(Z,\bm{X},A,Y,Y(1),Y(0))$ as an independent copy of $(Z_i,\bm{X}_i,A_i,Y_i,Y_i(1),Y_i(0))$.

Our target parameter is the total treatment heterogeneity (TTH), defined as the variance of the conditional average treatment effect (CATE), \( \theta_\mathrm{TTH} = \Var[\tau(\bm{X})] \). To identify the parameter of interest, we impose the following standard assumptions commonly used in the causal inference literature \citep{rosenbaum1983central,imbens2015causal}.

\begin{assumption}\label{a1}
 (a) (No unmeasured confounding) $\{Y(1),Y(0)\} \independent A \mid \bm{X}$. 
 (b) (Consistency) $Y = Y(A)$. 
 (c) (Overlap) There exists a constant $c_0\in(0,1/2)$ such that $P(c_0 \le \pi(\bm{X})\le 1-c_0)=1$, where $\pi(\bm{x}):=P(A=1\mid \bm{X}=\bm{x})$ is the propensity score function.
\end{assumption}

In what follows, we present a two-step semi-supervised estimation procedure.

\paragraph*{Step 1: Conditional average treatment effect (CATE) estimation}
Under the potential outcome framework, only one of the potential outcomes \( \{Y(1), Y(0)\} \) is observed for each individual, even within the labeled sample. As a result, the difference \( Y(1) - Y(0) \) is never directly observed, making it impossible to apply standard regression techniques to estimate \( \tau(\bm{x}) = E[Y(1) - Y(0) \mid \bm{X} = \bm{x}] \) directly. 

Fortunately, under the identification conditions in Assumption~\ref{a1}, several approaches are available for estimating the CATE function, including the S-learner, T-learner, X-learner, and DR-learner \citep{foster2023,Künzel2019,Kennedy2020,abrevaya2015,nie2021}. In this work, we adopt the doubly robust (DR) approach and apply the DR-learner \citep{foster2023,Kennedy2020}, which offers improved robustness to model misspecification and estimation error.

Specifically, we consider the doubly robust pseudo-outcome
\begin{equation}\label{171}
\varphi(Z) = \frac{A - \pi(\bm{X})}{\pi(\bm{X})[1 - \pi(\bm{X})]} \left[ Y - \mu_A(\bm{X}) \right] + \mu_1(\bm{X}) - \mu_0(\bm{X}),
\end{equation}
which depends only on observed variables and satisfies \( E[\varphi(Z) \mid \bm{X}] = \tau(\bm{X}) \). Under Assumption~\ref{a1}, the outcome regression functions can be expressed as \( \mu_a(\bm{x}) = E[Y(a) \mid \bm{X} = \bm{x}] = E[Y \mid \bm{X} = \bm{x}, A = a] \) for \( a \in \{0,1\} \), and thus can be estimated directly using the labeled data. Furthermore, when both \( (\bm{X}_i, A_i) \) are available in the large unlabeled dataset, the propensity score function \( \pi(\bm{x}) = P(A = 1 \mid \bm{X} = \bm{x}) \) can be estimated with high accuracy by leveraging the large total sample size. Let \( \hat\mu_a(\cdot) \) and \( \hat\pi(\cdot) \) denote generic estimators of \( \mu_a(\cdot) \) and \( \pi(\cdot) \), respectively. We then define the estimated pseudo-outcome as
\begin{align}\label{def:phihat}
\hat\varphi(Z) := \frac{A - \hat\pi(\bm{X})}{\hat\pi(\bm{X})[1 - \hat\pi(\bm{X})]} \left[ Y - \hat\mu_A(\bm{X}) \right] + \hat\mu_1(\bm{X}) - \hat\mu_0(\bm{X}).
\end{align}
A CATE estimator \( \hat\tau(\cdot) \) can be obtained by regressing \( \hat\varphi(Z_i) \) on \( \bm{X}_i \) using an arbitrary regression method. The inclusion of additional unlabeled samples contributes to a more accurate CATE estimate by improving the quality of the estimated propensity score $\hat\pi(\cdot)$.

\paragraph*{Step 2: Total treatment heterogeneity (TTH) estimation} The problem of semi-supervised variance estimation has been studied by \cite{1,2}, who focus on estimating the explained variance \( \Var[\mathbb{E}(Y \mid \bm{X})] \) in a non-causal setting. In this work, we extend their debiasing technique to the causal setting for estimating the total treatment heterogeneity (TTH), \( \theta_\mathrm{TTH} = \Var[\tau(\bm{X})] \), using the estimated CATE function \( \hat\tau(\cdot) \) obtained in Step 1. Specifically, we begin with the plug-in estimator
\[
\hat\theta_\mathrm{PI} := \frac{1}{N} \sum_{i=1}^N \hat h^2(\bm{X}_i), \quad \text{where} \quad \hat h(\bm{X}) := \hat\tau(\bm{X}) - \frac{1}{N} \sum_{i=1}^N \hat\tau(\bm{X}_i).
\]
This estimator admits the following decomposition:
\begin{align*}
&\hat\theta_\mathrm{PI} - \theta_{\mathrm{TTH}}\\
&\quad= \underbrace{\frac{1}{N} \sum_{i=1}^N \left[\tau(\bm{X}_i) - \tau \right]^2 - \theta_{\mathrm{TTH}}}_{=: \Delta_1}
- \underbrace{\frac{1}{N} \sum_{i=1}^N \left[ \hat h(\bm{X}_i) - h(\bm{X}_i) \right]^2}_{=: \Delta_2} 
+ \underbrace{\frac{2}{N} \sum_{i=1}^N \hat h(\bm{X}_i) \left[ \hat h(\bm{X}_i) - h(\bm{X}_i) \right]}_{=: \Delta_3},
\end{align*}
where \( h(\bm{x}) := \tau(\bm{x}) - \tau \), and \( \tau := \mathbb{E}[\tau(\bm{X})] = \mathbb{E}[Y(1) - Y(0)] \) denotes the average treatment effect (ATE). Under regularity conditions, \( \Delta_1 = O_p(N^{-1/2}) \) and is asymptotically normal; \( \Delta_2 \) depends quadratically on the CATE estimation error and is expected to be small; the dominant bias arises from the term \( \Delta_3 \).

\begin{algorithm}[!t]
 \centering
 \caption{The semi-supervised TTH estimator}
 \label{alg:Qhat_SL}

 \renewcommand{\baselinestretch}{1.1}\selectfont
 \footnotesize
 \begin{algorithmic}[1]
 \Require Labeled data \(\mathcal{L}=(\bm{X}_i,A_i,Y_i)_{i=1}^n\), unlabeled data \(\mathcal{U}=(\bm{X}_i,A_i)_{i=n+1}^N\), number of folds \(K \ge 3\).
 \State Partition $\mathcal{L}$ and $\mathcal{U}$ into $K$ disjoint folds of equal size, indexed by $\{I_1,\dots,I_K\}$ and $\{J_1,\dots,J_K\}$.
 \For{\(k \in \{1,\dots,K\}\)}
 \For{\(k' \in \{1,\dots,K\} \setminus \{k\}\)}
 \State Let $I_{-k,-k'}=\{1,\dots,n\}\setminus (I_k\cup I_{k'})$ and $J_{-k,-k'}=\{n+1,\dots,N\}\setminus (J_k\cup J_{k'})$.
 \State Obtain \(\hat\mu_a^{(-k,-k')}(\cdot)\) using samples in $I_{-k,-k'}$ for each $a\in\{0,1\}$.
 \State Obtain \(\hat\pi^{(-k,-k')}(\cdot)\) using samples in $G_{-k,-k'}:=I_{-k,-k'}\cup J_{-k,-k'}$.
 \State Compute imputed outcomes \(\hat\varphi^{(-k,-k')}(\cdot)\) as in \eqref{def:phihat}, using \(\hat\mu_a^{(-k,-k')}(\cdot)\) and \(\hat\pi^{(-k,-k')}(\cdot)\).
 \EndFor
 \State Obtain the CATE estimate $\hat\tau^{(-k)}(\cdot)$ regressing $\hat\varphi^{(-k,-k')}(Z_i)$ on \(\bm{X}_i\) using samples $i\in \cup_{k'\neq k}I_{k'}$.
 \State Let $\hat\varphi^{(-k)}(\cdot)=(K-1)^{-1}\sum_{k'\neq k}\hat\varphi^{(-k,-k')}(\cdot)$.
 \EndFor
 \State Compute
 $$\hat\tau=\frac{1}{N}\sum_{k=1}^K\sum_{i\in G_k}\hat\tau^{(-k)}(\bm{X}_i)+\frac{1}{n}\sum_{k=1}^K\sum_{i\in I_k}\Bigl[\hat\varphi^{(-k)}(Z_i)-\hat\tau^{(-k)}(\bm{X}_i)\Bigr].$$
 \State \Return The semi-supervised TTH estimator:
 $$\hat\theta_{\mathrm{TTH}}=\frac{1}{N}\sum_{k=1}^{K}\sum_{i\in G_k}\left[\hat h^{(-k)}(\bm{X}_i)\right]^2+\frac{2}{n}\sum_{k=1}^{K}\sum_{i\in I_k}\hat h^{(-k)}(\bm{X}_i)\left[\hat\varphi^{(-k)}(Z_i)-\hat\tau-\hat h^{(-k)}(\bm{X}_i)\right],$$
 where $G_k=I_k\cup J_k$ and $\hat h^{(-k)}(\cdot)=\hat\tau^{(-k)}(\cdot)-|G_k|^{-1}\sum_{j\in G_k}\hat\tau^{(-k)}(\bm{X}_j)$. 
 \end{algorithmic}
\end{algorithm}

To reduce this bias, we consider the representation \( h(\bm{X}_i) = \mathbb{E}[\varphi(Z) \mid \bm{X} = \bm{X}_i] - \tau \). We first introduce a semi-supervised estimator for \( \tau \), drawing on the approach of \cite{1}, who propose a semi-supervised estimator for the population mean \( \mathbb{E}(Y) \) of the form \( N^{-1}\sum_{i=1}^N \hat m(\bm{X}_i) + n^{-1}\sum_{i=1}^n [Y_i - \hat m(\bm{X}_i)] \), where \( \hat m(\bm{X}_i) \) is an estimate of the conditional mean function \( m(\bm{X}_i) = \mathbb{E}(Y_i \mid \bm{X}_i) \), possibly obtained using cross-fitting. Since \( \tau = \mathbb{E}[\varphi(Z)] \) is the expectation of the doubly robust pseudo-outcome \( \varphi(Z) \), we adapt this procedure by replacing \( Y_i \) and \( \hat m(\bm{X}_i) \) with \( \hat\varphi(Z_i) \) and \( \hat\tau(\bm{X}_i) \), respectively:
\[
\hat\tau = \frac{1}{N} \sum_{i=1}^N \hat\tau(\bm{X}_i) + \frac{1}{n} \sum_{i=1}^n \left[ \hat\varphi(Z_i) - \hat\tau(\bm{X}_i) \right].
\]
We then approximate \( h(\bm{X}_i) \) by the proxy \( \hat\varphi(Z_i) - \hat\tau \) and introduce the debiasing term $-2n^{-1}\sum_{i=1}^n[ \hat h(\bm{X}_i) - \hat\varphi(Z_i) + \hat\tau],$
constructed from the labeled samples, to mitigate the bias introduced by \( \Delta_3 \). This leads to the following semi-supervised TTH estimator:
\[
\hat\theta_\mathrm{TTH}=\hat\theta_\mathrm{PI}-\frac{2}{n}\sum_{i=1}^n\hat h(\bm{X}_i)\left[ \hat h(\bm{X}_i) - \hat\varphi(Z_i) + \hat\tau\right].
\]

To further mitigate the bias arising from the use of non-parametric or machine learning estimates, we incorporate an additional cross-fitting procedure. The detailed construction is provided in Algorithm \ref{alg:Qhat_SL}. Figure \ref{fig:cf_final_aligned} provides an illustration of the index sets generated by the cross-fitting procedure described in Algorithm \ref{alg:Qhat_SL}.

\begin{figure}[b!]
\centering
\begin{tikzpicture}[scale=1, every node/.style={scale=0.92}]

\def\block{2}
\def\total{5}
\def\width{\block*\total}
\def\centerY{-0.7}

\draw[thick] (0,0) rectangle (\width,1);
\foreach \i in {1,...,4} {
 \draw[thick] ({\block*\i},0) -- ({\block*\i},1);
}
\node[left,xshift=-0.5cm] at (0,0.5) {\textbf{Labeled Data}};

\draw[thick] (0,-2.4) rectangle (\width,-1.4);
\foreach \i in {1,...,4} {
 \draw[thick] ({\block*\i},-2.4) -- ({\block*\i},-1.4);
}
\node[left,xshift=-0.5cm] at (0,-1.9) {\textbf{Unlabeled Data}};

\fill[blue!20] (0,0) rectangle (2*\block,1);
\fill[blue!20] (0,-2.4) rectangle (2*\block,-1.4);

\node at ({0.5*\block},1.35) {$I_k$ ($k=1$)};
\node at ({1.5*\block},1.35) {$I_{k'}$ ($k'=2$)};
\node at ({2.5*\block},1.35) {$I_3$};
\node at ({3.5*\block},1.35) {$I_4$};
\node at ({4.1*\block},1.35) {$\cdots$};
\node at ({4.5*\block},1.35) {$I_K$};

\node at ({0.5*\block},-1.05) {$J_k$};
\node at ({1.5*\block},-1.05) {$J_{k'}$};
\node at ({2.5*\block},-1.05) {$J_3$};
\node at ({3.5*\block},-1.05) {$J_4$};
\node at ({4.1*\block},-1.05) {$\cdots$};
\node at ({4.5*\block},-1.05) {$J_K$};

\draw[dashed, thick, rounded corners=3pt] (-0.2,1.1) rectangle (\block + 0.2,-2.6);
\draw[dashed, thick, rounded corners=3pt] (2*\block - 0.1,1.1) rectangle (\width + 0.1,-2.6);

\draw[decorate,decoration={brace,amplitude=6pt,mirror},thick]
 (2*\block,-0.2) -- (\width,-0.2)
 node[midway,yshift=-0.5cm] {\small $I_{-k,-k'}$};

\draw[decorate,decoration={brace,amplitude=10pt,mirror},thick]
 (2*\block,-2.4) -- (\width,-2.4);

\node at ({\block*0.5}, -3.2) {\small $G_k = I_k \cup J_k$};
\node at ({\block*3.5}, -3.2) {\small $J_{-k,-k'}$};

\node at (\width + 1.2, \centerY) {\small $G_{-k,-k'}$};
\draw[<-, thick] (\width + 0.6, \centerY) -- (\width + 0.15, \centerY);
\end{tikzpicture}
\caption{Illustration of index sets generated by the cross-fitting procedure of Algorithm \ref{alg:Qhat_SL}}
\label{fig:cf_final_aligned}
\end{figure}
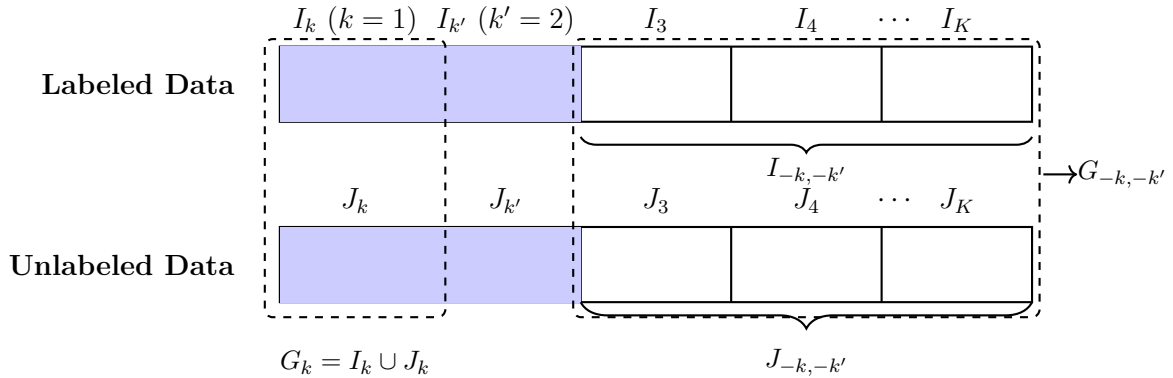

\subsection{Theoretical properties}

We assume the following conditions hold for all $a\in\{0,1\}$, $k, k' \in \{1, \dots, K\}$ with $k \ne k'$.

\begin{assumption}\label{a2}
For constants $C>0$ and $c_0\in(0,1/2)$, $P_{\bm{X}}[|\hat\mu_a^{(-k,-k')}(\bm{X})| < C]=1$ and $P_{\bm{X}}[\hat\pi^{(-k,-k')}(\bm{X}) \in (c_0,1 - c_0)]=1$ with probability approaching one. Moreover, $P[|Y(a)|<C]=1$. 
\end{assumption}

\begin{assumption}\label{a3}
 The nuisance estimation errors satisfy: with some target function $\mu_a^*(\cdot)$,
\begin{align*}
&\mathbb{E}_{\bm{X}}|\hat\pi^{(-k,-k')}(\bm{X}) - \pi(\bm{X})|^2 + \mathbb{E}_{\bm{X}}|\hat\mu_a^{(-k,-k')}(\bm{X}) - \mu_a^*(\bm{X})|^2= o_p(1).
\end{align*}
Moreover, either of the following conditions hold:
\begin{align*}
&(a)\;\;\mu_a^*(\cdot)=\mu_a(\cdot)\;\;\text{and}\;\;\mathbb{E}_{\bm{X}}|\hat\pi^{(-k,-k')}(\bm{X}) - \pi(\bm{X})|^2 \mathbb{E}_{\bm{X}}|\hat\mu_a^{(-k,-k')}({\bm{X}}) - \mu_a({\bm{X}})|^2 = o_p(n^{-1});\\
&(b)\;\;\mathbb{E}_{\bm{X}}|\hat{\pi}^{(-k,-k')}(\bm{X})-\pi(\bm{X})|^2=o_p(n^{-1}).
\end{align*}

\end{assumption}

\begin{assumption}\label{a4}
 The CATE estimator satisfies $P_{\bm{X}}[|\hat {\tau}^{(-k)}(\bm{X})|<C]=1$ and
\begin{equation*}
 \Var_{\bm{X}}[\hat {\tau}^{(-k)}(\bm{X})-\tau(\bm{X})]=o_p(n^{-1/2}).
\end{equation*}
\end{assumption}

The boundedness conditions in Assumption~\ref{a2} are standard when non-parametric nuisance estimates are considered, as seen in \cite{Chernozhukov2018}. Alternatively, these conditions can be relaxed to bounded moment conditions, provided that higher-order moment assumptions on the nuisance estimation errors are imposed. Assumption~\ref{a3} requires a correctly specified propensity score model but allows for a potentially misspecified outcome model, with target $\mu_a^*(\cdot)\neq\mu_a(\cdot)$, provided the propensity score is estimated accurately enough so that $\mathbb{E}_{\bm{X}}|\hat{\pi}^{(-k,-k')}(\bm{X})-\pi(\bm{X})|^2=o_p(n^{-1})$. Notably, since the error in propensity score estimation depends on the total sample size $N$ rather than the labeled sample size $n$, this condition can be satisfied even with non-parametric estimators as long as $N$ is sufficiently large. On the other hand, when the outcome model is correctly specified with $\mu_a^*(\cdot)=\mu_a(\cdot)$, a standard product-rate condition on the nuisance estimation errors is required, as in \cite{Chernozhukov2018,Levy2021,hines2022variable,1}. Assumption~\ref{a4} requires a consistent estimate for the true CATE function. In Section~\ref{sec3}, we discuss cases where the CATE model is misspecified -- in such settings, the target parameter becomes the explained treatment heterogeneity (ETH) of the working CATE model. Illustrations of the sparsity conditions under which the required convergence rates in Assumptions~\ref{a3}-\ref{a4} hold for high-dimensional parametric models are provided in Theorem~\ref{t1}.

Denote $\xi := \varphi^*(Z) - \tau(\bm{X})$ and 
\begin{align}
\varphi^*(Z) := \frac{A - \pi^*(\bm{X})}{\pi^*(\bm{X})\left[1 - \pi^*(\bm{X})\right]}\left[ Y - \mu_A^*(\bm{X}) \right] + \mu_1^*(\bm{X}) - \mu_0^*(\bm{X}), \label{def:phistar}
\end{align}
where $\pi^*(\cdot)$ and $\mu_a^*(\cdot)$ are target functions that serve as population-level approximations to the true nuisance functions. The following theorem shows the asymptotic normality of the proposed semi-supervised TTH estimator in the case where $\pi^*(\cdot)=\pi(\cdot)$, while $\mu_a^*(\cdot)$ may differ from $\mu_a(\cdot)$ for $a \in \{0,1\}$.

\begin{theorem}\label{t4}
Let Assumptions~\ref{a1}-\ref{a4} hold. 
Then, as $N,n\to\infty$,
\[
 \frac{\sqrt{n}(\hat\theta_\mathrm{TTH}-\theta_\mathrm{TTH})}{\sigma_\mathrm{TTH}}
 \ \xrightarrow{d}\ N(0,1),
\]
provided that $\sigma_\mathrm{TTH}^2 > c$ with some constant $c> 0$, where 
\[
\sigma_\mathrm{TTH}^2 := \Var[2\xi h(\bm{X})] + \frac{n}{N}\Var[h^2(\bm{X})].
\]
\end{theorem}

\begin{remark}[Enhanced robustness and efficiency]\label{remark1}
In the following, we discuss the benefits of incorporating the additional unlabeled samples compared with fully supervised methods. 

\begin{itemize}
\item \textbf{Model robustness.} To establish asymptotic normality, supervised methods \citep{Levy2021,hines2022variable} require both the propensity score model and the outcome regression model to be consistently estimated. In contrast, as stated in Assumption~\ref{a3}, we allow the outcome regression model to be misspecified, $\mu_a^*(\cdot)\neq\mu_a(\cdot)$, provided that the total sample size $N$ is sufficiently large relative to the labeled size $n$.
\item \textbf{Rate robustness.} When both nuisance models are correctly specified, we require a product-rate condition on the nuisance estimation errors in Assumption~\ref{a3}(a), as is standard in the supervised literature. However, by leveraging the additional unlabeled samples, we achieve a much faster convergence rate for the propensity score estimation. As a result, the same product-rate condition translates into considerably weaker model complexity requirements for the nuisance models.
\item \textbf{Enhanced efficiency.} Once asymptotic normality is established, our method yields an asymptotic variance $\sigma_\mathrm{TTH}^2 = \Var[2\xi h(\bm{X})] + (n/N)\Var[h^2(\bm{X})]$. In contrast, supervised methods yield an asymptotic variance $\Var[2\xi h(\bm{X})] + \Var[h^2(\bm{X})]$. Thus, the use of additional unlabeled samples reduces the variance by $(1-n/N)\Var[h^2(\bm{X})]$. This reduction is asymptotically non-negligible unless $n/N \to 1$, meaning the unlabeled size is negligible compared to the labeled size, or $\Var[h^2(\bm{X})]=\Var\{[\tau(\bm{X})-\tau]^2\}\to0$, meaning the CATE function is nearly constant.
\end{itemize}

To sum up, incorporating additional unlabeled samples relaxes the conditions on model correctness and model complexity required for asymptotic normality. Moreover, even when asymptotic normality holds, it yields a smaller asymptotic variance compared with supervised methods.
\end{remark}

\begin{remark}[Situations where unlabeled samples contain only covariates]\label{remark:onlyX}
The methods and results developed above assume that the unlabeled samples contain both covariates and treatment variables, \( (\bm{X}_i, A_i)_{i=n+1}^N \). In some applications, however, treatment variables may be difficult to collect, leaving only covariates \( (\bm{X}_i)_{i=n+1}^N \) available in the unlabeled samples. In such cases, Algorithm~\ref{alg:Qhat_SL} can still be applied, except that the propensity score in Step~6 must be estimated solely from the labeled samples. Consequently, the additional unlabeled samples no longer provide the robustness improvement discussed in Remark~\ref{remark1}, since there are no extra treatment variables to fit the propensity score model. Nevertheless, the efficiency gain remains, even when the unlabeled samples contain only covariates.
\end{remark}

\section{The explained treatment heterogeneity (ETH)}\label{sec3}
\label{sec:meth}

In this section, we consider settings where directly estimating the true CATE function \(\tau(\cdot)\) may be challenging due to model complexity, or not desirable because of interpretation or implementation concerns. In such cases, it is natural to consider simplified working models \(\tau^*(\cdot)\) of the true CATE. Moreover, for practical, ethical, or economic reasons, followed-up individualized treatment strategies may sometimes be developed based only on a subset of the covariates, denoted as \(\bm{W}\in\mathbb{R}^p\). For simplicity, we assume that the first coordinates of \(\bm{X}\) and \(\bm{W}\) are both 1, corresponding to intercept terms. We focus on linear working models \(\tau^*(\bm{w})=\bm{w}^\top\bm{\beta}^*\) for interpretability, where the population-level best linear slope is defined as
\begin{align}\label{def:beta}
\bm{\beta}^* := \arg\min_{\bm{\beta}\in\mathbb{R}^p} E\left[\tau(\bm{X})-\bm{W}^\top\bm{\beta}\right]^2.
\end{align}

A key quantity in this setting is the explained treatment heterogeneity (ETH) of the working model, \(\theta_\mathrm{ETH}=\Var[\tau^*(\bm{W})]\), which quantifies the heterogeneity captured by the chosen covariates and linear model. An ETH close to the total treatment heterogeneity (TTH) indicates that the simplified working model already captures most of the heterogeneity in the true CATE.

\subsection{A direct semi-supervised ETH estimation}\label{sec3.1}

We begin by presenting the results in high-dimensional parametric settings, where we consider linear models for both the outcome regression and the CATE function, together with a logistic model for the propensity score. We introduce Algorithm \ref{alg:theta_es}, adapted from Algorithm \ref{alg:Qhat_SL}, with all nuisance estimates obtained using \(\ell_1\)-regularized methods. Since a simplified working model is specified for the CATE function, the semi-supervised estimator in Algorithm \ref{alg:theta_es} targets the ETH of the working model, \(\theta_\mathrm{ETH}=\Var[\tau^*(\bm{W})]\), which is generally no larger than \(\theta_\mathrm{TTH}=\Var[\tau(\bm{X})]\).

\begin{algorithm}[t!]
 \centering
 \caption{The direct semi-supervised ETH estimator}
 \label{alg:theta_es}
 \renewcommand{\baselinestretch}{1.1}\selectfont
 \footnotesize
 \begin{algorithmic}[1]
 \Require Labeled data \(\mathcal{L}=(\bm{X}_i,A_i,Y_i)_{i=1}^n\), unlabeled data \(\mathcal{U}=(\bm{X}_i,A_i)_{i=n+1}^N\), number of folds \(K \ge 3\).
 \State Partition $\mathcal{L}$ and $\mathcal{U}$ into $K$ disjoint folds of equal size, indexed by $\{I_1,\dots,I_K\}$ and $\{J_1,\dots,J_K\}$.
 \For{\(k \in \{1,\dots,K\}\)}
 \For{\(k' \in \{1,\dots,K\} \setminus \{k\}\)}
 \State Let $I_{-k,-k'}=\{1,\dots,n\}\setminus (I_k\cup I_{k'})$ and $J_{-k,-k'}=\{n+1,\dots,N\}\setminus (J_k\cup J_{k'})$.
 \State For each $a\in\{0,1\}$, obtain $\hat\mu_a^{(-k,-k')}(\bm{x})=\bm{x}^\top\hat{\bm{\alpha}}_{a}^{(-k,-k')}$ through solving
\begin{align}\label{def:alphahat}
\hat{\bm{\alpha}}_a^{(-k,-k')} = \arg\min_{\bm{\alpha}\in\mathbb{R}^{d}} \Bigl[\left|I_{-k,-k'}\right|^{-1} \sum_{i\in I_{-k,-k'}}\mathbbm{1}_{A_i=a}\Bigl(Y_i - \bm{X}_i^\top\bm{\alpha}\Bigr)^2 + \lambda_{\alpha}\|\bm{\alpha}\|_1 \Bigr].
\end{align}
 \State Denote $G_{-k,-k'}:=I_{-k,-k'}\cup J_{-k,-k'}$ and obtain $\hat\pi^{(-k,-k')}(\bm{x})=\phi(\bm{x}^\top\hat{\bm{\gamma}}^{(-k,-k')})$ through solving
\begin{align}\label{def:gamma}
\hat{\bm{\gamma}}^{(-k,-k')} = \arg\min_{\bm{\gamma}\in\mathbb{R}^{d}} \Bigl(|G_{-k,-k'}|^{-1} \sum_{i\in G_{-k,-k'}} \Bigl\{\log\Bigl[1+\exp\Bigl(\bm{X}_i^\top\bm{\gamma}\Bigr)\Bigr]-A_i\,\bm{X}_i^\top\bm{\gamma}\Bigr\} + \lambda_\gamma\|\bm{\gamma}\|_1 \Bigr).
\end{align}
 \State Compute imputed outcomes \(\hat\varphi^{(-k,-k')}(\cdot)\) as in \eqref{def:phihat}, using \(\hat\mu_a^{(-k,-k')}(\cdot)\) and \(\hat\pi^{(-k,-k')}(\cdot)\).
 \EndFor
 \State Obtain
\begin{align}\label{def:betahat}
\hat{\bm{\beta}}^{(-k)} = \arg\min_{\bm{\beta}\in\mathbb{R}^{p}} \Bigl\{|\cup_{k'\neq k}I_{k'}|^{-1}\sum_{k'\neq k}\sum_{i\in I_{k'}} \Bigl[\hat\varphi^{(-k,-k')}(Z_i)-\bm{W}_i^\top\bm{\beta}\Bigr]^2 + \lambda_\beta\|\bm{\beta}\|_1 \Bigr\}.
\end{align}
 \State Let $\hat\varphi^{(-k)}(\cdot)=(K-1)^{-1}\sum_{k'\neq k}\hat\varphi^{(-k,-k')}(\cdot)$, $G_k=I_k\cup J_k$, and $\hat {\bm{D}}_i^{(k)}=\bm{W}_i-|G_k|^{-1}\sum_{j\in G_k}\bm{W}_j$.
 \EndFor
 \State Compute 
\begin{align}\label{def:tau_para}
\hat\tau_\mathrm{para}=\frac{1}{N}\sum_{k=1}^K\sum_{i\in G_k}\bm{W}_i^\top\hat{\bm{\beta}}^{(-k)}+\frac{1}{n}\sum_{k=1}^K\sum_{i\in I_k}\Bigl[\hat\varphi^{(-k)}(Z_i)-\bm{W}_i^\top\hat{\bm{\beta}}^{(-k)}\Bigr].
\end{align}
\State For each \(k \in \{1,\dots,K\}\), denote $\hat{\epsilon}_i^{(k)}=\hat\varphi^{(-k)}(Z_i)-\hat\tau_\mathrm{para}-\hat {\bm{D}}_i^{(k)^\top}\hat{\bm{\beta}}^{(-k)}$.
 \State \Return The semi-supervised ETH estimator:
\begin{align}
\hat\theta_{\mathrm{ETH}}=\frac{1}{N}\sum_{k=1}^{K}\sum_{i\in G_k}\Bigl(\hat {\bm{D}}_i^{(k)^\top}\hat{\bm{\beta}}^{(-k)}\Bigr)^2+\frac{2}{n}\sum_{k=1}^{K}\sum_{i\in I_k}\hat {\bm{D}}_i^{(k)^\top}\hat{\bm{\beta}}^{(-k)}\hat{\epsilon}_i^{(k)}.\label{def:ETH-para}
\end{align}
 \end{algorithmic}
\end{algorithm}

We define the population-level best linear and logistic slopes for the outcome regression and propensity score models as 
\begin{align*}
\bm{\alpha}_a^*&:=\arg\min_{\alpha\in\mathbb R^d}E\left[\mathbbm{1}_{A=a}\left(Y - \bm{X}^\top\bm{\alpha}\right)^2\right]\;\;\text{for each}\;\;a\in\{0,1\},\\
\bm{\gamma}^*&:=\arg\min_{\bm{\gamma}\in\mathbb R^d}E\left\{\log\left[1+\exp\left(\bm{X}^\top\bm{\gamma}\right)\right]-A\bm{X}^\top\bm{\gamma}\right\}.
\end{align*}
For simplicity, we define the sparsity levels as \(s_\alpha=\max(\|\bm{\alpha}_1\|_0,\|\bm{\alpha}_0\|_0,1)\), \(s_\beta=\max(\|\bm{\beta}^*\|_0,1)\), and \(s_\gamma=\max(\|\bm{\gamma}^*\|_0,1)\), so as to avoid degenerate cases with zero sparsity, where $\bm{\beta}^*$ is given in \eqref{def:beta}. Consider working models $\mu_a^*(\bm{x})=\bm{x}^\top\bm{\alpha}_a^*$ for each $a\in\{0,1\}$ and $\pi^*(\bm{x})=\phi(\bm{x}^\top\bm{\gamma}^*)$, where $\phi(u)=\exp(u)/[1+\exp(u)]$ denotes the logistic function. The following regularity assumptions are imposed to establish the asymptotic properties of the proposed estimator.

\begin{assumption}\label{a5}
 Define $\zeta_a=\mathbbm{1}_{A=a}\left[Y-\mu_a^*(\bm{X})\right]$ for each $a\in\{0,1\}$. Suppose that there exist constants $\sigma_{\zeta},C>0$ such that $\zeta_a$ is sub‐Gaussian with 
 $\|\zeta_a\|_{\psi_2}\le\sigma_{\zeta}$ and $\mathbb{E}[Y(a)]^2<C$. 
\end{assumption}

\begin{assumption}\label{a6}
 Let $\bm{X}$ be a sub‐Gaussian random vector satisfying $\|\bm{X}^\top \bm{v}\|_{\psi_2}\le\sigma_x\|\bm{v}\|_2$ for all $\bm{v}\in\mathbbm{R}^{d}$, $\|\bm{X}\|_\infty = O(1)$, and $\lambda_{\min}[\mathbb{E}(\bm{X}\bm{X}^\top)]\ge\kappa_l$ with some constants $\sigma_x,\kappa_l>0.$
\end{assumption}

\begin{assumption}\label{a7}
Let $P(c_0 \le \pi^*(\bm{X})\le 1-c_0)=1$ with some constant $c_0\in(0,1/2).$
\end{assumption}

Choose some tuning parameters $\lambda_{\beta}\asymp\sqrt{\log(p)/n}$ and $\lambda_{\alpha}\asymp\lambda_{\gamma}\asymp\sqrt{\log(d)/n}$. Then, the following theorem characterizes the CATE estimation error.

\begin{theorem}\label{t1}
 Let Assumptions~\ref{a1}, \ref{a5}, \ref{a6}, and \ref{a7} hold. Suppose that either (a) $\mu_a^*(\cdot)= \mu_a(\cdot)$ for each $a\in\{0,1\}$ or (b) $\pi^*(\cdot)=\pi(\cdot)$. Moreover, let $s_{\alpha}\log(d)+s_{\beta}\log(p)=o(n)$ and $s_{\gamma}\log(d) = o(N)$. Then, as $N,d,p\to\infty$,
 \begin{align*}
 &\|\hat{\bm{\beta}}^{(-k)}-\bm{\beta}^*\|_{2}= O_p\left(\sqrt{\frac{s_{\beta}\log(p)}{n}}+R_n\right),\;\;\text{where}\\
 &R_n:= \sqrt{\frac{s_{\alpha}s_{\gamma}\log^2(d)}{nN}}
 + \sqrt{\frac{s_{\gamma}\log(d)}{N}}\left(\mathbbm{1}_{\mu_1^*(\cdot)\neq \mu_1(\cdot)}+\mathbbm{1}_{\mu_0^*(\cdot)\neq \mu_0(\cdot)}\right)
 + \sqrt{\frac{s_{\alpha}\log(d)}{n}}\mathbbm{1}_{\pi^*(\cdot)\neq \pi(\cdot)}.
 \end{align*}
\end{theorem}

As shown in Theorem \ref{t1}, the linear slope estimator for the CATE consistently estimates the population slope \eqref{def:beta} when either the outcome regression model or the propensity score model is correctly specified. Under this condition, $\bm{\beta}^*$ coincides with the best linear slope obtained by replacing $\tau(\bm{x})$ in \eqref{def:beta} with $\varphi^*(Z)$ from \eqref{def:phistar}, thereby ensuring the validity of using the plug-in estimate of pseudo-outcomes.

The CATE estimation error consists of two terms. The first term, $\sqrt{s_{\beta}\log(p)/n}$, reflects the complexity of the CATE model itself, and can be regarded as the ``oracle'' estimation error if the true pseudo-outcomes \eqref{171} were available, which requires knowledge of the true nuisance functions $(\mu_1,\mu_0,\pi)$. The second term, $R_n$, captures the imputation error introduced by the initial nuisance estimation errors, and therefore depends on the complexity and correctness of the outcome regression and propensity score models. 

When all nuisance models are correctly specified, $R_n$ depends on the product of the nuisance estimation errors, which can be negligible compared to the oracle error when $N$ is sufficiently large, with $N\gg s_\alpha s_\gamma\log^2(d)/\{s_\beta\log(p)\}$. When $\pi^*(\cdot)\neq \pi(\cdot)$, $R_n$ includes an additional error term $\sqrt{s_{\alpha}\log(d)/n}$, corresponding to the estimation error of the outcome regression. Conversely, when $\mu_1^*(\cdot)\neq \mu_1(\cdot)$ or $\mu_0^*(\cdot)\neq \mu_0(\cdot)$, $R_n$ includes an additional error term $\sqrt{s_{\gamma}\log(d)/N}$, corresponding to the estimation error of the propensity score. Thanks to the availability of additional unlabeled samples, this additional error can also be negligible when $N$ is large enough, with $N\gg ns_\gamma\log(d)/\{s_\beta\log(p)\}$.

\begin{theorem}\label{t2}
Let the assumptions of Theorem~\ref{t1} hold. Then, as $N,d,p\to\infty$,
\[
 \hat \theta_{\mathrm{ETH}} = \theta_{\mathrm{ETH}} + O_p\Bigl(R_n + \frac{s_{\beta}\,\log(p)}{n} + \frac{1}{\sqrt n}\Bigr).
\]
Moreover, let $s_{\alpha}s_{\gamma}\log^2(d) =o(N)$, $s_{\beta}\log(p)= o(\sqrt{n})$, and $\pi^*(\cdot)=\pi(\cdot)$. Assume either (a) $\mu_a^*(\cdot)=\mu_a(\cdot)$ for each $a\in\{0,1\}$, or (b) $ns_\gamma\log(d)=o(N)$. Denote $\bm{D}:=\bm{W}-\mathbb{E}(\bm{W})$ and $\epsilon:=\varphi^*(Z)-\bm{W}^\top\bm{\beta}^*$. Suppose that $\sigma_{\mathrm{para}}^2 := A + (n/N)B + 2(n/N)C>c$ with some constant $c>0$, where $A:=\Var(2\epsilon \bm{D}^\top\bm{\beta}^*)$, $B:=\Var[(\bm{D}^\top\bm{\beta}^*)^2]$, and $C:=\mathrm{Cov}[2\epsilon \bm{D}^\top \bm{\beta}^*,(\bm{D}^\top\bm{\beta}^*)^2].$ Then, as $N,d,p\to\infty$,
\[
 \frac{\sqrt{n}\,(\hat \theta_{\mathrm{ETH}} - \theta_{\mathrm{ETH}})}{\sigma_{\mathrm{para}}}
 \;\xrightarrow{d}\; N(0,1).%\quad\text{and}\quad\hat\sigma_{\mathrm{para}}^2=\sigma_{\mathrm{para}}^2+o_p(1),
\]
%where
%$$\hat\sigma_{\mathrm{para}}^2=K^{-1}\sum_{k=1}^K\left(\hat{A}^{(k)} + \frac{n\hat{B}^{(k)}}{N}+ \frac{2n\hat{C}^{(k)}}{N}\right).$$
\end{theorem}

As shown in Theorem \ref{t2}, when the total sample size $N$ is sufficiently large, it is enough to have a correctly specified propensity score model, even if the outcome regression model is misspecified, provided that the labeled size satisfies $n\gg s_\alpha\log(d)+s_\beta^2\log^2(p)$. In contrast, without additional unlabeled samples, the supervised version requires both the outcome regression and propensity score models to be correctly specified, along with a stronger condition on the labeled size, $n\gg s_{\alpha}s_{\gamma}\log^2(d)+s_\beta^2\log^2(p)$.

\begin{remark}[Efficiency of the direct approach]\label{remark:eff-direct}
For asymptotic efficiency, when the sub-CATE model is correctly specified with $\tau^*(\bm{w})=E[\tau(\bm{X})\mid \bm{W}=\bm{w}]$, we have $\mathbb{E}(\epsilon\mid \bm{W})=E[\tau(\bm{X})\mid \bm{W}]-\tau^*(\bm{W})=0$ provided that either the outcome regression model or the propensity score model is correctly specified. In this case, $C=0$, and the asymptotic variance reduces to $\sigma_{\mathrm{para}}^2=A+(n/N)B$. In comparison, the supervised estimator has asymptotic variance $A+B$. Thus, the use of additional unlabeled samples provides an efficiency gain of $(1-n/N)B\geq0$, ensuring that the semi-supervised estimator is at least as efficient as the supervised one. 

However, in more general cases where the true sub-CATE model deviates from the linear working model, the semi-supervised asymptotic variance $\sigma_{\mathrm{para}}^2=A+(n/N)B+2(n/N)C$ may be either larger or smaller than the supervised variance $\sigma_{\mathrm{sup}}^2:=A+B+2C$, since the term $C$ can take either positive or negative values. A similar phenomenon arises in semi-supervised variance estimation in the simpler non-causal settings; see Section 3.2 of the Supplementary Material in \cite{1} and Section 3 of \cite{kim2024semi} for related discussions. The approaches in these works, as well as the direct estimation method in Algorithm \ref{alg:theta_es}, do not guarantee safe use of the semi-supervised method, since its efficiency may fall below that of supervised methods. In what follows, we address this limitation by introducing a re-weighting strategy that ensures the semi-supervised estimator attains efficiency no worse than its supervised counterpart.
\end{remark}

\subsection{Optimally weighted semi-supervised estimation}\label{sec3.2}

We first revisit the proposed semi-supervised ETH estimator in Algorithm \ref{alg:theta_es}. In the construction of \eqref{def:ETH-para}, the ETH estimator has two components. The first part, $N^{-1}\sum_{k=1}^{K}\sum_{i\in G_k}(\hat {\bm{D}}_i^{(k)^\top}\hat{\bm{\beta}}^{(-k)})^2$, is a plug-in estimate of the ETH, while the second part, $2n^{-1}\sum_{k=1}^{K}\sum_{i\in I_k}\hat {\bm{D}}_i^{(k)^\top}\hat{\bm{\beta}}^{(-k)}\hat{\epsilon}_i^{(k)}$, is a debiasing term that reduces the bias from the CATE estimation error. Since the debiasing term only involves labeled samples once the nuisance estimates are obtained, additional unlabeled samples do not improve its efficiency. In contrast, the plug-in component becomes more efficient with additional unlabeled samples. However, this does not guarantee an overall efficiency gain, because the debiasing term may introduce negative correlations with the plug-in term, so a more accurate plug-in estimate does not necessarily yield a more efficient overall estimator, as observed in Remark \ref{remark:eff-direct}.

In the above construction, the plug-in estimate $N^{-1}\sum_{k=1}^{K}\sum_{i\in G_k}(\hat {\bm{D}}_i^{(k)^\top}\hat{\bm{\beta}}^{(-k)})^2$ assigns equal weights $1/N$ to all labeled and unlabeled samples. In comparison, a supervised plug-in estimate assigns weights $1/n$ to labeled samples and zero weights to unlabeled samples. To achieve an efficient semi-supervised estimator, we first introduce a more general re-weighting construction, assigning weights $\omega_L$ and $\omega_U$ to labeled and unlabeled samples, respectively. Specifically, we consider the following re-weighted estimator:
\begin{align*}
\hat\theta_{\mathrm{RW}}=\omega_L\sum_{k=1}^{K}\sum_{i\in I_k}\Bigl(\hat {\bm{D}}_i^{(k)^\top}\hat{\bm{\beta}}^{(-k)}\Bigr)^2+\omega_U\sum_{k=1}^{K}\sum_{i\in J_k}\Bigl(\hat {\bm{D}}_i^{(k)^\top}\hat{\bm{\beta}}^{(-k)}\Bigr)^2+\frac{2}{n}\sum_{k=1}^{K}\sum_{i\in I_k}\hat {\bm{D}}_i^{(k)^\top}\hat{\bm{\beta}}^{(-k)}\hat{\epsilon}_i^{(k)}.
\end{align*}

To guarantee that the re-weighting procedure does not introduce additional bias, the weights must satisfy the constraint $n\omega_L + m\omega_U = 1$. This ensures that the term $\omega_L\sum_{k=1}^{K}\sum_{i\in I_k}(\hat {\bm{D}}_i^{(k)^\top}\hat{\bm{\beta}}^{(-k)})^2+\omega_U\sum_{k=1}^{K}\sum_{i\in J_k}(\hat {\bm{D}}_i^{(k)^\top}\hat{\bm{\beta}}^{(-k)})^2$ remains a valid plug-in estimator of the ETH. Notably, one of the weights may even take a negative value without affecting the bias, provided that the constraint is satisfied. For any given pair $(\omega_L, \omega_U)$, and under the assumptions of Theorem \ref{t2}, the asymptotic variance of the re-weighted estimator $\hat\theta_{\mathrm{RW}}$ is 
\begin{align}
\sigma_{\mathrm{RW}}^2(\omega_L,\omega_U):=&A+n(n\omega_L^2+m\omega_U^2)B+2n\omega_LC\nonumber\\
=&n^2B(1+n/m)\omega_L^2+2n(C-nB/m)\omega_L+A+nB/m,\label{def:sigma_RW}
\end{align}
where we have substituted $\omega_U = (1-n\omega_L)/m$ from the required constraint. We then determine the population-level optimal weights $(\omega_L^*, \omega_U^*)$, which ensure efficiency no worse than both the supervised estimator and the direct semi-supervised estimator in Algorithm \ref{alg:theta_es}, as these two estimators can be regarded as special cases within the broader class of re-weighted estimators. Minimizing the quadratic form in \eqref{def:sigma_RW} yields the optimal weights:
\begin{align}
\omega_L^*=\frac{nB-mC}{nNB}\;\;\text{and}\;\;\omega_U^*=\frac{B+C}{NB}.\label{def:omega}
\end{align}
The formulas in \eqref{def:omega} reveal that the optimal weighting strategy is determined by the sign of $C=\mathrm{Cov}[2\epsilon D^\top\bm{\beta}^*,(D^\top\bm{\beta}^*)^2]$, which quantifies the covariance between the debiasing component and the plug-in component. When $C=0$ (for example, when the sub-CATE model is correctly specified), the optimal strategy is to assign equal weights of $1/N$ to both labeled and unlabeled samples, coinciding with the direct approach in Algorithm \ref{alg:theta_es}. When $C<0$, we have $\omega_L^*>1/N$ and $\omega_U^*<1/N$, indicating that more weight should be placed on the labeled samples, while the unlabeled samples should be down-weighted. Conversely, when $C>0$, the unlabeled samples should receive higher weights.

\begin{algorithm}[t!]
 \centering
 \caption{The optimally weighted semi-supervised ETH estimator}
 \label{alg:OW}
 \renewcommand{\baselinestretch}{1.1}\selectfont
 \footnotesize
 \begin{algorithmic}[1]
 \Require Labeled data \(\mathcal{L}=(\bm{X}_i,A_i,Y_i)_{i=1}^n\), unlabeled data \(\mathcal{U}=(\bm{X}_i,A_i)_{i=n+1}^N\), number of folds \(K \ge 3\).
 \State Partition $\mathcal{L}$ and $\mathcal{U}$ into $K$ disjoint folds of equal size, indexed by $\{I_1,\dots,I_K\}$ and $\{J_1,\dots,J_K\}$.
 \For{\(k \in \{1,\dots,K\}\)}
 \For{\(k' \in \{1,\dots,K\} \setminus \{k\}\)}
 \State Let $I_{-k,-k'}=\{1,\dots,n\}\setminus (I_k\cup I_{k'})$ and $J_{-k,-k'}=\{n+1,\dots,N\}\setminus (J_k\cup J_{k'})$.
 \State For each $a\in\{0,1\}$, obtain $\hat\mu_a^{(-k,-k')}(\bm{x})=\bm{x}^\top\hat\alpha_a^{(-k,-k')}$ through solving \eqref{def:alphahat}.
 \State Obtain $\hat\pi^{(-k,-k')}(\bm{x})=\phi(\bm{x}^\top\hat{\bm{\gamma}}^{(-k,-k')})$ through solving \eqref{def:gamma}.
 \State Compute imputed outcomes \(\hat\varphi^{(-k,-k')}(\cdot)\) as in \eqref{def:phihat}, using \(\hat\mu_a^{(-k,-k')}(\cdot)\) and \(\hat\pi^{(-k,-k')}(\cdot)\).
 \EndFor
 \State Obtain $\hat{\bm{\beta}}^{(-k)}$ through solving \eqref{def:betahat}.
 \State Let $\hat\varphi^{(-k)}(\cdot)=(K-1)^{-1}\sum_{k'\neq k}\hat\varphi^{(-k,-k')}(\cdot)$, $G_k=I_k\cup J_k$, and $\hat{\bm{D}}_i^{(k)}=\bm{W}_i-|G_k|^{-1}\sum_{j\in G_k}\bm{W}_j$.
 \EndFor
 \State Compute $\hat\tau_\mathrm{para}$ as in \eqref{def:tau_para}.
 \State For each $k\in\{1,\dots,K\}$, let $\hat{\epsilon}_i^{(k)}=\hat\varphi^{(-k)}(Z_i)-\hat\tau_\mathrm{para}-\hat {\bm{D}}_i^{(k)^\top}\hat{\bm{\beta}}^{(-k)}$ and compute $(\hat\omega_L^{(k)},\hat\omega_U^{(k)})$ as in \eqref{def:omegahat}.
 \State \Return The optimally weighted semi-supervised ETH estimator \eqref{def:OW}.
 \end{algorithmic}
\end{algorithm}

Finally, we propose data-dependent procedures to estimate the population-level optimal weights and to construct the corresponding optimally weighted semi-supervised estimator for the ETH. Specifically, we build cross-fitted plug-in estimators for the quantities $(A,B,C)$. For each $k \leq K$, define $\hat Q^{(k)} := |G_k|^{-1}\sum_{i\in G_k}(\hat {\bm{D}}_i^{(k)^\top}\hat{\bm{\beta}}^{(-k)})^2$ and construct
\begin{align*}
\hat{A}^{(k)} &:= |I_k|^{-1}\sum_{i\in I_k}\Bigl(2\hat{\epsilon}_i^{(k)}\hat{\bm{\beta}}^{(-k)^\top }\hat{\bm{D}}_i^{(k)}\Bigr)^2,\\
\hat{B}^{(k)} &:= |G_k|^{-1}\sum_{i\in G_k}\Bigl(\hat{\bm{\beta}}^{(-k)^\top }\hat{\bm{D}}_i^{(k)}\Bigr)^4
 -\Bigl(\hat Q^{(k)}\Bigr)^2,\\
\hat{C}^{(k)} &:= |I_k|^{-1}\sum_{i\in I_k}\Bigl(2\hat{\epsilon}_i^{(k)}\hat{\bm{\beta}}^{(-k)^\top }\hat{\bm{D}}_i^{(k)}\Bigr)
 \Bigl[\Bigl(\hat{\bm{\beta}}^{(-k)^\top }\hat{\bm{D}}_i^{(k)}\Bigr)^2 - \hat Q^{(k)}\Bigr].
\end{align*}
Based on these, we define the weights
\begin{align}\label{def:omegahat}
\hat\omega_L^{(k)}:=\frac{n\hat{B}^{(k)}-m\hat{C}^{(k)}}{nN\hat{B}^{(k)}}\;\;\text{and}\;\;\hat\omega_U^{(k)}=\frac{\hat{B}^{(k)}+\hat{C}^{(k)}}{N\hat{B}^{(k)}}.
\end{align}
The optimally weighted semi-supervised ETH estimator and the corresponding asymptotic variance estimator are then given by
\begin{align}
\hat\theta_{\mathrm{OW}}&=\hat\omega_L^{(k)}\sum_{k=1}^{K}\sum_{i\in I_k}\Bigl(\hat {\bm{D}}_i^{(k)^\top}\hat{\bm{\beta}}^{(-k)}\Bigr)^2+\hat\omega_U^{(k)}\sum_{k=1}^{K}\sum_{i\in J_k}\Bigl(\hat {\bm{D}}_i^{(k)^\top}\hat{\bm{\beta}}^{(-k)}\Bigr)^2+\frac{2}{n}\sum_{k=1}^{K}\sum_{i\in I_k}\hat {\bm{D}}_i^{(k)^\top}\hat{\bm{\beta}}^{(-k)}\hat{\epsilon}_i^{(k)},\label{def:OW}\\
\hat\sigma_{\mathrm{OW}}^2&=K^{-1}\sum_{k=1}^K\left[\hat{A}^{(k)} 
 + \frac{n\hat{B}^{(k)}}{N} +\frac{2n\hat{C}^{(k)} }{N} - \frac{m(\hat{C}^{(k)})^2}{N\hat{B}^{(k)}}\right].\nonumber
\end{align}

The complete procedure is outlined in Algorithm \ref{alg:OW}. The theorem below establishes the asymptotic properties of the optimally weighted estimator.

\begin{theorem}\label{t3}
Let Assumptions~\ref{a1}, \ref{a5}, \ref{a6}, and \ref{a7} hold. Moreover, let $s_{\alpha}\log(d)=o(n)$, $s_{\alpha}s_{\gamma}\log^2(d) =o(N)$, $s_{\beta}\log(p)= o(\sqrt{n})$, and $\pi^*(\cdot)=\pi(\cdot)$. Assume either (a) $\mu_a^*(\cdot)=\mu_a(\cdot)$ for each $a\in\{0,1\}$, or (b) $ns_\gamma\log(d)=o(N)$. Define the quantities $(A,B,C)$ as in Theorem \ref{t2} and suppose that $B>c$ and
$$\sigma_{\mathrm{OW}}^2:=\sigma_{\mathrm{RW}}^2(\omega_L^*,\omega_U^*)=A+\frac{nB}{N}+\frac{2nC}{N}-\frac{mC^2}{NB}>c,$$
with some constant $c>0$. Then, as $N,d,p\to\infty$,
\[
 \frac{\sqrt{n}\,(\hat \theta_{\mathrm{OW}} - \theta_{\mathrm{ETH}})}{\sigma_{\mathrm{OW}}}
 \;\xrightarrow{d}\; N(0,1)\quad\text{and}\quad\hat\sigma_{\mathrm{OW}}^2=\sigma_{\mathrm{OW}}^2+o_p(1).
\]
\end{theorem}

\begin{remark}[The efficiency improvement]
The usage of additional unlabeled samples in variance estimation has been also discussed by \cite{1,2,kim2024semi} for the simpler non-causal settings. \cite{2} only consider cases with correct model specification, where semi-supervised methods are guaranteed to attain an efficiency gain over supervised methods. On the other hand, both \cite{1,kim2024semi} further study scenarios where model misspecification occurs. Under such cases, semi-supervised methods may be less efficient than the supervised methods.

The existing semi-supervised approaches assign equal weights of $1/N$ to all labeled and unlabeled samples. In contrast, the newly proposed method in Algorithm \ref{alg:OW} assigns different weights to the two groups of samples, thereby improving efficiency and ensuring safe use of the unlabeled data. Theorem \ref{t3} shows that the optimally weighted semi-supervised ETH estimator based on the estimated data-dependent weights $(\hat\omega_L,\hat\omega_U)$ achieves the same asymptotic variance as the re-weighted estimator $\hat\theta_{\mathrm{RW}}$ with population-level optimal weights $(\omega_L^*,\omega_U^*)$. Thus, estimation errors in the weights do not affect the asymptotic behavior of the resulting estimator.

Comparing with the asymptotic variance of the direct semi-supervised approach, $\sigma_\mathrm{para}^2=A + (n/N)B + 2(n/N)C$, we obtain $\sigma_{\mathrm{OW}}^2=\sigma_\mathrm{para}^2-mC^2/(NB)\leq\sigma_\mathrm{para}^2$. Thus, the optimal weights achieve strictly improved efficiency over the direct semi-supervised approach unless (a) $C=0$, or (b) $m=o(N)$, which corresponds to the case where the number of unlabeled samples is negligible compared with the number of labeled samples. Similarly, comparing with the asymptotic variance of the supervised approach, $\sigma_\mathrm{sup}^2=A + B + 2C$, we have $\sigma_{\mathrm{OW}}^2=\sigma_\mathrm{sup}^2-m(B+C)^2/(NB)\leq\sigma_\mathrm{sup}^2$. Hence, the optimal weights also yield strictly improved efficiency over the supervised approach unless (a) $B+C=0$, or (b) $m=o(N)$.

To summarize, the proposed optimally weighted approach achieves efficiency no worse than both the direct semi-supervised method and the supervised method, thereby ensuring safe use of the unlabeled samples. Our approach also differs from the safe semi-supervised estimator of \cite{deng2024optimal}, which targets the linear slope in high-dimensional linear regression. Their method relies on a refitting step to guarantee a \emph{convergence rate} no slower than that of supervised methods. In contrast, our method only requires an additional weight estimation step and attains an \emph{asymptotic variance} no worse than that of supervised methods.
\end{remark}

\subsection{A semi-parametric extension}

\begin{algorithm}[b!]
 \centering
 \caption{The semi-parametric optimally weighted ETH estimator}
 \label{alg:SPOW}
 \renewcommand{\baselinestretch}{1.1}\selectfont
 \footnotesize
 \begin{algorithmic}[1]
 \Require Labeled data \(\mathcal{L}=(\bm{X}_i,A_i,Y_i)_{i=1}^n\), unlabeled data \(\mathcal{U}=(\bm{X}_i,A_i)_{i=n+1}^N\), number of folds \(K \ge 3\).
 \State Partition $\mathcal{L}$ and $\mathcal{U}$ into $K$ disjoint folds of equal size, indexed by $\{I_1,\dots,I_K\}$ and $\{J_1,\dots,J_K\}$.
 \For{\(k \in \{1,\dots,K\}\)}
 \For{\(k' \in \{1,\dots,K\} \setminus \{k\}\)}
 \State Let $I_{-k,-k'}=\{1,\dots,n\}\setminus (I_k\cup I_{k'})$ and $J_{-k,-k'}=\{n+1,\dots,N\}\setminus (J_k\cup J_{k'})$.
 \State Obtain \(\hat\mu_a^{(-k,-k')}(\cdot)\) using samples in $I_{-k,-k'}$ for each $a\in\{0,1\}$.
 \State Obtain \(\hat\pi^{(-k,-k')}(\cdot)\) using samples in $G_{-k,-k'}:=I_{-k,-k'}\cup J_{-k,-k'}$.
 \State Compute imputed outcomes \(\hat\varphi^{(-k,-k')}(\cdot)\) as in \eqref{def:phihat}, using \(\hat\mu_a^{(-k,-k')}(\cdot)\) and \(\hat\pi^{(-k,-k')}(\cdot)\).
 \EndFor
 \State Obtain $\hat{\bm{\beta}}^{(-k)}$ through solving \eqref{def:betahat}.
 \State Let $\hat\varphi^{(-k)}(\cdot)=(K-1)^{-1}\sum_{k'\neq k}\hat\varphi^{(-k,-k')}(\cdot)$, $G_k=I_k\cup J_k$, and $\hat{\bm{D}}_i^{(k)}=\bm{W}_i-|G_k|^{-1}\sum_{j\in G_k}\bm{W}_j$.
 \EndFor
 \State Compute 
\begin{align*}
\hat\tau_\mathrm{SP}=\frac{1}{N}\sum_{k=1}^K\sum_{i\in G_k}\bm{W}_i^\top\hat{\bm{\beta}}^{(-k)}+\frac{1}{n}\sum_{k=1}^K\sum_{i\in I_k}\Bigl[\hat\varphi^{(-k)}(Z_i)-\bm{W}_i^\top\hat{\bm{\beta}}^{(-k)}\Bigr].
\end{align*}
 \State For each $k\in\{1,\dots,K\}$, let $\hat{\epsilon}_i^{(k)}=\hat\varphi^{(-k)}(Z_i)-\hat\tau_\mathrm{SP}-\hat {\bm{D}}_i^{(k)^\top}\hat{\bm{\beta}}^{(-k)}$ and compute $(\hat\omega_L^{(k)},\hat\omega_U^{(k)})$ as in \eqref{def:omegahat}.
 \State \Return The semi-parametric optimally weighted ETH estimator
 \begin{align*}
\hat\theta_{\mathrm{SPOW}}=\hat\omega_L^{(k)}\sum_{k=1}^{K}\sum_{i\in I_k}\Bigl(\hat {\bm{D}}_i^{(k)^\top}\hat{\bm{\beta}}^{(-k)}\Bigr)^2+\hat\omega_U^{(k)}\sum_{k=1}^{K}\sum_{i\in J_k}\Bigl(\hat {\bm{D}}_i^{(k)^\top}\hat{\bm{\beta}}^{(-k)}\Bigr)^2+\frac{2}{n}\sum_{k=1}^{K}\sum_{i\in I_k}\hat {\bm{D}}_i^{(k)^\top}\hat{\bm{\beta}}^{(-k)}\hat{\epsilon}_i^{(k)}.
\end{align*}
 \end{algorithmic}
\end{algorithm}

In the following, we develop a semi-parametric extension that continues to focus on simplified linear CATE models, with the goal of enhancing interpretability and facilitating implementation in follow-up studies, such as the deployment of individualized treatment regimes based on CATE estimates. In contrast to the approaches in Sections \ref{sec3.1}-\ref{sec3.2}, we employ more flexible non-parametric or machine learning estimators for the nuisance models, thereby improving estimation accuracy and reducing bias due to model misspecification. This strategy is particularly appealing because the CATE function, defined as the difference between two outcome regression functions, often admits a simpler structure (or needs to be modeled in a simpler form due to practical constraints), whereas the outcome regression functions themselves and the propensity score model may exhibit complex nonlinear behavior \citep{fan2022estimation,Kennedy2020}. The detailed construction of the semi-parametric optimally weighted estimator of the ETH for linear CATE models is given in Algorithm \ref{alg:SPOW}.

The following theorem characterizes the asymptotic behavior of the semi-parametric optimally weighted ETH estimator.

\begin{theorem}\label{t5}
Let Assumptions~\ref{a1}, \ref{a2}, \ref{a3}, \ref{a6}, and \ref{a7} hold. Moreover, let $s_{\beta}\log(p)=o(\sqrt{n})$ and $\sigma_{\mathrm{OW}}^2>c$ with some constant $c>0$. Moreover, let $\mathbb{E}_{\bm{X}}|\hat\pi^{(-k,-k')}(\bm{X}) - \pi(\bm{X})|^2|\hat\mu_a^{(-k,-k')}({\bm{X}}) - \mu_a({\bm{X}})|^2 = o_p(n^{-1})$ for each $a\in\{0,1\}$. Then, as $N,d,p\to\infty$,
\[
\frac{\sqrt{n}\,(\hat {\theta}_{\mathrm{SPOW}}-\theta_{\mathrm{ETH}})}{\sigma_{\mathrm{OW}}} \xrightarrow{d} N(0,1).
\]
\end{theorem}

The estimator $\hat{\theta}_{\mathrm{SPOW}}$ combines the advantages of the optimal weighting strategy with fully non-parametric nuisance estimation. This guarantees safe use of the unlabeled samples and ensures efficiency no worse than the supervised approach. At the same time, the semi-parametric framework provides more robust inference by relying on non-parametric outcome regression and propensity score estimates, thereby avoiding assumptions on parametric forms of the nuisance models.

\section{Numerical studies}\label{sec4}

\subsection{Simulation results}\label{sec:sim}

We evaluate the performance of the proposed methods under three data-generating scenarios. Model 1 assumes that all nuisance functions and the CATE function are correctly specified by the (generalized) linear models. Model 2 considers misspecified nuisance functions while the CATE function remains linear. Model 3 assumes that both the nuisance functions and the CATE function are misspecified by the parametric models. In all models, we generate \(n\) labeled samples \((\bm{X}_i, A_i, Y_i)_{i=1}^n\) and \(m = N-n\) unlabeled samples \((\bm{X}_i, A_i)_{i=n+1}^{N}\), where \(\bm{X}_i \overset{\text{iid}}{\sim} \mathcal{N}_d(0, I_d)\). The remaining details of the data construction are provided below.

\textbf{Model 1.} Set \(d = 200\). Generate $A_i\mid \bm{X}_i\sim\mathrm{Bernoulli}[\phi(0.3X_{i1} + 0.5X_{i4})]$ and $Y_i=\sum_{j=1}^{20} X_{ij}/\sqrt{20}+A_i(X_{i1}+X_{i2}+X_{i3})+\mathcal{N}(0,0.1^2)$, where $\phi(\cdot)$ is the logistic function.

\textbf{Model 2.} Set \(d = 200\). Generate $A_i\mid \bm{X}_i\sim\mathrm{Bernoulli}[\phi(0.2X_{i3}^2)]$ and $Y_i=0.5X_{i3}^2+A_i(X_{i1}+X_{i2}+X_{i3})+\mathcal{N}(0,0.1^2)$.

\textbf{Model 3.} Set \(d = 10\). Generate $A_i\mid \bm{X}_i\sim\mathrm{Bernoulli}[\phi(0.2X_{i3}^2)]$ and $Y_i=0.5X_{i3}^2+A_i(X_{i2}+0.5X_{i3}^2)+\mathcal{N}(0,0.1^2)$.

\begin{table}[!b]
\centering
\caption{Simulation results for TTH and ETH estimation under Model 1. Bias: average estimation bias; Emp SE: empirical standard deviation; ASE: average estimated standard error; RMSE: root mean squared error; AC: empirical coverage of the nominal 95\% confidence interval; Length: average length of the confidence interval.}
\label{tab:simulation_results}
\renewcommand{\arraystretch}{0.7}
\begin{tabular}{ccccccccc}
\toprule
Method & $n$ & $m$ & Bias & Emp SE & ASE & RMSE & AC & Length\\
\midrule
\multirow{2}{*}{$\hat{\theta}_{\mathrm{TEVIM}}$}
 & 1000 & 0 & 0.0067 & 0.1499 & 0.1365 & 0.1495 & 0.940 & 0.5351\\
 & 2000 & 0 & -0.0058 & 0.0942 & 0.0969 & 0.0940 & 0.930 & 0.3798\\
\midrule
\multirow{4}{*}{$\hat{\theta}_{\mathrm{TTH}}$}
 & 1000 & 0 & -0.0101 & 0.1495 & 0.1363 & 0.1495 & 0.930 & 0.5343\\
 & 2000 & 0 & -0.0079 & 0.0940 & 0.0968 & 0.0941 & 0.935 & 0.3795\\
 & 1000 & 5000 & -0.0039 & 0.0634 & 0.0620 & 0.0633 & 0.930 & 0.2430\\
 & 2000 & 5000 & -0.0031 & 0.0547 & 0.0539 & 0.0556 & 0.920 &0.2113\\
\midrule
\multirow{4}{*}{$\hat{\theta}_{\mathrm{OW}}$}
 & 1000 & 0 & -0.0016 & 0.1456 & 0.1358 & 0.1453 & 0.950 &0.5323\\
 & 2000 & 0 & 0.0063 & 0.0972 & 0.0968 & 0.0972 & 0.960 &0.3795\\
 & 1000 & 5000 & -0.0081 & 0.0655 & 0.0614 & 0.0659 & 0.935 &0.2407\\
 & 2000 & 5000 & -0.0023 & 0.0516 & 0.0539 & 0.0515 & 0.965 &0.2113\\
\midrule
\multirow{4}{*}{$\hat{\theta}_{\mathrm{SPOW}}$}
 & 1000 & 0 & -0.0100 & 0.1495 & 0.1363 & 0.1495 & 0.930 &0.5343\\
 & 2000 & 0 & -0.0079 & 0.0940 & 0.0968 & 0.0941 & 0.935 &0.3795\\
 & 1000 & 5000 & -0.0063 & 0.0640 & 0.0617 & 0.0641 & 0.930 &0.2419\\
 & 2000 & 5000 & -0.0114 & 0.0546 & 0.0539 & 0.0556 & 0.915 &0.2113\\
\bottomrule
\end{tabular}
\end{table}

We evaluate estimation and inference for two target parameters: (a) the total treatment heterogeneity (TTH) and (b) the explained treatment heterogeneity (ETH) of the linear working model that includes all covariates (i.e., $\bm{W}=\bm{X}$). For TTH, we compare the proposed semi-supervised estimator $\hat\theta_{\mathrm{TTH}}$ in Algorithm \ref{alg:Qhat_SL} with the supervised estimator $\hat{\theta}_{\mathrm{TEVIM}}$ introduced by \cite{hines2022variable}. For ETH, we examine the performance of the optimally weighted estimators $\hat{\theta}_{\mathrm{OW}}$ and $\hat{\theta}_{\mathrm{SPOW}}$, introduced in Algorithms \ref{alg:OW} and \ref{alg:SPOW}, respectively. For comparison, we also implement the supervised counterparts of the proposed estimators. All non-parametric estimations are obtained using the SuperLearner algorithm \citep{van2007super}, with base learners including generalized additive models, regularized generalized linear models, XGBoost, random forests, and neural networks.

The simulations are repeated 200 times, and the results are summarized in Tables \ref{tab:simulation_results}-\ref{tab:simulation_results_4}. In Models 1 and 2, the true CATE function is linear, so the TTH and ETH coincide. Hence, we present the combined comparison results in Tables \ref{tab:simulation_results}-\ref{tab:simulation_results_2}. In contrast, since the true CATE is non-linear in Model 3, we report the TTH and ETH estimation results separately in Tables \ref{tab:simulation_results_3} and \ref{tab:simulation_results_4}, respectively. 

\begin{table}[!t]
\centering
\caption{Simulation results for TTH and ETH estimation under Model 2. The rest of the caption details remain the same as those in Table \ref{tab:simulation_results}.}
\label{tab:simulation_results_2}
\renewcommand{\arraystretch}{0.7}
\begin{tabular}{ccccccccc}
\toprule
Method & $n$ & $m$ & Bias & Emp SE & ASE & RMSE & AC & Length\\
\midrule
\multirow{2}{*}{$\hat{\theta}_{\mathrm{TEVIM}}$}
 & 1000 & 0 & 0.0168 & 0.1941 & 0.1832 & 0.1944 & 0.930 &0.7181\\
 & 4000 & 0 & 0.0095 & 0.0776 & 0.0740 & 0.0780 & 0.925 &0.2901\\
\midrule
\multirow{4}{*}{$\hat{\theta}_{\mathrm{TTH}}$}
 & 1000 & 0 & 0.0056 & 0.1933 & 0.1837 & 0.1929 & 0.925 &0.7201\\
 & 4000 & 0 & 0.0068 & 0.0777 & 0.0740 & 0.0778 & 0.925 &0.2901\\
 & 1000 & 8000 & 0.0004 & 0.1373 & 0.1210 & 0.1369 & 0.900 &0.4743\\
 & 4000 & 8000 & 0.0033 & 0.0514 & 0.0486 & 0.0514 & 0.935 &0.1905\\
\midrule
\multirow{4}{*}{$\hat{\theta}_{\mathrm{OW}}$}
 & 1000 & 0 & -0.0428 & 0.2682 & 0.2459 & 0.2709 & 0.880 &0.9639\\
 & 4000 & 0 & -0.0001 & 0.1216 & 0.1273 & 0.1213 & 0.955 &0.4990\\
 & 1000 & 8000 & -0.0180 & 0.2348 & 0.2105 & 0.2349 & 0.910 &0.8252\\
 & 4000 & 8000 & -0.0034 & 0.1037 & 0.1145 & 0.1035 & 0.980 &0.4488\\
\midrule
\multirow{4}{*}{$\hat{\theta}_{\mathrm{SPOW}}$}
 & 1000 & 0 & 0.0052 & 0.1853 & 0.1838 & 0.1850 & 0.950 &0.7205\\
 & 4000 & 0 & -0.0011 & 0.0774 & 0.0737 & 0.0772 & 0.925 &0.2889\\
 & 1000 & 8000 & -0.0003 & 0.1240 & 0.1179 & 0.1237 & 0.940 &0.4622\\
 & 4000 & 8000 & -0.0012 & 0.0497 & 0.0487 & 0.0496 & 0.945 &0.1909\\
\bottomrule
\end{tabular}
\end{table}

As shown in Tables \ref{tab:simulation_results}-\ref{tab:simulation_results_3}, the supervised version of the proposed TTH estimator, $\hat\theta_{\mathrm{TTH}}$, performs comparably to the existing supervised estimator $\hat{\theta}_{\mathrm{TEVIM}}$. In contrast, by leveraging additional unlabeled samples, the semi-supervised TTH estimator consistently achieves smaller root mean squared error (RMSE) than the supervised alternatives under the same labeled sample size. In Model 1 (Table \ref{tab:simulation_results}), where all models are correctly specified, the proposed semi-supervised estimators $(\hat\theta_{\mathrm{TTH}}, \hat\theta_{\mathrm{OW}}, \hat\theta_{\mathrm{SPOW}})$ perform similarly, with the parametric estimator $\hat\theta_{\mathrm{OW}}$ attaining the smallest RMSE at $(n,m)=(1000,5000)$ due to the use of the simplest and correctly specified nuisance models. In contrast, in Model 2, where the nuisance models deviate from (generalized) linear forms, $\hat\theta_{\mathrm{OW}}$ exhibits sub-optimal performance; see Table \ref{tab:simulation_results_2}. In this setting, $\hat\theta_{\mathrm{SPOW}}$ achieves the best performance by incorporating non-parametric nuisance estimation while still relying on the correctly specified linear CATE model. For the ETH estimation under Model 3, as shown in Table \ref{tab:simulation_results_4}, $\hat\theta_{\mathrm{SPOW}}$ clearly outperforms $\hat\theta_{\mathrm{OW}}$ in terms of both estimation accuracy and inference validity due to the misspecification of the parametric nuisance models. In addition, the benefit of using unlabeled samples is evident, as the semi-supervised estimators achieve smaller RMSE than their supervised counterparts.

\begin{table}[!t]
\centering
\caption{Simulation results for TTH estimation under Model 3. The rest of the caption details remain the same as those in Table \ref{tab:simulation_results}.}
\label{tab:simulation_results_3}
\renewcommand{\arraystretch}{0.7}
\begin{tabular}{ccccccccc}
\toprule
Method & $n$ & $m$ & Bias & Emp SE & ASE & RMSE & AC & Length\\
\midrule
\multirow{2}{*}{$\hat{\theta}_{\mathrm{TEVIM}}$}
 & 1000 & 0 & 0.1073 & 0.1191 & 0.1207 & 0.1601 & 0.925 &0.4731\\
 & 4000 & 0 & 0.0699 & 0.0536 & 0.0580 & 0.0880 & 0.880 &0.2274\\
\midrule
\multirow{4}{*}{$\hat{\theta}_{\mathrm{TTH}}$}
 & 1000 & 0 & 0.1011 & 0.1185 & 0.1206 & 0.1555 & 0.940 &0.4728\\
 & 4000 & 0 & 0.0683 & 0.0537 & 0.0580 & 0.0868 & 0.869 &0.2274\\
 & 1000 & 8000 & 0.0340 & 0.0949 & 0.0723 & 0.1006 & 0.880 &0.2834\\
 & 4000 & 8000 & 0.0321 & 0.0439 & 0.0411 & 0.0567 & 0.890 &0.1611\\
\bottomrule
\end{tabular}
\end{table}

\begin{table}[!t]
\centering
\caption{Simulation results for ETH estimation under Model 3. The rest of the caption details remain the same as those in Table \ref{tab:simulation_results}.}
\label{tab:simulation_results_4}
\renewcommand{\arraystretch}{0.7}
\begin{tabular}{ccccccccc}
\toprule
Method & $n$ & $m$ & Bias & Emp SE & ASE & RMSE & AC & Length\\
\midrule
\multirow{4}{*}{$\hat{\theta}_{\mathrm{OW}}$}
 & 1000 & 0 & -0.0669 & 0.1634 & 0.1417 & 0.1762 & 0.850 &0.5555\\
 & 4000 & 0 & -0.0092 & 0.0828 & 0.0733 & 0.0831 & 0.895 &0.2873\\
 & 1000 & 8000 & -0.0608 & 0.1582 & 0.1345 & 0.1691 & 0.855 &0.5272\\
 & 4000 & 8000 & -0.0145 & 0.0719 & 0.0706 & 0.0731 & 0.930 &0.2768\\
\midrule
\multirow{4}{*}{$\hat{\theta}_{\mathrm{SPOW}}$}
 & 1000 & 0 & -0.0034 & 0.0672 & 0.0680 & 0.0671 & 0.955 &0.2666\\
 & 4000 & 0 & -0.0022 & 0.0316 & 0.0336 & 0.0316 & 0.970 &0.1317\\
 & 1000 & 8000 & -0.0116 & 0.0520 & 0.0514 & 0.0532 & 0.955 &0.2015\\
 & 4000 & 8000 & -0.0020 & 0.0275 & 0.0278 & 0.0275 & 0.950 &0.1090\\
\bottomrule
\end{tabular}
\end{table}

\subsection{Application to AIDS clinical trials}\label{sec:real}

We apply the proposed methods to data from the AIDS Clinical Trials Group Protocol 175 (ACTG175) \citep{hammer1996}, which enrolled 2139 HIV‐infected patients with baseline CD4 T‐cell counts between 200 and 500 mm\textsuperscript{–3}. Patients were randomized to one of four treatment regimens: (i) zidovudine (ZDV) monotherapy, (ii) ZDV + didanosine (ddI), (iii) ZDV + zalcitabine, and (iv) ddI monotherapy, with sample sizes 532, 522, 524, and 561, respectively. For our analysis, we focus on comparing regimen (iv) against regimen (ii), coding \(A = 0\) for ddI monotherapy and \(A = 1\) for ZDV + ddI. The outcome \(Y\) is the CD4 count measured at \(20 \pm 5\) weeks post‐randomization. We adjust for 12 baseline covariates: five continuous (age, weight, Karnofsky score, baseline CD4 count, baseline CD8 count) and seven binary (sex, homosexual activity, gender, symptomatic status, intravenous drug use history, hemophilia, and antiretroviral history). The dataset is publicly available through the \texttt{speff2trial} R package.

In this analysis, we aim to estimate the treatment heterogeneity between ZDV + ddI combination therapy and ddI monotherapy. Estimating the TTH helps determine whether personalized treatment decisions are warranted or if a one-size-fits-all approach suffices. Estimating the ETH, on the other hand, indicates whether simplified personalized treatment decisions based on a linear CATE model are adequate. 

Since we focus only on heterogeneity between two treatment groups, standard approaches typically discard data from the remaining two groups, as in \cite{hines2022variable}. Rather than ignoring this additional information, we leverage it by treating the full dataset as a semi-supervised dataset. Specifically, we treat samples from the target groups as labeled, where outcomes of interest are observed, and samples from the non-target groups as unlabeled, since their outcomes correspond to other potential treatments not currently under study. Based on this construction, we end up with a semi-supervised dataset with labeled size $n=1083$ and unlabeled size $m=1056$. Moreover, since the dataset comes from a clinical trial with fully randomized treatments, the missing completely at random (MCAR) condition holds, and the covariates from labeled and unlabeled samples share the same marginal distribution. This ensures the validity of applying semi-supervised methods. 

\begin{table}[!t]
\centering
\caption{Results for the TTH and ETH estimation in the real-data analysis of the ACTG175 study. Parameter: the parameter being estimated by the corresponding method; Estimate: point estimate of the TTH or ETH; CI: 95\% confidence interval; Length: length of the 95\% confidence interval; p-value: p-value for testing the null hypothesis that the TTH or ETH is zero against the alternative that it is positive.}
\label{tab:empirical_estimates}
\renewcommand{\arraystretch}{0.7}
\begin{tabular}{cccccc}
\toprule
Method & Parameter & Estimate & CI & Length & p-value\\
\midrule
$\hat{\theta}_{\mathrm{TEVIM}}$ & \multirow{2}{*}{TTH} & 881.672 & [-360.19, 2123.53] & 2483.72 & 0.082\\
$\hat{\theta}_{\mathrm{TTH}}$ & & 979.404 & [-193.90, 2152.71] & 2346.61& 0.051\\
\midrule
$\hat{\theta}_{\mathrm{OW}}$ & \multirow{2}{*}{ETH} & 378.125 & [-754.66, 1510.91] & 2265.57 & 0.257\\
$\hat{\theta}_{\mathrm{SPOW}}$ & &417.505 & [-742.36, 1577.37] & 2319.73 & 0.238\\
\bottomrule
\end{tabular}
\end{table}

Same as in the simulation studies in Section \ref{sec:sim}, we implement the supervised estimator $\hat{\theta}_{\mathrm{TEVIM}}$ of \cite{hines2022variable} and the proposed semi-supervised estimator $\hat{\theta}_{\mathrm{TTH}}$ for the TTH estimation; for the ETH estimation, we implement the proposed semi-supervised estimators $\hat{\theta}_{\mathrm{OW}}$ and $\hat{\theta}_{\mathrm{SPOW}}$. Table \ref{tab:empirical_estimates} illustrates the estimation and inference results based on the considered methods. 

For the TTH estimation, we observe that the semi-supervised estimator $\hat{\theta}_{\mathrm{TTH}}$ indicates a larger TTH than the supervised estimator $\hat{\theta}_{\mathrm{TEVIM}}$. By incorporating additional unlabeled samples, the semi-supervised method also produces a shorter 95\% confidence interval. Consequently, when testing $H_0:\theta_\mathrm{TTH} = 0$ versus $H_1:\theta_\mathrm{TTH} > 0$, the semi-supervised method yields a smaller p-value than the supervised approach, although the p-value (0.051) remains slightly above the conventional 0.05 significance level. 

Furthermore, our results suggest that approximately 40\% of the heterogeneity can be explained by the linear working model for the CATE, with $\hat{\theta}_{\mathrm{OW}} / \hat{\theta}_{\mathrm{TTH}} = 38.6\%$ and $\hat{\theta}_{\mathrm{SPOW}} / \hat{\theta}_{\mathrm{TTH}} = 42.6\%$. The latter estimate is considered more reliable as it does not rely on parametric forms for the nuisance models. This information can guide the design of practically implementable personalized treatments, balancing effectiveness with feasibility and interpretability.

\section{Discussion}\label{sec5}
This paper investigates the problem of estimating treatment heterogeneity in semi‐supervised settings, focusing on both the \emph{total treatment heterogeneity} (TTH) and the \emph{explained treatment heterogeneity} (ETH) of a simplified working model. We develop semi‐supervised estimators that leverage large unlabeled samples to enhance the robustness and efficiency. Our findings highlight the value of semi‐supervised learning in modern causal inference. By carefully combining large pools of cheap and easily accessible unlabeled data with relatively smaller sized labeled samples that contains important outcome information, researchers can obtain sharper and more robust estimates of treatment heterogeneity. This provides a flexible and efficient toolkit for advancing personalized decision‐making in medicine, economics, and related domains.

Our results highlight the distinction between semi-supervised estimation for the TTH and the ETH. While the TTH estimation can be viewed as an ETH estimation problem with a correctly specified CATE model, under this stronger assumption, direct semi-supervised approaches that assign equal weights to labeled and unlabeled samples are guaranteed to achieve efficiency no worse than supervised methods. However, when a consistent estimate of the true CATE is not feasible and only the ETH for a specific working model can be estimated, direct semi-supervised approaches may suffer efficiency loss compared to supervised counterparts. To ensure safe use of the additional unlabeled samples, a re-weighting (or optimal weighting) strategy is generally required. Although our theoretical analysis focuses on ETH estimation under linear working models, the same framework can be readily extended to non-linear or even non-parametric working models. Furthermore, the proposed optimal weighting strategy can be applied to other semi-supervised inference problems where direct estimation may not guarantee efficiency gains under model misspecification, including, for example, semi-supervised U-statistics studied in \cite{kim2024semi,cannings2022correlation}.

\section*{Fundings}

This work was supported by the National Natural Science Foundation of China (NSFC) under grant 12301390 (Y.Z.) and the Renmin University of China under grant RUC24QSDL062 (Y.A.).

 \bibliography{bibliography.bib}

\bigskip
\newpage
\begin{center}
 {\LARGE\bfseries Supplementary material to}\\[1ex]
 {\LARGE\bfseries \enquote{Semi-supervised inference for treatment heterogeneity}}
\end{center}

\appendix
\renewcommand{\thesection}{\Alph{section}}

\vspace{1em}
This document contains additional proofs of the main theoretical results. All results and
notations are numbered and used, as in the main text unless stated otherwise.

\renewcommand{\theequation}{S\arabic{equation}}
\setcounter{equation}{0} 
\setcounter{theorem}{0}

\section{Auxiliary Lemmas}
\begin{Lemma}\label{l1}
(Selection from Lemma D.1 in \cite{8}). 
Let \(X\) and \(Y\) be random variables. Then
\begin{itemize}
 \item[(A1)] For any scalar \(c\in\mathbb{R}\) and any \(a\ge1\), $ \|cX\|_{\psi_a}=|c|\,\|X\|_{\psi_a}$,
$\|{X}+Y\|_{\psi_a}\le\|X\|_{\psi_a}+\|Y\|_{\psi_a}.$
 \item[(A2)] If \(|X|\le|Y|\) almost surely, then \(\|X\|_{\psi_a}\le\|Y\|_{\psi_a}\). In particular, any bounded \(X\) with \(|X|\le M\) is sub‐Gaussian with $ \|X\|_{\psi_2}\le(\log2)^{-1/2}M.$
 Moreover, for sub‐Gaussian \(X\) and \(Y\), $\|XY\|_{\psi_1}\le\|{X}\|_{\psi_2}\,\|Y\|_{\psi_2}.$
 \item[(A3)] If \(\|{X}\|_{\psi_1}\le\sigma\), then for all integers \(k\ge1\), $\mathbb{E}_{X}|X|^k\le\sigma^k\,k!\le\sigma^k\,k^k.$
 If \(\|{X}\|_{\psi_2}\le\sigma\), then for all \(k\ge1\), $\mathbb{E}_{X}|{X}|^k\le2\sigma^k{{\gamma}}(\frac{k}{2}+1)$, $
{{\Gamma}}({{\alpha}})=\int_{0}^{\infty}X^{{{\alpha}}-1}\exp(-x)\,dx.$
\end{itemize}
\end{Lemma}

\begin{Lemma}\label{l11}
Let the assumptions of Theorem \ref{t4} hold. Then
\begin{align*}
 \frac{1}{\tilde n}\sum_{i\in I_k}\left[\hat\varphi^{(-k)}(Z_i)-\varphi^*(Z_i)\right]
 &= o_p(n^{-1/2}),\\
 \frac{1}{\tilde n}\sum_{i\in I_k}\left[\hat {\tau}^{(-k)}(\bm{X}_i)-\nu^{(-k)}\right]
 \,\left[\hat\varphi^{(-k)}(Z_i)-\varphi^*(Z_i)\right]
 &= o_p(n^{-1/2}),\\
 \hat{\tau}-\tau &= O_p(n^{-1/2}),
\end{align*}
where $\nu^{(-k)}:=\mathbb{E}_{\bm{X}}[\hat{\tau}^{(-k)}(\bm{X})]$, $\tilde{n}:=n/K$, and $\tilde{m}:=(N-n)/K$.
\end{Lemma}

\begin{Lemma}\label{l2}
Let Assumptions~\ref{a1}, \ref{a5}, and \ref{a6} hold. 
Consider any \(t>0\) and $\lambda_{\alpha}= C(t+\sqrt{{\log(d)}/{\lvert I_{-k,-k'}\rvert}})$, where $C>0$ is some constant and \(\lvert I_{-k,-k'}\rvert \ge \max\{\log(d),\,100\kappa_2^2\}\). 
Then, with probability at least
$1 - 2\exp[-{4\lvert I_{-k,-k'}\rvert\,t^2/(1+2t+\sqrt{2t}})]
 - c_1\exp(-c_2\lvert I_{-k,-k'}\rvert),$
for some constants \(c_1,c_2,\kappa_1,\kappa_2>0\),
\begin{align*}
\left\|\hat{\bm{\alpha}}_a^{(-k,-k')}-{\bm{\alpha}}_a^*\right\|_2
&\le 8\,\kappa_1^{-1}\,\lambda_{\alpha}\sqrt{s_{\alpha}},\\
\frac{1}{\lvert I_{-k,-k'}\rvert}
\sum_{i\in I_{-k,-k'}}
\mathbbm{1}_{\{A_i=a\}}
\left[\bm{X}_i^\top (\hat{\bm{\alpha}}_a^{(-k,-k')}-{\bm{\alpha}}_a^*)\right]^2
&\le 32\,\kappa_1^{-1}\,\lambda_{\alpha}^2 s_{\alpha}.
\end{align*}
Moreover, if \(n \gg s_{\alpha}\log d\) and 
\(\lambda_{\alpha}\asymp \sqrt{{\log d}/{n}}\). Then
\[
\left\|\hat{\bm{\alpha}}_a^{(-k,-k')}-{\bm{\alpha}}_a^*\right\|_2
= O_p\left(\sqrt{\frac{s_{\alpha}\log d}{n}}\right).
\]
\end{Lemma}

\begin{Lemma}\label{l3}
Let the assumptions of Theorem \ref{t1} hold. Then
\[
\left\|\hat{\bm{\gamma}}^{(-k,-k')}-{\bm{\gamma}}^*\right\|_2
=O_p\left(\sqrt{\frac{s_{\gamma}\log d}{N}}\right),
\quad
\mathbb{E}_{\bm{X}}\left[\hat\pi^{(-k,-k')}(\bm{X})-\pi^*(\bm{X})\right]^2
=O_p\left(\frac{s_{\gamma}\log d}{N}\right).
\]
\end{Lemma}

\begin{Lemma}\label{l4}
Let the assumptions of Theorem \ref{t1} hold. Define 
$
B:= \{\|\hat{{\bm{\gamma}}}^{(-k,-k')}-{\bm{\gamma}}^*\|_{2}\le 1,\forall k,k'\leq K,k\neq k'\}
$
and consider for any constant \(r>2\). Under the event $B$,
\([\mathbb{E}_{\bm{X}}|\hat{\pi}^{(-k,-k')}(\bm{X})|^{-r}]^{1/r}<C\) and 
\([\mathbb{E}_{\bm{X}}|1-\hat{\pi}^{(-k,-k')}(\bm{X})|^{-r}]^{1/r}<C\) with some constant $C>0$. Moreover,
\begin{align*}
\left\{\mathbb{E}_{\bm{X}}\left|[\hat{\pi}^{(-k,-k')}(\bm{X})]^{-1}-[{\pi^*}(\bm{X})]^{-1}\right|^{r}\right\}^{1/r}
&=O_p\left(\sqrt{\frac{s_{{\bm{\gamma}}}\log(d)}{N}}\right),\\
\left\{\mathbb{E}_{\bm{X}}\left|\left[1-\hat{\pi}^{(-k,-k')}(\bm{X})\right]^{-1}-\left[1-{\pi^*}(\bm{X})\right]^{-1}\right|^{r}\right\}^{1/r}
&=O_p\left(\sqrt{\frac{s_{{\bm{\gamma}}}\log(d)}{N}}\right).
\end{align*}
\end{Lemma}

\begin{Lemma}\label{l5}
Let the assumptions of Theorem \ref{t1} hold. Then 
\(\|\epsilon\|_{\psi_2}\le\sigma_\epsilon\) for some constant \(\sigma_\epsilon>0\), 
and \(\|{\bm{\beta}}^*\|_2=O(1)\).

\end{Lemma}

\begin{Lemma}\label{l6}
Let the assumptions of Theorem \ref{t1} hold. Then
\[
\mathbb{E}_{\bm{X}}\left|\hat{\varphi}^{(-k)}(Z)-\varphi^*(Z)\right|^4 = o_p(1),
\quad
\hat\tau_{\mathrm{para}}-\tau = o_p(1).
\]
\end{Lemma}

\begin{Lemma}\label{l7}
Let the assumptions of Theorem \ref{t1} hold. Then, for any constant $l>0$,
\begin{align*}
\frac{1}{\tilde n}\sum_{i\in I_k}\hat{{\bm{\beta}}}^{(-k)^\top }{\bm{D}_i}
\left[\hat{\varphi}^{(-k)}(Z_i)-\varphi^*(Z_i)\right]
&= O_p(R_n),\\
\frac{1}{\tilde n}\sum_{i\in I_k}\left({\bm{\beta}}^{*^\top}{\bm{D}_i}\right)^l
\left[\hat{\varphi}^{(-k)}(Z_i)-\varphi^*(Z_i)\right]
&= O_p(R_n),
\end{align*}
where 
$$R_n
:= \sqrt{\frac{s_{{\bm{\alpha}}}s_{{\bm{\gamma}}}\,\log^2(d)}{nN}}
 + \sqrt{\frac{s_{{\bm{\gamma}}}\,\log(d)}{N}}\,
 \left(\mathbbm{1}_{\mu_1^*(\cdot)\neq \mu_1(\cdot)}+\mathbbm{1}_{\mu_0^*(\cdot)\neq \mu_0(\cdot)}\right)
 + \sqrt{\frac{s_{{\bm{\alpha}}}\,\log(d)}{n}}\,
 \mathbbm{1}_{\pi^*(\cdot)\neq\pi(\cdot)}.$$
\end{Lemma}

\begin{Lemma}\label{l8}
Let the assumptions of Theorem \ref{t1} hold. For any nonnegative constants \(b,v,l,q,r\) with \(q+r>0\) and any \(a\in\{0,1,2\}\),
\begin{align*}
 &\frac{1}{\tilde n}
\sum_{i\in I_k}
\hat{\epsilon_i}^{(k)^a}\left(\hat{{\bm{\beta}}}^{(-k)^\top }\hat {\bm{D}}_i^{(k)}\right)^b\left({\bm{\beta}}^{*^\top}{\bm{D}_i}\right)^v
\left(\hat{{\bm{\beta}}}^{(-k)^\top }{\bm{D}_i}\right)^l \\
&\qquad\cdot\left[\hat{{\bm{\beta}}}^{(-k)^\top }\left(\mathbb{E}(\bm{W})-|G_k|^{-1}\sum_{i\in G_k}{\bm{W}_i}\right)\right]^q
\left[\left(\hat{{\bm{\beta}}}^{(-k)}-{\bm{\beta}}^*\right)^\top {\bm{D}_i}\right]^r
= o_p(1).
\end{align*}
\end{Lemma}

\begin{Lemma}\label{l9}
Let the assumptions of Theorem \ref{t3} hold. Then
\begin{align*}
 \hat{w}^{(k)}_U = w_U^* + o_p\left(\frac{1}{N}\right),\quad\hat{w}^{(k)}_L = w_L^* + o_p\left(\frac{m}{nN}\right).
\end{align*}
\end{Lemma}

\begin{Lemma}\label{l10}
Let the assumptions of Theorem \ref{t3} hold. Then
\begin{align*}
 \hat{\sigma}_{\mathrm{OW}}^2 = \sigma_{\mathrm{OW}}^2 + o_p(1).%,\quad \hat{\sigma}_{\mathrm{para}}^2 = \sigma_{\mathrm{para}} + o_p(1).
\end{align*}
\end{Lemma}

\section{Proof of Lemmas}
\begin{proof}[Proof of Lemma \ref{l11}]
For any $k\leq K$, we have
\begin{align}\label{152}
 \notag&\frac{1}{\tilde{n}}\sum_{i \in I_k}\left[\hat{\tau}^{(-k)}(\bm{X}_i)-\nu^{(-k)}\right] \left[\hat{\varphi}^{(-k)}(Z_i
 )-\varphi^*(Z_i)\right]\\
 \notag &\quad=
 \frac{1}{K-1}\sum_{k'\neq k}\frac{1}{\tilde{n}}\sum_{i \in I_k}\left\{\left(\hat{\tau}^{(-k)}(\bm{X}_i)-\nu^{(-k)}\right)\left[ \hat\varphi^{(-k,-k')}(Z_i)-\varphi^*(Z_i)\right]\right\}.
\end{align}
For any $k'\neq k$, consider the representation
\begin{align*}
 \frac{1}{\tilde{n}}\sum_{i \in I_k}\left[\hat{\tau}^{(-k)}(\bm{X}_i)-\nu^{(-k)}\right] \left[\hat{\varphi}^{(-k,-k')}(Z_i
 )-\varphi^*(Z_i)\right]
 =A_{1}+A_{2}+A_{3},
\end{align*}
where 
\begin{align*}
 A_{1}&=\frac{1}{\tilde{n}}\sum_{i \in I_k}\left[\hat{\tau}^{(-k)}(\bm{X}_i)-\nu^{(-k)}\right]
 \left\{\left[Y_i(1)-\mu_1^*(\bm{X}_i)\right]\left[\frac{A_i}{\hat{\pi}^{(-k,-k')}(\bm{X}_i)}-\frac{A_i}{\pi(\bm{X}_i)}\right]\right.\\
 &\left.\notag\qquad -
 \left[Y_i(0)-\mu_0^*(\bm{X}_i)\right]\left[\frac{1-A_i}{1-\hat{\pi}^{(-k,-k')}(\bm{X}_i)}-\frac{1-A_i}{1-\pi(\bm{X}_i)}\right]\right\},\\
 A_{2}&=\frac{1}{\tilde{n}}\sum_{i \in I_k}\left[\hat{\tau}^{(-k)}(\bm{X}_i)-\nu^{(-k)}\right]
 \left\{A_i\left[\mu_1^*(\bm{X}_i)-\hat{\mu}^{(-k,-k')}_1(\bm{X}_i)\right]\left[\frac{1}{\hat{\pi}^{(-k,-k')}(\bm{X}_i)}-\frac{1}{\pi(\bm{X}_i)}\right]\right.\\
 &\left.\qquad-(1-A_i)\left[\mu_0^*(\bm{X}_i)-\hat{\mu}^{(-k,-k')}_0(\bm{X}_i)\right]\left[\frac{1}{1-\hat{\pi}^{(-k,-k')}(\bm{X}_i)}-\frac{1}{1-\pi(\bm{X}_i)}\right]\right\},\\
 A_{3}&=\frac{1}{\tilde{n}}\sum_{i \in I_k}\left[\hat{\tau}^{(-k)}(\bm{X}_i)-\nu^{(-k)}\right]
 \left\{\left[\mu_1^*(\bm{X}_i)-\hat{\mu}^{(-k,-k')}_1(\bm{X}_i)\right]\left[\frac{A_i}{\pi(\bm{X}_i)}-1\right]\right.\\
 &\left.\qquad-\left[\mu_0^*(\bm{X}_i)-\hat{\mu}^{(-k,-k')}_0(\bm{X}_i)\right]\left[\frac{1-A_i}{1-\pi(\bm{X}_i)}-1\right]\right\}.
\end{align*}

When $\mu_a^*(\cdot)=\mu_a(\cdot)$ for each $a\in\{0,1\}$, by the law of iterated expectations, we have \(\mathbb{E}_{\bm{X}}(A_{1})=0\). Moreover, 
\begin{align*}
 &\mathbb{E}_{\bm{X}}\left\{\left[\hat{\tau}^{(-k)}(\bm{X}_i)-\nu^{(-k)}\right]
 \left[Y(1)-\mu_1^*(\bm{X})\right]\left[\frac{A}{\hat{\pi}^{(-k,-k')}(\bm{X})}-\frac{A}{\pi(\bm{X})}\right]\right\}^2\\
 &\qquad\overset{(i)}{=}O_p\left(\mathbb{E}_{\bm{X}}\left|\hat{\pi}^{(-k,-k')}(\bm{X})-\pi(\bm{X})\right|^2\right)\\
 &\qquad\overset{(ii)}{=}o_p(1),
\end{align*}
where (i) holds from Assumptions \ref{a2} and \ref{a4}, (ii) holds from Assumption \ref{a3}.
By Chebyshev's inequality, 
\begin{align*}
 &\frac{1}{\tilde{n}}\sum_{i \in I_k}\left[\hat{\tau}^{(-k)}(\bm{X}_i)-\nu^{(-k)}\right]
 \left\{\left[Y_i(1)-\mu_1^*(\bm{X}_i)\right]\left[\frac{A_i}{\hat{\pi}^{(-k,-k')}(\bm{X}_i)}-\frac{A_i}{\pi(\bm{X}_i)}\right]\right\}\\
 &\quad=o_p\left(n^{-\frac{1}{2}}\right).
\end{align*}
On the other hand, when $\mathbb{E}_{\bm{X}}|\hat{\pi}^{(-k,-k')}(\bm{X})-\pi(\bm{X})|^2=o_p(n^{-1})$ is assumed instead, by Markov's inequality,
\begin{align*}
 &\frac{1}{\tilde{n}}\sum_{i \in I_k}\left[\hat{\tau}^{(-k)}(\bm{X}_i)-\nu^{(-k)}\right]
 \left[Y_i(1)-\mu_1^*(\bm{X}_i)\right]\left[\frac{A_i}{\hat{\pi}^{(-k,-k')}(\bm{X}_i)}-\frac{A_i}{\pi(\bm{X}_i)}\right]\\
 &\quad=O_p\left\{\mathbb{E}_{\bm{X}}\left|\left[\hat{\tau}^{(-k)}(\bm{X})-\nu^{(-k)}\right]
 \left[Y(1)-\mu_1^*(\bm{X})\right]\left[\frac{A}{\hat{\pi}^{(-k,-k')}(\bm{X})}-\frac{A}{\pi(\bm{X})}\right]\right|\right\}\\
 &\quad=O_p\left([\mathbb{E}_{\bm{X}}|\hat{\pi}^{(-k,-k')}(\bm{X})-\pi(\bm{X})|^2]^{-1/2}\right)\\
 &\quad=o_p\left(n^{-\frac{1}{2}}\right).
\end{align*}
Similarly, we have
\begin{align*}
 &\frac{1}{\tilde{n}}\sum_{i \in I_k}\left[\hat{\tau}^{(-k)}(\bm{X}_i)-\nu^{(-k)}\right]
 \left[Y_i(0)-\mu_0^*(\bm{X}_i)\right]\left[\frac{1-A_i}{1-\hat{\pi}^{(-k,-k')}(\bm{X}_i)}-\frac{1-A_i}{1-\pi(\bm{X}_i)}\right]\\
 &\quad=o_p\left(n^{-\frac{1}{2}}\right).
\end{align*}
To sup up, we have $A_1=o_p(n^{-{1/2}})$.

Concerning the product error term \(A_{2}\), one obtains 
\begin{align*}
 &\frac{1}{\tilde{n}}\sum_{i \in I_k}\left|\left[\hat{\tau}^{(-k)}(\bm{X}_i)-\nu^{(-k)}\right]
 \left\{\left[\mu_1^*(\bm{X}_i)-\hat{\mu}^{(-k,-k')}_1(\bm{X}_i)\right]\left[\frac{A_i}{\hat{\pi}^{(-k,-k')}(\bm{X}_i)}-\frac{A_i}{\pi(\bm{X}_i)}\right]\right\}\right|\\
 &\quad=O_p\left(\mathbb{E}_{\bm{X}}\left|\left[\mu_1^*(\bm{X})-\hat{\mu}^{(-k,-k')}_1(\bm{X})\right]\left[\hat{\pi}^{(-k,-k')}(\bm{X})-\pi(\bm{X})\right]\right|\right)\\
 &\quad\overset{(i)}{=}o_p\left(n^{-\frac{1}{2}}\right),
\end{align*}
where (i) holds from Assumption \ref{a3} together with Cauchy–Schwarz inequality. Analogously,
\begin{align*}
 &\frac{1}{\tilde{n}}\sum_{i \in I_k}\left|\left[\hat{\tau}^{(-k)}(\bm{X}_i)-\nu^{(-k)}\right]
 \left[\mu_0^*(\bm{X}_i)-\hat{\mu}^{(-k,-k')}_0(\bm{X}_i)\right]\left[\frac{1-A_i}{1-\hat{\pi}^{(-k,-k')}(\bm{X}_i)}-\frac{1-A_i}{1-\pi(\bm{X}_i)}\right]\right|\\&\qquad
 =o_p(n^{-1/2}).
\end{align*}
Therefore, we have \(A_2=o_p(n^{-1/2})\). 

 Besides, by the law of iterated expectations, we have \(\mathbb{E}_{\bm{X}}(A_{3})=0\). Additionally,
\begin{align*}
 &\mathbb{E}_{\bm{X}}\left|\left[\hat{\tau}^{(-k)}(\bm{X})-\nu^{(-k)}\right]
 \left[\mu_1^*(\bm{X})-\hat{\mu}^{(-k,-k')}_1(\bm{X})\right]\left[\frac{A}{\pi(\bm{X})}-1\right]\right|^2\\
 &\quad=O_p\left(\mathbb{E}_{\bm{X}}\left|\mu_1^*(\bm{X})-\hat{\mu}^{(-k,-k')}_1(\bm{X})\right|^2\right)=o_p(1)
\end{align*}
and 
\begin{align*}
 &\mathbb{E}_{\bm{X}}\left|\left[\hat{\tau}^{(-k)}(\bm{X})-\nu^{(-k)}\right]\left[\mu_0^*(\bm{X})-\hat{\mu}^{(-k,-k')}_0(\bm{X})\right]\left[\frac{1-A}{1-\pi(\bm{X})}-1\right]\right|^2\\
 &\quad=O_p\left(\mathbb{E}_{\bm{X}}\left|\mu_0^*(\bm{X})-\hat{\mu}^{(-k,-k')}_0(\bm{X})\right|^2\right)=o_p(1).
\end{align*}
By Chebyshev's inequality, $A_3=o_p(n^{-\frac{1}{2}})$.

Combining the results above, we have
\begin{equation*}
 \frac{1}{\tilde{n}}\sum_{i \in I_k}\left[\hat{\tau}^{(-k)}(\bm{X}_i)-\nu^{(-k)}\right]\,
 \left[\hat{\varphi}^{(-k)}(Z_i)-\varphi^*(Z_i)\right]
 =o_p(n^{-1/2}).
\end{equation*}
Repeating the same procedure replacing \(\hat{\tau}^{(-k)}(\bm{X}_i)-\nu^{(-k)}\) with 1 yields
\begin{equation}\label{156}
 \frac{1}{\tilde{n}}\sum_{i \in I_k}\left[\hat{\varphi}^{(-k)}(Z_i)-\varphi^*(Z_i)\right]
 =o_p(n^{-1/2}).
\end{equation}
Under the assumptions of Lemma \ref{l11} and by Chebyshev’s inequality,
\begin{equation}\label{157}
 \frac{1}{\tilde{n}}\sum_{i \in I_k}\left[\varphi^*(Z_i)-\tau\right]
 =O_p(n^{-1/2}).
\end{equation}
Furthermore, since $\hat{\tau}^{(-k)}(\bm{X})$ is uniformly bounded, Chebyshev’s inequality implies that 
\begin{equation}\label{172}
 |G_k|^{-1}\sum_{i \in G_k}\left[\hat{\tau}^{(-k)}(\bm{X}_i)-\nu^{(-k)}\right]=O_p\left(N^{-\frac{1}{2}}\right).
\end{equation}

Therefore,
\begin{align*}
 &\hat{\tau}^{(-k)}-\tau
 = |G_k|^{-1}\sum_{i \in G_k}\hat{\tau}^{(-k)}(\bm{X}_i)
 + \tilde{n}^{-1}\sum_{i \in I_k}\left[\hat{\varphi}^{(-k)}(Z_i)-\tau-\hat{\tau}^{(-k)}(\bm{X}_i)\right]\\
 &\qquad= |G_k|^{-1}\sum_{i \in G_k}\left[\hat{\tau}^{(-k)}(\bm{X}_i)-\nu^{(-k)}\right]
 + \tilde{n}^{-1}\sum_{i \in I_k}\left[\hat{\varphi}^{(-k)}(Z_i)-\tau - \left(\hat{\tau}^{(-k)}(\bm{X}_i)-\nu^{(-k)}\right)\right]\\
 &\qquad= O_p(n^{-1/2}),
\end{align*}
where the last line follows from \eqref{156}, \eqref{157}, and \eqref{172}.
\end{proof}

\begin{proof}[Proof of Lemma \ref{l2}]
 From the definition of $\hat{{\bm{\alpha}}}^{(-k,-k')}_a$ in equation \eqref{def:alphahat}, it follows that
\begin{align}\label{15}
 &\frac{1}{\lvert I_{-k,-k'}\rvert}
 \sum_{i\in I_{-k,-k'}}
 \mathbbm{1}_{A_i=a}
 \left(Y_{i}-\bm{X}_i^\top \,\hat{\bm{\alpha}}_a^{(-k,-k')}\right)^2
 \;+\;\lambda_{{\bm{\alpha}}}\,\left\lVert \hat{\bm{\alpha}}_a^{(-k,-k')}\right\rVert_1
 \nonumber \\
 &\qquad\leq
 \frac{1}{\lvert I_{-k,-k'}\rvert}
 \sum_{i\in I_{-k,-k'}}
 \mathbbm{1}_{A_i=a}
 \left(Y_{i}-\bm{X}_i^\top \,{\bm{\alpha}}_a^*\right)^2
 \;+\;\lambda_{{\bm{\alpha}}}\,\left\lVert {\bm{\alpha}}_a^*\right\rVert_1,
 \nonumber \\ 
 \qquad\phantom{\leq} 
 &\;\frac{1}{\lvert I_{-k,-k'}\rvert}
 \sum_{i\in I_{-k,-k'}}
 \mathbbm{1}_{A_i=a}
 \left(\bm{X}_i^\top \,\Delta{\bm{\alpha}}_a\right)^2
 \;+\;\lambda_{{\bm{\alpha}}}\,\left\lVert \hat{\bm{\alpha}}_a^{(-k,-k')}\right\rVert_1
 \nonumber \\
 &\qquad\leq
 \frac{2}{\lvert I_{-k,-k'}\rvert}
 \sum_{i\in I_{-k,-k'}}
 \zeta_{a}\,\bm{X}_i^\top \,\Delta{\bm{\alpha}}_a
 \;+\;\lambda_{{\bm{\alpha}}}\,\left\lVert {\bm{\alpha}}_a^*\right\rVert_1,
\end{align}
where $\Delta{\bm{\alpha}}_a=\hat{{\bm{\alpha}}}^{(-k,-k')}_a-{\bm{\alpha}}_a^*$ and $\zeta_{a}=\mathbbm{1}_{A_i=a}(Y_{i}-\bm{X}_i^\top {\bm{\alpha}}_a^*)$. For any $t>0$, set 
\begin{equation*}
 \lambda_{{\bm{\alpha}}}
=16\sigma_{\zeta}\sigma_x\left(\sqrt{\frac{\log(d)}{\left|I_{-k,-k'}\right|}}+t\right).
\end{equation*}
 Define the event $A:=\{\max_{1\leq j\leq d}|{|I_{-k,-k'}|}^{-1}
\sum_{i\in I_{-k,-k'}}X_{ij}\zeta_{a}|
\leq {\lambda_{{\bm{\alpha}}}/4}\},$
where $X_{ij}$ denotes the $j$th component of $\bm{X}_i$. Since
\[
P(A^c)
\leq \sum_{j=1}^{d}
P\left(\left|\frac{1}{\left|I_{-k,-k'}\right|} 
\sum_{i\in I_{-k,-k'}}X_{ij}\zeta_{a}\right|
\geq \frac{\lambda_{{\bm{\alpha}}}}{4}\right),
\]
an application of Lemma~\ref{l1} yields the stated exponential probability bound. The remainder of the argument follows standard Lasso deviation and restricted‐eigenvalue arguments.

 Let \({\bm{e}_j}\in\mathbb{R}^{d}\) denote the standard basis vector with a one in position \(j\) and zeros elsewhere, for \(1\le j\le d\). By Assumptions \ref{a5} and \ref{a6} and Lemma \ref{l1},
\begin{equation*}
 \|{\bm{e}_j}^\top \bm{X}\,\zeta_a\|_{\psi_1}
\;\le\;\|{\bm{e}_j}^\top \bm{X}\|_{\psi_2}\,\|\zeta_a\|_{\psi_2}
\;\le\;\sigma_{\zeta}\sigma_x.
\end{equation*}
\qquad By the definition of \({\bm{\alpha}}_a^*\), we have \(\mathbb{E}[\bm{X}\,\zeta_a]=\bm{0}\).
Bernstein’s inequality yields, for each \(1\le j\le d\) and any \(\epsilon>0\),
\[
P\left(\left|\frac{1}{\left|I_{-k,-k'}\right|}
\sum_{i\in I_{-k,-k'}}X_{ij}\,\zeta_a\right|
\ge2\sigma\,\sigma_x\,\sigma_{\zeta}\,\epsilon
+\sigma\,\sigma_x\,\sigma_{\zeta}\,\epsilon^2\right)
\le2\exp\left(-\left|I_{-k,-k'}\right|\,\epsilon^2\right).
\]
Setting $\epsilon
=\sqrt{{\log(d)/|I_{-k,-k'}|}}
+(\sqrt t - 1),$
and noting \({\log(d)/|I_{-k,-k'}|}<1\), one obtains 
\begin{equation*}
 2\epsilon+\epsilon^2
\le4\sqrt{\frac{\log(d)}{\left|I_{-k,-k'}\right|}}+4t,
\end{equation*}
and hence 
\begin{equation*}
 2\sigma_{\zeta}\,\sigma_x\,\epsilon
+\sigma_{\zeta}\,\sigma_x\,\epsilon^2
\le4\sigma_x\,\sigma_{\zeta}
\left(\sqrt{\frac{\log(d)}{\left|I_{-k,-k'}\right|}}+t\right)
=\frac{\lambda_{{\bm{\alpha}}}}{4}.
\end{equation*}
Therefore, for each \(j\), 
\[
P\left(\left|\frac{1}{\left|I_{-k,-k'}\right|}
\sum_{i\in I_{-k,-k'}}X_{ij}\,\zeta_a\right|
\ge\frac{\lambda_{{\bm{\alpha}}}}{4}\right)
\le2\exp\left(-\left|I_{-k,-k'}\right|\,\epsilon^2\right)
\le\frac{2}{d}\exp(-t+2\sqrt t-1).
\]
A union bound then gives $P(A)\ge1-2\exp[-|I_{-k,-k'}|\,(t-2\sqrt t+1)].$
On the event \(A\), we have
\[
\left|\frac{2}{\left|I_{-k,-k'}\right|}
\sum_{i\in I_{-k,-k'}}\zeta_a\,\bm{X}_i^\top \Delta{\bm{\alpha}}_a\right|
\le\frac{\lambda_{{\bm{\alpha}}}}{2}\|\Delta{\bm{\alpha}}_a\|_1.
\]

Multiplying both sides of \eqref{15} by 2 yields
\[
\frac{2}{\left|I_{-k,-k'}\right|}\sum\mathbf1_{A_i=a}\left(\bm{X}_i^\top \Delta{\bm{\alpha}}_a\right)^2
+2\lambda_{{\bm{\alpha}}}\left\|\hat{\bm{\alpha}}_a^{(-k,-k')}\right\|_1
\le\lambda_{{\bm{\alpha}}}\|\Delta{\bm{\alpha}}_a\|_1
+2\lambda_{{\bm{\alpha}}}\|{\bm{\alpha}}_a^*\|_1.
\]
Let $S$ be the support set of $\hat{{\bm{\alpha}}}^{(-k,-k')}_a$. By the triangle inequality, 
\[
\|\hat{{\bm{\alpha}}}^{(-k,-k')}_a\|_{1}
= \|\hat{{\bm{\alpha}}}^{(-k,-k')}_{a,S}\|_{1} + \|\hat{{\bm{\alpha}}}^{(-k,-k')}_{a,S^c}\|_{1}
\ge \|{\bm{\alpha}}^*_{a,S}\|_{1}
 - \|\Delta{\bm{\alpha}}_{a,S}\|_{1}
 + \|\hat{{\bm{\alpha}}}^{(-k,-k')}_{a,S^c}\|_{1},
\]
and 
\begin{equation}\label{3}
\frac{2}{\left|I_{-k,-k'}\right|}
\sum_{i\in I_{-k,-k'}}\mathbbm{1}_{A_i=a}\left(\bm{X}_i^\top \Delta{\bm{\alpha}}_a\right)^2
+\lambda_{{\bm{\alpha}}}\|\hat{{\bm{\alpha}}}^{(-k,-k')}_{a,S^c}\|_1
\;\le\;3\lambda_{{\bm{\alpha}}}\|\Delta{\bm{\alpha}}_{a,S}\|_{1}.
\end{equation}
Lemma D.6 of \cite{6} provides constants \(\kappa_1,\kappa_2>0\) such that 
\[
\frac{1}{\left|I_{-k,-k'}\right|}
\sum_{i\in I_{-k,-k'}}\mathbbm{1}_{A_i=a}(\bm{X}_i^\top \bm{a})^2
\ge \kappa_1\|\bm{a}\|_{2}\left[\|\bm{a}\|_{2}
 -\kappa_2\sqrt{\frac{\log(d)}{\left|I_{-k,-k'}\right|}}\|\bm{a}\|_{1}\right],
\]
for all \(\|\bm{a}\|_{2}\le1\), with probability at least \(1 - c_1\exp(-c_2\left|I_{-k,-k'}\right|)\). Define the event 
\[
A_1
:= \left\{\frac{1}{\left|I_{-k,-k'}\right|}
\sum_{i\in I_{-k,-k'}}\mathbbm{1}_{A_i=a}\left(\bm{X}_i^\top \Delta{\bm{\alpha}}_a\right)^2
\ge \kappa_1\|\Delta{\bm{\alpha}}_a\|_{2}^2
 -\kappa_1\kappa_2\sqrt{\frac{\log(d)}{\left|I_{-k,-k'}\right|}}
 \|\Delta{\bm{\alpha}}_a\|_{1}\|\Delta{\bm{\alpha}}_a\|_{2}
\right\},
\]
then \(P(A_1)\ge1 - c_1\exp(-c_2\left|I_{-k,-k'}\right|)\). On \(A_1\), together with \eqref{3}, we have 
\[
\frac{2}{\left|I_{-k,-k'}\right|}
\sum_{i\in I_{-k,-k'}}\mathbbm{1}_{A_i=a}\left(\bm{X}_i^\top \Delta{\bm{\alpha}}_a\right)^2
+\lambda_{{\bm{\alpha}}}\|\Delta{\bm{\alpha}}_{a,S^c}\|_1
\le3\lambda_{{\bm{\alpha}}}\|\Delta{\bm{\alpha}}_{a,S}\|_{1}.
\]
Hence, \(\|\Delta{\bm{\alpha}}_{a,S^c}\|_1\le3\|\Delta{\bm{\alpha}}_{a,S}\|_{1}\). Noting that \(\|\Delta{\bm{\alpha}}_{a,S}\|_{1}\le\sqrt{s_{{\bm{\alpha}}}}\|\Delta{\bm{\alpha}}_{a,S}\|_{2}\) gives 
\[
\left\|\Delta{\bm{\alpha}}_a\right\|_{1}
= \|\Delta{\bm{\alpha}}_{a,S}\|_{1} + \|\Delta{\bm{\alpha}}_{a,S^c}\|_{1}
\le4\sqrt{s_{{\bm{\alpha}}}}\|\Delta{\bm{\alpha}}_a\|_{2}.
\]
If \(\left|I_{-k,-k'}\right|>100\kappa_2^2\,s_{{\bm{\alpha}}}\log(d)\), then 
\begin{align*}
\notag
\frac{1}{\left|I_{-k,-k'}\right|}
\sum_{i\in I_{-k,-k'}}\mathbbm{1}_{A_i=a}\left(\bm{X}_i^\top \Delta{\bm{\alpha}}_a\right)^2
&\ge \kappa_1\|\Delta{\bm{\alpha}}_a\|_{2}^2
 -4\kappa_1\kappa_2\sqrt{\frac{s_{{\bm{\alpha}}}\log(d)}{\left|I_{-k,-k'}\right|}}
 \|\Delta{\bm{\alpha}}_a\|_{2}^2\\
&\ge \frac{\kappa_1}{2}\|\Delta{\bm{\alpha}}_a\|_{2}^2
\ge \frac{\kappa_1}{2s_{{\bm{\alpha}}}}\|\Delta{\bm{\alpha}}_{a,S}\|_{1}^2.
\end{align*}
Combination with the inequality above on \(A\cap A_1\) leads to $\kappa_1\|\Delta{\bm{\alpha}}_{a,S}\|_{1}^2/(2s_{{\bm{\alpha}}})
\le3\lambda_{{\bm{\alpha}}}\|\Delta{\bm{\alpha}}_{a,S}\|_{1},$
hence $\|\Delta{\bm{\alpha}}_{a,S}\|_{1}
\le8\,\kappa_1^{-1}\,s_{{\bm{\alpha}}}\,\lambda_{{\bm{\alpha}}},$ $
\|\Delta{\bm{\alpha}}_a\|_{1}
\le32\,\kappa_1^{-1}\,s_{{\bm{\alpha}}}\,\lambda_{{\bm{\alpha}}},$
and
\[
\frac{1}{\left|I_{-k,-k'}\right|}
\sum_{i\in I_{-k,-k'}}\mathbbm{1}_{A_i=a}\left(\bm{X}_i^\top \Delta{\bm{\alpha}}_a\right)^2
\le24\,\kappa_1^{-1}\,s_{{\bm{\alpha}}}\,\lambda_{{\bm{\alpha}}}^2.
\]
Since \(\left|I_{-k,-k'}\right|\asymp n\), it follows that 
\[
\|\Delta{\bm{\alpha}}_a\|_{2}
= O_p\left(\sqrt{\frac{s_{\alpha}\log d}{n}}\right).
\]
\end{proof}

\begin{proof}[Proof of Lemma \ref{l3}]
Here, convergence rates for \(\hat{{\bm{\gamma}}}^{(-k,-k')}\) and \(\hat{\pi}^{(-k,-k')}\) are established without requiring \(\pi^*(\bm{X})=\pi(\bm{X})\). Set 
$\rho(x)=\log[1+\exp(x)],$ $\rho'(x)={\exp(x)/[1+\exp(x)]}=\phi(x),$
 and
 $\rho''(x)={\exp(x)/[1+\exp(x)]^2}=\pi^*(x)[1-\pi^*(x)].$
 Define the empirical loss
\[
L_N({\bm{\gamma}})
:=\frac{1}{|G_{-k,-k'}|}
\sum_{i\in G_{-k,-k'}}
\left[\rho(\bm{X}_i^\top {\bm{\gamma}})-A_i\,\bm{X}_i^\top {\bm{\gamma}}\right],
\]
and its Taylor remainder $\varepsilon_{L_N}(\bm{\Delta},{\bm{\gamma}}^*)
:=L_N({\bm{\gamma}}^*+\bm{\Delta})-L_N({\bm{\gamma}}^*)-\bm{\Delta}^\top \nabla L_N({\bm{\gamma}}^*)$.
Lemma D.6 of \cite{6}  yields, for $\|\bm{\Delta}\|_2\le1$ and with probability at least $1-c_1\exp(-c_2|G_{-k,-k'}|)$,
\begin{align*}
\varepsilon_{L_N}(\bm{\Delta},{\bm{\gamma}}^*)
&\ge \kappa_1\|\bm{\Delta}\|_2
 \left[\|\bm{\Delta}\|_2
 -\kappa_2\sqrt{\frac{\log(d)}{|G_{-k,-k'}|}}\,\|\bm{\Delta}\|_1\right]\\
&\ge \frac{\kappa_1}{2}\|\bm{\Delta}\|_2^2
 -\frac{\kappa_1\kappa_2^2\log(d)}{2\,|G_{-k,-k'}|}\,\|\bm{\Delta}\|_1^2.
\end{align*}

Since \({\bm{\gamma}}^*\) minimizes \(\mathbb{E}[L({\bm{\gamma}})]\), the first‐order condition holds that
\begin{equation*}
\nabla \mathbb{E}[L({\bm{\gamma}}^*)]
=\mathbb{E}\left[\phi\left({\bm{X}}^\top {\bm{\gamma}}^*\right)X - A\,\bm{X}\right]
=0.
\end{equation*}
A union‐bound argument then shows
\[
P\left[\|\nabla L_N({\bm{\gamma}}^*)\|_\infty\ge\frac{\lambda_{\gamma}}{2}\right]
\le \sum_{j=1}^{d}
P\left(\left|\frac{1}{|G_{-k,-k'}|}
\sum_{i\in G_{-k,-k'}}
\left[\phi(\bm{X}_i^\top {\bm{\gamma}}^*)-A_i\right]X_{ij}\right|
\ge\frac{\lambda_{\gamma}}{2}\right).
\]
Since \(|\phi({\bm{X}}^\top {\bm{\gamma}}^*)-A|\le 1\), Lemma~\ref{l1} together with Assumption~\ref{a6} implies that, for any \(i\in G_{-k,-k'}\) and \(1\le j\le d\),
\begin{equation*}
 \left\|\phi(\bm{X}_i^\top {\bm{\gamma}}^*)X_{ij}-A_iX_{ij}\right\|_{\psi_2}\leq \left\|{X}_{ij}\right\|_{\psi_2}\leq \sigma_x.
\end{equation*}
Thus each $(\phi(\bm{X}_i^\top {\bm{\gamma}}^*)-A_i)X_{ij}$ is zero-mean sub-Gaussian. By Hoeffding’s inequality, for each $j\le d$,
\begin{align*}
 &P\left(\left|\frac{1}{|G_{-k,-k'}|}
 \sum_{i \in G_{-k,-k'}}\left(\phi\left(\bm{X}_i^\top {\bm{\gamma}}^*\right)-A_i\right)X_{ij}\right|
 \ge \frac{\lambda_{{\bm{\gamma}}}}{2}\right)\\
 &\quad\le 2\exp\left(-\frac{|G_{-k,-k'}|\lambda_{{\bm{\gamma}}}^2}{32\sigma_{x}^2}\right)\le \frac{2}{d}\exp\left(-|G_{-k,-k'}|t^2\right),
\end{align*}
where for $t>0$, set $\lambda_{\gamma}=4\sqrt{2}\,\sigma_x(\sqrt{{\log(d)/|G_{-k,-k'}|}}+t).$
It follows that
\[
P\left(\left\|\nabla L_N({\bm{\gamma}}^*)\right\|_{\infty}\le \frac{\lambda_{{\bm{\gamma}}}}{2}\right)
 \ge 1 - 2\exp\left(-\left|G_{-k,-k'}\right|\,t^2\right).
\]
When $G_{-k,-k'} \ge 64\,\kappa_2^2\,s_{{\bm{\gamma}}}\log(d)$, $|G_{-k,-k'}|\asymp N$, and $18\,s_{{\bm{\gamma}}}\,\lambda_{{\bm{\gamma}}}^2 \le \kappa_1^2$, Theorem 9.19 and Corollary 9.20 of \cite{7} yield
\begin{equation*}
 \left\|\hat{{\bm{\gamma}}}^{(-k,-k')}-{\bm{\gamma}}^*\right\|_{2}
 \le \frac{3\sqrt{s_{{\bm{\gamma}}}}\,\lambda_{{\bm{\gamma}}}}{\kappa_1},
 \qquad
 \left\|\hat{{\bm{\gamma}}}^{(-k,-k')}-{\bm{\gamma}}^*\right\|_{1}
 \le \frac{6\,s_{{\bm{\gamma}}}\,\lambda_{{\bm{\gamma}}}}{\kappa_1},
\end{equation*}
with probability at least $1-2\exp\left[-(\tilde n+\tilde m)t^2\right]-c_1\exp\left[-c_2(\tilde n+\tilde m)\right]$. If $N\gg s_{{\bm{\gamma}}}\log d$ and $\lambda_{{\bm{\gamma}}}\asymp\sqrt{\frac{\log d}{N}}$, then
\begin{equation*}
 \left\|\hat{{\bm{\gamma}}}^{(-k,-k')}-{\bm{\gamma}}^*\right\|_{2}^2
 = O_p\left(\frac{s_{{\bm{\gamma}}}\log(d)}{N}\right).
\end{equation*}

Next, we derive the convergence rate of \(\hat{\pi}^{(-k,-k')}(\bm{X})\). Since \(\rho'(x)=\phi(x)\), Taylor’s theorem implies, for some \(\hat t\in(0,1)\),
\begin{align*}
&\mathbb{E}_{\bm{X}}\left[\phi\left({\bm{X}}^\top \hat{\bm{\gamma}}\right)-\phi\left({\bm{X}}^\top {\bm{\gamma}}^*\right)\right]^2
=\mathbb{E}_{\bm{X}}\left[\phi\left(\bm{v}^*+\bm{\Delta}\right)-\phi(\bm{v}^*)\right]^2\\
&\qquad=\mathbb{E}_{\bm{X}}\left[\phi(\bm{v}^*)-\phi(\bm{v}^*)\right]^2
 +\frac{d}{dt}\mathbb{E}_{\bm{X}}\left[\phi(\bm{v}^*+t\bm{\Delta})-\phi(\bm{v}^*)\right]^2\Big|_{t=0}\\
 &\qquad\qquad+\frac{1}{2}\,\frac{d^2}{dt^2}\mathbb{E}_{\bm{X}}\left[\phi(\bm{v}^*+t\bm{\Delta})-\phi(\bm{v}^*)\right]^2\Big|_{t=\hat t},
\end{align*}
where \(\bm{v}^*={\bm{X}}^\top {\bm{\gamma}}^*\) and \(\bm{\Delta}={\bm{X}}^\top (\hat{\bm{\gamma}}^{(-k,-k')}-{\bm{\gamma}}^*)\). Noting that
\begin{align*}
\frac{d}{dt}\left[\phi(\bm{v}^*+t\bm{\Delta})-\phi(\bm{v}^*)\right]^2
&=2\left[\phi(\bm{v}^*+t\bm{\Delta})-\phi(\bm{v}^*)\right]\,\phi'(\bm{v}^*+t\bm{\Delta})\,\bm{\Delta},\\
\frac{d^2}{dt^2}\left[\phi(\bm{v}^*+t\bm{\Delta})-\phi(\bm{v}^*)\right]^2
&=2\left[\phi'(\bm{v}^*+t\bm{\Delta})\right]^2\bm{\Delta}^2
+2\left[\phi(\bm{v}^*+t\bm{\Delta})-\phi(\bm{v}^*)\right]\,\phi''(\bm{v}^*+t\bm{\Delta})\,\bm{\Delta}^2,
\end{align*}
and since \(\phi'(x)\in(0,1)\), \(\phi''(x)\in(-1,1)\), it follows that
\[
\mathbb{E}_{\bm{X}}\left[\phi\left({\bm{X}}^\top \hat{\bm{\gamma}}\right)-\phi\left({\bm{X}}^\top {\bm{\gamma}}^*\right)\right]^2
\le2\,\mathbb{E}[\bm{\Delta}^2]
=2\,\mathbb{E}_{\bm{X}}\left[{\bm{X}}^\top \left(\hat{\bm{\gamma}}^{(-k,-k')}-{\bm{\gamma}}^*\right)\right]^2.
\]
Hence
\[
\mathbb{E}_{\bm{X}}\left[\hat\pi^{(-k,-k')}(\bm{X})-\pi^*(\bm{X})\right]^2
=\mathbb{E}_{\bm{X}}\left[\phi\left({\bm{X}}^\top \hat{\bm{\gamma}}^{(-k,-k')}\right)-\phi\left({\bm{X}}^\top {\bm{\gamma}}^*\right)\right]^2
\;=\;O_p\left(\frac{s_{\gamma}\log(d)}{N}\right).
\]
\end{proof}
\begin{proof}[Proof of Lemma \ref{l4}]
Minkowski’s inequality implies
\begin{align*}
 \left[\mathbb{E}_{\bm{X}}\left|\hat{\pi}^{(-k,-k')}(\bm{X})\right|^{-r}\right]^{\frac{1}{r}}
 &= \left[\mathbb{E}_{\bm{X}}\left|1+\exp\left(-{\bm{X}}^\top \hat{{\bm{\gamma}}}^{(-k,-k')}\right)\right|^r\right]^{\frac{1}{r}}\\
 &
 \le 1 + \left[\mathbb{E}_{\bm{X}}\left|\exp\left(-{\bm{X}}^\top \hat{{\bm{\gamma}}}^{(-k,-k')}\right)\right|^r\right]^{\frac{1}{r}}
\end{align*}
and
\begin{align*}
 \left[\mathbb{E}_{\bm{X}}\left|1-\hat{\pi}^{(-k,-k')}(\bm{X})\right|^{-r}\right]^{\frac{1}{r}}
 &= \left[\mathbb{E}_{\bm{X}}\left|1+\exp\left({\bm{X}}^\top \hat{{\bm{\gamma}}}^{(-k,-k')}\right)\right|^r\right]^{\frac{1}{r}}\\
 &
 \le 1 + \left[\mathbb{E}_{\bm{X}}\left|\exp\left({\bm{X}}^\top \hat{{\bm{\gamma}}}^{(-k,-k')}\right)\right|^r\right]^{\frac{1}{r}}.
\end{align*}
Assumption \ref{a7} ensures
\begin{equation*}
 P\left(\frac{c_0}{1-c_0}\le \exp(-{\bm{X}}^\top {\bm{\gamma}}^*)\le\frac{1-c_0}{c_0}\right)=1,
\quad
 P\left(\frac{c_0}{1-c_0}\le \exp\left({\bm{X}}^\top {\bm{\gamma}}^*\right)\le\frac{1-c_0}{c_0}\right)=1.
\end{equation*}
Hence
\begin{align*}
 \left[\mathbb{E}_{\bm{X}}\left|\exp\left(-{\bm{X}}^\top \hat{{\bm{\gamma}}}^{(-k,-k')}\right)\right|^r\right]^{\frac{1}{r}}
 &\le \frac{1-c_0}{c_0}\,\left\{\mathbb{E}_{\bm{X}}\left|\exp\left[-{\bm{X}}^\top \left(\hat{{\bm{\gamma}}}^{(-k,-k')}-{\bm{\gamma}}^*\right)\right]\right|^r\right\}^{\frac{1}{r}},\\
 \left[\mathbb{E}_{\bm{X}}\left|\exp\left({\bm{X}}^\top \hat{{\bm{\gamma}}}^{(-k,-k')}\right)\right|^r\right]^{\frac{1}{r}}
 &\le \frac{1-c_0}{c_0}\,\left\{\mathbb{E}_{\bm{X}}\left|\exp\left[{\bm{X}}^\top \left(\hat{{\bm{\gamma}}}^{(-k,-k')}-{\bm{\gamma}}^*\right)\right]\right|^r\right\}^{\frac{1}{r}}.
\end{align*}
By Assumption \ref{a6}, 
\begin{equation*}
 \left\|\bm{X}^\top \left(\hat{{\bm{\gamma}}}^{(-k,-k')}-{\bm{\gamma}}^*\right)\right\|_{\psi_2}
 \le 2\,\sigma_x\,\left\|\hat{{\bm{\gamma}}}^{(-k,-k')}-{\bm{\gamma}}^*\right\|_2.
\end{equation*}
Let \(\delta = \mathbb{E}_{\bm{X}}|\bm{X}^\top (\hat{{\bm{\gamma}}}^{(-k,-k')}-{\bm{\gamma}}^*)|\). On the event $B$ defined in Lemma~\ref{l4}, 
\begin{equation*}
 \delta \le \sqrt{\pi}\,\sigma_x,
 \qquad
 \left\|\delta\right\|_{\psi_2} \le (\log 2)^{-1/2}\sqrt{\pi}\,\sigma_x,
\end{equation*}
and Lemma \ref{l1} gives
\begin{equation*}
 \left\|\left|\bm{X}^\top \left(\hat{{\bm{\gamma}}}^{(-k,-k')}-{\bm{\gamma}}^*\right)\right|-\delta\right\|_{\psi_2}
 \le \left\|\bm{X}^\top \left(\hat{{\bm{\gamma}}}^{(-k,-k')}-{\bm{\gamma}}^*\right)\right\|_{\psi_2}
 + \|\delta\|_{\psi_2}
 \le \left[1+(\log2)^{-1/2}\sqrt\pi\right]\,\sigma_x.
\end{equation*}
Finally, the moment‐generating bound
\[
\mathbb{E}_{\bm{X}}\left[\exp\left[\lambda\left(\left|\bm{X}^\top \left(\hat{{\bm{\gamma}}}^{(-k,-k')}-{\bm{\gamma}}^*\right)\right|-\delta\right)\right]\right]
\le \exp\left\{2\,\lambda^2\,\left[1+(\log2)^{-1/2}\sqrt\pi\right]^2\,\sigma_{\bm{X}}^2\right\}
\]
shows that \([\mathbb{E}_{\bm{X}}|\exp[{\bm{X}}^\top (\hat{{\bm{\gamma}}}^{(-k,-k')}-{\bm{\gamma}}^*)]|^r]^{1/r}\) remains bounded. Then,
\begin{align*}
 &\left\{\mathbb{E}_{\bm{X}}\left|\exp\left[-{\bm{X}}^\top \left(\hat{{\bm{\gamma}}}^{(-k,-k')}-{\bm{\gamma}}^*\right)\right]\right|^r\right\}^{\frac{1}{r}}\\
 &\quad=\left\{\mathbb{E}_{\bm{X}}\left|\exp\left[-r{\bm{X}}^\top \left(\hat{{\bm{\gamma}}}^{(-k,-k')}-{\bm{\gamma}}^*\right)\right]\right|\right\}^{\frac{1}{r}}\\
 &\quad\leq \left\{\mathbb{E}\exp\left[r\left|\bm{X}^\top \left(\hat{{\bm{\gamma}}}^{(-k,-k')}-{\bm{\gamma}}^*\right)\right|\right]\right\}^{\frac{1}{r}}\\
 &\quad
 \leq \exp\left\{\sqrt{\pi}\sigma_x+2r\left[1+(\log(2))^{-\frac{1}{2}}\sqrt{\pi}\right]^2\sigma_{x}^2\right\}.
\end{align*}
Hence \([\mathbb{E}_{\bm{X}}|\hat{\pi}^{(-k,-k')}(\bm{X})|^{-r}]^{1/r}\) is bounded; replacing \(r\) by \(2r\) yields a similar bound for
\([\mathbb{E}_{\bm{X}}|\hat{\pi}^{(-k,-k')}(\bm{X})|^{-2r}]^{1/(2r)}\). Similarly, \([\mathbb{E}_{\bm{X}}|1-\hat{\pi}^{(-k,-k')}(\bm{X})|^{-r}]^{1/r}\) is uniformly bounded.

Next, Taylor’s theorem gives
\begin{align*}
 &\left[\mathbb{E}_{\bm{X}}\left|\frac1{\hat{\pi}^{(-k,-k')}(\bm{X})}-\frac1{\pi^*(\bm{X})}\right|^{r}\right]^{1/r}\\
 &\quad
 =\left\{\mathbb{E}_{\bm{X}}\left|\exp(-{\bm{X}}^\top {\bm{\gamma}}^*)\left[\exp[-{\bm{X}}^\top \left(\hat{\bm{\gamma}}^{(-k,-k')}-{\bm{\gamma}}^*\right)]-1\right]\right|^r\right\}^{1/r}\\
 &\quad\le\frac{1-c_0}{c_0}\left\{\mathbb{E}_{\bm{X}}\left|\exp\left[-{\bm{X}}^\top \left(\hat{\bm{\gamma}}^{(-k,-k')}-{\bm{\gamma}}^*\right)\right]-1\right|^r\right\}^{1/r}\\
 &\quad\overset{(i)}{\le}\frac{1-c_0}{c_0}\left\{\mathbb{E}_{\bm{X}}\left|\left[1+\exp[-{\bm{X}}^\top \left(\hat{\bm{\gamma}}^{(-k,-k')}-{\bm{\gamma}}^*\right)]\right]\,{\bm{X}}^\top \left(\hat{\bm{\gamma}}^{(-k,-k')}-{\bm{\gamma}}^*\right)\right|^r\right\}^{1/r}\\
 &\quad\overset{(ii)}{\le}\frac{1-c_0}{c_0}\left[\mathbb{E}_{\bm{X}}\left|\bm{X}^\top \left(\hat{\bm{\gamma}}^{(-k,-k')}-{\bm{\gamma}}^*\right)\right|^r\right]^{1/r}\\
 &\qquad+\frac{1-c_0}{c_0}\left\{\mathbb{E}_{\bm{X}}\left|\exp\left[-{\bm{X}}^\top \left(\hat{\bm{\gamma}}^{(-k,-k')}-{\bm{\gamma}}^*\right)\right]\right|^{2r}\right\}^{1/(2r)}
 \left[\mathbb{E}_{\bm{X}}\left|\bm{X}^\top \left(\hat{\bm{\gamma}}^{(-k,-k')}-{\bm{\gamma}}^*\right)\right|^{2r}\right]^{1/(2r)},
\end{align*}
where (i) follows from a Taylor expansion, and (ii) from Minkowski’s and H\"older’s inequalities.
 Since \(|G_{-k,-k'}|\asymp N\), Lemma \ref{l1} yields
\[
\left[\mathbb{E}_{\bm{X}}\left|\bm{X}^\top \left(\hat{\bm{\gamma}}^{(-k,-k')}-{\bm{\gamma}}^*\right)\right|^r\right]^{1/r}
=O\left(\left\|\hat{\bm{\gamma}}^{(-k,-k')}-{\bm{\gamma}}^*\right\|_2\right)
=O_p\left(\sqrt{\frac{s_{\gamma}\log(d)}{N}}\right).
\]
The same argument applies to $\{\mathbb{E}_{\bm{X}}|[{1-\hat{\pi}^{(-k,-k')}(\bm{X})}]^{-1}-[{1-\pi^*(\bm{X})}]^{-1}|^r\}^{1/r},$
so that
\begin{equation*}
 \left[\mathbb{E}_{\bm{X}}\left|\frac1{\hat{\pi}^{(-k,-k')}(\bm{X})}-\frac1{\pi^*(\bm{X})}\right|^r\right]^{1/r}
 =O_p\left(\sqrt{\frac{s_{\gamma}\log(d)}{N}}\right),
\end{equation*}
\begin{equation*}
 \left[\mathbb{E}_{\bm{X}}\left|\frac1{1-\hat{\pi}^{(-k,-k')}(\bm{X})}-\frac1{1-\pi^*(\bm{X})}\right|^r\right]^{1/r}
 =O_p\left(\sqrt{\frac{s_{\gamma}\log(d)}{N}}\right).
\end{equation*}
\end{proof}
\begin{proof}[Proof of Lemma \ref{l5}]
First, establish \(\|\varphi^*(Z)\|_{\psi_2}=O(1)\). By definition,
\[
\varphi^*(Z)
=\frac{A}{\pi^*(\bm{X})}\left[Y(1)-\mu_1^*(\bm{X})\right]
-\frac{1-A}{1-\pi^*(\bm{X})}\left[Y(0)-\mu_0^*(\bm{X})\right]
+\mu_1^*(\bm{X})-\mu_0^*(\bm{X}).
\]
Assumption \ref{a5} implies each \(\mathbbm{1}_{A=a}[Y-\mu_a^*(\bm{X})]\) is sub-Gaussian, and Assumption \ref{a7} ensures \(\pi^*(\bm{X})^{-1}\) and \([1-\pi^*(\bm{X})]^{-1}\) are bounded. Hence
\[
\left\|\frac{A}{\pi^*(\bm{X})}\left[Y(1)-\mu_1^*(\bm{X})\right]
 -\frac{1-A}{1-\pi^*(\bm{X})}\left[Y(0)-\mu_0^*(\bm{X})\right]\right\|_{\psi_2}
=O(1).
\]
Next, \(\mu_1^*(\bm{X})={\bm{X}}^\top {\bm{\alpha}}_1^*\) where \({\bm{\alpha}}_1^*\) minimizes 
\(\mathbb{E}[A(Y-{\bm{X}}^\top {\bm{\alpha}})^2]\). Since
\[
\mathbb{E}Y^2 \ge \mathbb{E}\left(A\,Y^2\right)
\ge \mathbb{E}\left[\pi(\bm{X})\,\left({\bm{X}}^\top {\bm{\alpha}}_1^*\right)^2\right]
\ge c_0\,\mathbb{E}\,\left({\bm{X}}^\top {\bm{\alpha}}_1^*\right)^2
\ge c_0\kappa_l\|{\bm{\alpha}}_1^*\|_2^2,
\]
and \(\mathbb{E}Y^2=O(1)\) by Assumption \ref{a5}, it follows \(\|{\bm{\alpha}}_1^*\|_2=O(1)\). Then Assumption \ref{a6} gives
\[
\|\mu_1^*(\bm{X})\|_{\psi_2}
=\|\bm{X}^\top {\bm{\alpha}}_1^*\|_{\psi_2}
\le\sigma_x\,\|{\bm{\alpha}}_1^*\|_2
=O(1).
\]
Thus \(\|\varphi^*(Z)\|_{\psi_2}=O(1)\).
Now consider that \({\bm{\beta}}^*=\arg\min_{\bm{\beta}} \mathbb{E}[\tau(\bm{X})-\bm{W}^\top {\bm{\beta}}]^2\). Observe
\[
\mathbb{E}[\tau({\bm{X}})]^2 \ge \mathbb{E}[(\bm{W}^\top {\bm{\beta}}^*)^2]
= {\bm{\beta}}^{*^\top}\mathbb{E}[\bm{W}\bm{W}^\top ]{\bm{\beta}}^*
\ge \kappa_l\|{\bm{\beta}}^*\|_2^2.
\]
Meanwhile, $\mathbb{E}[\tau({\bm{X}})]^2=\mathbb{E}\{\mathbb{E}[Y(1)-Y(0)\mid\bm{X}]^2\}=O(1)$ under Assumption \ref{a5}. Hence, \(\|{\bm{\beta}}^*\|_2=O(1)\). Finally,
\[
\|\epsilon\|_{\psi_2}
=\|\varphi^*(Z)-\bm{W}^\top {\bm{\beta}}^*\|_{\psi_2}
\le\|\varphi^*(Z)\|_{\psi_2} + \|\bm{W}^\top {\bm{\beta}}^*\|_{\psi_2}
\overset{(i)}{\le}O(1) + \sigma_x\,\|{\bm{\beta}}^*\|_2
=O(1),
\]
where (i) holds from \(\|\varphi^*(Z)\|_{\psi_2}=O(1)\) and Assumption \ref{a6}.
\end{proof}

\begin{proof}[Proof of Lemma \ref{l6}]
By Jensen’s inequality, $\sup_{\|{\bm{\alpha}}\|_2\le1}|\mathbb{E}[\bm{W}^\top {\bm{\alpha}}]|
\le\sup_{\|{\bm{\alpha}}\|_2=1}\mathbb{E}|\bm{W}^\top {\bm{\alpha}}|
=O(1),$
where the last step holds by Assumption~\ref{a6} and Lemma~\ref{l1}. Since any bounded variable is sub‐Gaussian, \(\mathbb{E}(\bm{W})^\top {\bm{\alpha}}\) is sub‐Gaussian for all \(\|{\bm{\alpha}}\|_2\le1\). Hence the triangle inequality gives
\begin{equation}\label{40}
\|\bm{D}^\top {\bm{\alpha}}\|_{\psi_2}
\le\|\bm{W}^\top {\bm{\alpha}}\|_{\psi_2}+\left\|\mathbb{E}(\bm{W})^\top {\bm{\alpha}}\right\|_{\psi_2}
=O(1).
\end{equation}
 Next, recall that
\begin{equation}\label{160}
 \hat{\varphi}^{(-k)}(Z_i)
=\frac{1}{K-1}
\sum_{k'\neq k}
\hat{\varphi}^{(-k,-k')}(Z_i),
\end{equation}
with each
\begin{align*}
&\hat{\varphi}^{(-k,-k')}(Z_i)\\&\qquad
=\frac{A_i}{\hat\pi^{(-k,-k')}(\bm{X}_i)}
\left[Y_i(1)-\hat\mu_1^{(-k,-k')}(\bm{X}_i)\right]
-\frac{1-A_i}{1-\hat\pi^{(-k,-k')}(\bm{X}_i)}
\left[Y_i(0)-\hat\mu_0^{(-k,-k')}(\bm{X}_i)\right]\\
&\qquad\qquad+\hat\mu_1^{(-k,-k')}(\bm{X}_i)-\hat\mu_0^{(-k,-k')}(\bm{X}_i).
\end{align*}
By Lemma~\ref{l2} and Assumption~\ref{a6}, if \(s_{{\bm{\alpha}}}\log(d)=o( n)\), then
\[
\left\|\bm{X}^\top \left(\hat{\bm{\alpha}}_a^{(-k,-k')}-{\bm{\alpha}}_a^*\right)\right\|_{\psi_2}
\le\sigma_x\left\|\hat{\bm{\alpha}}_a^{(-k,-k')}-{\bm{\alpha}}_a^*\right\|_2
= o_p(1).
\]
Application of Lemma~\ref{l1} yields, for any \(l>0\),
\begin{equation}\label{138}
\mathbb{E}_{\bm{X}}\left|\bm{X}^\top \left(\hat{\bm{\alpha}}_a^{(-k,-k')}-{\bm{\alpha}}_a^*\right)\right|^l
=o_p(1).
\end{equation}
Furthermore,
\begin{align}\label{139}
 \notag&\left\{\mathbb{E}_{\bm{X}}\left|\frac{A}{\hat{\pi}^{(-k,-k')}(\bm{X})}\left[Y(1)-\hat{\mu}^{(-k,-k')}_1(\bm{X})\right]-\frac{A}{\pi^*(\bm{X})}\left[Y(1)-\mu^*_1(\bm{X})\right]\right|^4\right\}^{\frac{1}{4}}\\
 \notag&\qquad=\left\{\mathbb{E}_{\bm{X}}\left|\frac{A}{\hat{\pi}^{(-k,-k')}(\bm{X})}\left[Y(1)-\mu_1^*(\bm{X})+\mu_1^*(\bm{X})-\hat{\mu}^{(-k,-k')}_1(\bm{X})\right]\right.\right.\\
 &\notag\left.\left.\qquad\qquad-\frac{A}{\pi^*(\bm{X})}\left[Y(1)-\mu^*_1(\bm{X})\right]\right|^4\right\}^{\frac{1}{4}}\\
 \notag&\qquad=\left\{\mathbb{E}_{\bm{X}}\left|A\left[Y(1)-\mu^*_1(\bm{X})\right]\left[\frac{1}{\hat{\pi}^{(-k,-k')}(\bm{X})}-\frac{1}{\pi^*(\bm{X})}\right]\right.\right.\\
 &\notag\left.\left.\qquad\qquad+\frac{A}{\hat{\pi}^{(-k,-k')}(\bm{X})}\left[\mu_1^*(\bm{X})-\hat{\mu}^{(-k,-k')}_1(\bm{X})\right]\right|^4\right\}^{\frac{1}{4}}\\
 \notag&\qquad\overset{(i)}{\leq} \left\{\mathbb{E}_{\bm{X}}\left|A\left[Y(1)-\mu^*_1(\bm{X})\right]\left[\frac{1}{\hat{\pi}^{(-k,-k')}(\bm{X})}-\frac{1}{\pi^*(\bm{X})}\right]\right|^4\right\}^{\frac{1}{4}}\\
 \notag&\qquad\qquad+
 \left\{\mathbb{E}_{\bm{X}}\left|\frac{A}{\hat{\pi}^{(-k,-k')}(\bm{X})}\left[\mu_1^*(\bm{X})-\hat{\mu}^{(-k,-k')}_1(\bm{X})\right]\right|^4\right\}^{\frac{1}{4}}\\
 \notag&\qquad\leq \left\{\mathbb{E}_{\bm{X}}\left|A\left[Y(1)-\mu^*_1(\bm{X})\right]\right|^8\right\}^{\frac{1}{8}} \left\{\mathbb{E}_{\bm{X}}\left|\left[\frac{1}{\hat{\pi}^{(-k,-k')}(\bm{X})}-\frac{1}{\pi^*(\bm{X})}\right]\right|^8\right\}^{\frac{1}{8}}\\
 \notag&\qquad\qquad+\left\{\mathbb{E}_{\bm{X}}\left|\frac{1}{\hat{\pi}^{(-k,-k')}(\bm{X})}\right|^8\right\}^{\frac{1}{8}}
 \left\{\mathbb{E}_{\bm{X}}\left|\left[\mu_1^*(\bm{X})-\hat{\mu}^{(-k,-k')}_1(\bm{X})\right]\right|^8\right\}^{\frac{1}{8}}\\
 &\qquad\overset{(ii)}{=}O(1)o_p(1)+O_p(1)o_p(1)=o_p(1),
\end{align}
where (i) holds from Minkovski's inequality, (ii) holds since Assumption \ref{a5}, Lemma \ref{l4}, and \eqref{138}. Similarly 
\begin{align}\label{140}
 \notag&\left\{\mathbb{E}_{\bm{X}}\left|\frac{1-A}{1-\hat{\pi}^{(-k,-k')}(\bm{X})}\left[Y(0)-\hat{\mu}^{(-k,-k')}_0(\bm{X})\right]-\frac{1-A}{1-\pi^*(\bm{X})}\left[Y(0)-\mu^*_0(\bm{X})\right]\right|^4\right\}^{\frac{1}{4}}\\
 \notag&\qquad=\left\{\mathbb{E}_{\bm{X}}\left|\frac{1-A}{1-\hat{\pi}^{(-k,-k')}(\bm{X})}\left[Y(0)-\mu_0^*(\bm{X})+\mu_0^*(\bm{X})-\hat{\mu}^{(-k,-k')}_0(\bm{X})\right]\right.\right.\\
 &\left.\left.\qquad\qquad\notag-\frac{1-A}{1-\pi^*(\bm{X})}\left[Y(0)-\mu^*_0(\bm{X})\right]\right|^4\right\}^{\frac{1}{4}}\\
 \notag&\qquad=\left\{\mathbb{E}_{\bm{X}}\left|(1-A)\left[Y(0)-\mu^*_0(\bm{X})\right]\left[\frac{1}{1-\hat{\pi}^{(-k,-k')}(\bm{X})}-\frac{1}{1-\pi^*(\bm{X})}\right]\right.\right.\\
 &\left.\left.\notag\qquad\qquad+\frac{1-A}{1-\hat{\pi}^{(-k,-k')}(\bm{X})}\left[\mu_0^*(\bm{X})-\hat{\mu}^{(-k,-k')}_0(\bm{X})\right]\right|^4\right\}^{\frac{1}{4}}\\
 \notag&\qquad\leq \left\{\mathbb{E}_{\bm{X}}\left|(1-A)\left[Y(0)-\mu^*_0(\bm{X})\right]\left[\frac{1}{1-\hat{\pi}^{(-k,-k')}(\bm{X})}-\frac{1}{1-\pi^*(\bm{X})}\right]\right|^4\right\}^{\frac{1}{4}}\\
 &\notag\qquad\qquad\notag +
 \left\{\mathbb{E}_{\bm{X}}\left|\frac{1-A}{1-\hat{\pi}^{(-k,-k')}(\bm{X})}\left[\mu_0^*(\bm{X})-\hat{\mu}^{(-k,-k')}_0(\bm{X})\right]\right|^4\right\}^{\frac{1}{4}}\\
 \notag&\qquad\leq \left\{\mathbb{E}_{\bm{X}}\left|(1-A)\left[Y(0)-\mu^*_0(\bm{X})\right]\right|^8\right\}^{\frac{1}{8}} \left\{\mathbb{E}_{\bm{X}}\left|\left[\frac{1}{1-\hat{\pi}^{(-k,-k')}(\bm{X})}-\frac{1}{1-\pi^*(\bm{X})}\right]\right|^8\right\}^{\frac{1}{8}}\\
 \notag&\qquad\qquad+\left\{\mathbb{E}_{\bm{X}}\left|\frac{1}{1-\hat{\pi}^{(-k,-k')}(\bm{X})}\right|^8\right\}^{\frac{1}{8}}
 \left\{\mathbb{E}_{\bm{X}}\left|\left[\mu_0^*(\bm{X})-\hat{\mu}^{(-k,-k')}_0(\bm{X})\right]\right|^8\right\}^{\frac{1}{8}}\\
 &\qquad=O(1)o_p(1)+O_p(1)o_p(1)=o_p(1).
\end{align}
Combining \eqref{160}, \eqref{138}, \eqref{139}, and \eqref{140} yields
$\{\mathbb{E}_{\bm{X}}[|\hat{\varphi}^{(-k)}(Z)-\varphi^*(Z)]|^4\}^{1/4}
= o_p(1).$
Moreover, Markov’s inequality gives
\begin{equation}\label{35}
\frac{1}{\tilde n}\sum_{i\in I_k}\left[\hat{\varphi}^{(-k)}(Z_i)-\varphi^*(Z_i)\right]
= o_p(1).
\end{equation}

To show that \(\hat{\tau}_{\mathrm{para}}-\tau=o_p(1)\), note from the definition of \(\hat{\tau}_{\mathrm{para}}^{(k)}\),
\begin{align*}
&\hat{\tau}_{\mathrm{para}}^{(k)}-\tau
=|G_k|^{-1}\sum_{i \in G_k}{\bm{W}_i}^\top \hat{\bm{\beta}}^{(-k)}
 +\tilde n^{-1}\sum_{i\in I_k}\left[\hat\varphi^{(-k)}(Z_i)-\tau-{\bm{W}_i}^\top \hat{\bm{\beta}}^{(-k)}\right]\\
&\qquad=|G_k|^{-1}\sum_{i\in G_k}{\bm{D}_i}^\top \hat{\bm{\beta}}^{(-k)}
 +\tilde n^{-1}\sum_{i\in I_k}\left[\hat\varphi^{(-k)}(Z_i)-\varphi^*(Z_i)
 +\varphi^*(Z_i)-\tau-{\bm{D}_i}^\top \hat{\bm{\beta}}^{(-k)}\right].
\end{align*}
Since $\|\hat{\bm{\beta}}^{(-k)}\|_2
\le\|\hat{\bm{\beta}}^{(-k)}-{\bm{\beta}}^*\|_2+\|{\bm{\beta}}^*\|_2
=O_p(1),$
and \(\bm{D}\) is zero-mean sub‐Gaussian by \eqref{40}. Chebyshev’s inequality implies
\begin{align}\label{11}
|G_k|^{-1}\sum_{i\in G_k}{\bm{D}_i}^\top \hat{\bm{\beta}}^{(-k)}
=O_p\left(N^{-1/2}\right),
\qquad
\frac{1}{\tilde n}\sum_{i\in I_k}{\bm{D}_i}^\top \hat{\bm{\beta}}^{(-k)}
=O_p(n^{-1/2}).
\end{align}
Finally, Chebyshev’s inequality together with \(\|\varphi^*(Z)\|_{\psi_2}=O(1)\) gives
\[
\tilde n^{-1}\sum_{i\in I_k}\left(\varphi^*(Z_i)-\tau\right)
=O_p(n^{-1/2}),
\]
and \eqref{35} controls the remaining term. Averaging over \(k\in\{1,\dots,K\}\) yields \(\hat\tau_{\mathrm{para}}-\tau=o_p(1)\) for any finite \(K\).
\end{proof}

\begin{proof}[Proof of Lemma \ref{l7}]
 Starting from the definition of $\hat{\varphi}^{(-k)}(\bm{X})$ and $\varphi^*(\bm{X})$,
\begin{align}\label{149}
 \frac{1}{\tilde{n}}\sum_{i \in I_k}\hat{{\bm{\beta}}}^{(-k)^\top }{\bm{D}_i} \left[\hat{\varphi}^{(-k,-k')}(Z_i
 )-\varphi^*(Z_i
 )\right]
 =B_{1}+B_{2}+B_{3},
\end{align}
where 
\begin{align*}
 B_{1}&=\frac{1}{\tilde{n}}\sum_{i \in I_k}\hat{{\bm{\beta}}}^{(-k)^\top }{\bm{D}_i}
 \left\{\left[Y_i(1)-\mu_1^*(\bm{X}_i)\right]\left[\frac{A_i}{\hat{\pi}^{(-k,-k')}(\bm{X}_i)}-\frac{A_i}{\pi^*(\bm{X}_i)}\right]\right.\\
 &\left.\qquad-\left[Y_i(0)-\mu_0^*(\bm{X}_i)\right]\left[\frac{1-A_i}{1-\hat{\pi}^{(-k,-k')}(\bm{X}_i)}-\frac{1-A_i}{1-\pi^*(\bm{X}_i)}\right]\right\},\\
 B_{2}&=\frac{1}{\tilde{n}}\sum_{i \in I_k}\hat{{\bm{\beta}}}^{(-k)^\top }{\bm{D}_i}
 \left\{\left[\mu_1^*(\bm{X}_i)-\hat{\mu}^{(-k,-k')}_1(\bm{X}_i)\right]\left[\frac{A_i}{\pi^*(\bm{X}_i)}-1\right]\right.\\
 &\left.\qquad-\left[\mu_0^*(\bm{X}_i)-\hat{\mu}^{(-k,-k')}_0(\bm{X}_i)\right]\left[\frac{1-A_i}{1-\pi^*(\bm{X}_i)}-1\right]\right\},\\
 B_{3}&=\frac{1}{\tilde{n}}\sum_{i \in I_k}\hat{{\bm{\beta}}}^{(-k)^\top }{\bm{D}_i}
 \left\{A_i\left[\mu_1^*(\bm{X}_i)-\hat{\mu}^{(-k,-k')}_1(\bm{X}_i)\right]\left[\frac{1}{\hat{\pi}^{(-k,-k')}(\bm{X}_i)}-\frac{1}{\pi^*(\bm{X}_i)}\right]\right.\\
 &\left.\qquad-(1-A_i)\left[\mu_0^*(\bm{X}_i)-\hat{\mu}^{(-k,-k')}_0(\bm{X}_i)\right]\left[\frac{1}{1-\hat{\pi}^{(-k,-k')}(\bm{X}_i)}-\frac{1}{1-\pi^*(\bm{X}_i)}\right]\right\}.
\end{align*}

When $\mu_a^*(\cdot)=\mu_a(\cdot)$ for each $a\in\{0,1\}$, by iterative expectation formula, one can derive $\mathbb{E}_{\bm{X}}(B_{1})=0$. Further 
\begin{align*}
 &\mathbb{E}_{\bm{X}}\left\{\hat{{\bm{\beta}}}^{(-k)^\top }\bm{D}
 \left[Y(1)-\mu_1^*(\bm{X})\right]\left[\frac{A}{\hat{\pi}^{(-k,-k')}(\bm{X})}-\frac{A}{\pi^*(\bm{X})}\right]\right\}^2\\
 &\quad\leq \sqrt{\mathbb{E}_{\bm{X}}\left[\hat{{\bm{\beta}}}^{(-k)^\top }\bm{D}\right]^4}\left\{\mathbb{E}_{\bm{X}}\left\{A\left[Y(1)-\mu_1^*(\bm{X})\right]\right\}^8\right\}^{\frac{1}{4}}\left\{\mathbb{E}_{\bm{X}}\left[\frac{1}{\hat{\pi}^{(-k,-k')}(\bm{X})}-\frac{1}{\pi^*(\bm{X})}\right]^8\right\}^{\frac{1}{4}}\\
 &\quad\overset{(i)}{=}O_p\left(\frac{s_{{\bm{\gamma}}}\log(d)}{N}\right)=o_p(1),
\end{align*}
where (i) follows from \eqref{40}, Assumption \ref{a5}, and Lemma \ref{l4}.
Then according to Chebyshev's inequality,
\begin{equation*}
 \frac{1}{\tilde{n}}\sum_{i \in I_k}\hat{{\bm{\beta}}}^{(-k)^\top }{\bm{D}_i}
 \left\{\left[Y_i(1)-\mu_1^*(\bm{X}_i)\right]\left[\frac{A_i}{\hat{\pi}^{(-k,-k')}(\bm{X}_i)}-\frac{A_i}{\pi^*(\bm{X}_i)}\right]\right\}=o_p\left(n^{-\frac{1}{2}}\right).
\end{equation*}
Similarly,
\begin{align*}
 &\frac{1}{\tilde{n}}\sum_{i \in I_k}\hat{{\bm{\beta}}}^{(-k)^\top }{\bm{D}_i}\left\{(1-A_i)\left[Y_i(0)-\mu_0^*(\bm{X}_i)\right]\left[\frac{1}{1-\hat{\pi}^{(-k,-k')}(\bm{X}_i)}-\frac{1}{1-\pi^*(\bm{X}_i)}\right]\right\}\\\
 &\qquad=o_p\left(n^{-\frac{1}{2}}\right),
\end{align*}
which means that $B_{1}=o_p(n^{-\frac{1}{2}})$ if $\mu_a^*(\cdot)=\mu_a(\cdot)$ for each $a\in\{0,1\}$. 

When $\mu_a^*(\cdot)\neq\mu_a(\cdot)$ with some $a\in\{0,1\}$, by H\"older's inequality
\begin{align*}
 &\mathbb{E}_{\bm{X}}\left|\hat{{\bm{\beta}}}^{(-k)^\top }\bm{D}
 \left[Y(1)-\mu_1^*(\bm{X})\right]\left[\frac{A}{\hat{\pi}^{(-k,-k')}(\bm{X})}-\frac{A}{\pi^*(\bm{X})}\right]\right|\\
 &\quad\leq \sqrt{\mathbb{E}_{\bm{X}}\left(\hat{{\bm{\beta}}}^{(-k)^\top }\bm{D}\right)^2}\left\{\mathbb{E}_{\bm{X}}\left[A\left(Y(1)-\mu_1^*(\bm{X})\right)\right]^4\right\}^{\frac{1}{4}}\left\{\mathbb{E}_{\bm{X}}\left[\frac{1}{\hat{\pi}^{(-k,-k')}(\bm{X})}-\frac{1}{\pi^*(\bm{X})}\right]^4\right\}^{\frac{1}{4}}\\
 &\quad\overset{(i)}{=}O_p\left(\sqrt{\frac{s_{{\bm{\gamma}}}\log(d)}{N}}\right),
\end{align*}
where (i) holds from \eqref{40}, Lemma \ref{l4}, and Assumption \ref{a5}. Besides,
\begin{align*}
 &\mathbb{E}_{\bm{X}}\left|\hat{{\bm{\beta}}}^{(-k)^\top }\bm{D}
 \left[Y(0)-\mu_0^*(\bm{X})\right]\left[\frac{1-A}{1-\hat{\pi}^{(-k,-k')}(\bm{X})}-\frac{1-A}{1-\pi^*(\bm{X})}\right]\right|\\
 &\quad\leq \sqrt{\mathbb{E}_{\bm{X}}\left(\hat{{\bm{\beta}}}^{(-k)^\top }\bm{D}\right)^2}\left\{\mathbb{E}_{\bm{X}}\left[Y(0)-\mu_0^*(\bm{X})\right]^4\right\}^{\frac{1}{4}}\left\{\mathbb{E}_{\bm{X}}\left[\frac{1-A}{1-\hat{\pi}^{(-k,-k')}(\bm{X})}-\frac{1-A}{1-\pi^*(\bm{X})}\right]^4\right\}^{\frac{1}{4}}\\
 &\quad=O_p\left(\sqrt{\frac{s_{{\bm{\gamma}}}\log(d)}{N}}\right).
\end{align*}
By Markov's inequality, $B_1=O_p(\sqrt{{s_{{\bm{\gamma}}}\log(d)/N}})$ if $\mu_a^*(\cdot)\neq\mu_a(\cdot)$ with some $a\in\{0,1\}$. 

To sum up,
$$B_1=o_p\left(n^{-\frac{1}{2}}\right)+O_p\left(\sqrt{\frac{s_{{\bm{\gamma}}}\log(d)}{N}}\left(\mathbbm{1}_{\mu_1^*(\cdot)\neq\mu_1(\cdot)}+\mathbbm{1}_{\mu_0^*(\cdot)\neq\mu_0(\cdot)}\right)\right).$$

Now turn to prove the convergence rate of \(B_{2}\). When $\pi^*(\cdot)=\pi(\cdot)$, by iterative expectation formula, $\mathbb{E}_{\bm{X}}(B_{2})=0$. Observe that
\begin{align*}
 &\mathbb{E}_{\bm{X}}\left\{\hat{{\bm{\beta}}}^{(-k)^\top }\bm{D}
 \left[\mu_1^*(\bm{X})-\hat{\mu}^{(-k,-k')}_1(\bm{X})\right]\left[\frac{A}{\pi^*(\bm{X})}-1\right]\right\}^2\\
 &\quad\leq \sqrt{\mathbb{E}_{\bm{X}}\left[\hat{{\bm{\beta}}}^{(-k)^\top }\bm{D}\right]^4}\left\{\mathbb{E}_{\bm{X}}\left[\mu_1^*(\bm{X})-\hat{\mu}^{(-k,-k')}_1(\bm{X})\right]^8\right\}^{\frac{1}{4}}\left\{\mathbb{E}_{\bm{X}}\left[\frac{A}{\pi^*(\bm{X})}-1\right]^8\right\}^{\frac{1}{4}}\\
 &\quad=O_p\left(\frac{s_{{\bm{\alpha}}
 }\log(d)}{n}\right)=o_p(1)
\end{align*}
and
\begin{align*}
 &\mathbb{E}_{\bm{X}}\left\{\hat{{\bm{\beta}}}^{(-k)^\top }\bm{D}
 \left[\mu_0^*(\bm{X})-\hat{\mu}^{(-k,-k')}_0(\bm{X})\right]\left[\frac{1-A}{1-\pi^*(\bm{X})}\right]\right\}^2\\
 &\quad\leq \sqrt{\mathbb{E}_{\bm{X}}\left[\hat{{\bm{\beta}}}^{(-k)^\top }\bm{D}\right]^4}\left\{\mathbb{E}_{\bm{X}}\left[\mu_0^*(\bm{X})-\hat{\mu}^{(-k,-k')}_0(\bm{X})\right]^8\right\}^{\frac{1}{4}}\left\{\mathbb{E}_{\bm{X}}\left[\frac{1-A}{1-\pi^*(\bm{X})}\right]^8\right\}^{\frac{1}{4}}\\
 &\quad=o_p(1).
\end{align*}
By Chebyshev's inequality, \(B_2 = o_p( n^{-1/2})\) if $\pi^*(\cdot)=\pi(\cdot)$. 

When $\pi^*(\cdot)\neq\pi(\cdot)$, observe that
\begin{align*}
 &\mathbb{E}_{\bm{X}}\left|\hat{{\bm{\beta}}}^{(-k)^\top }\bm{D}
 \left[\mu_1^*(\bm{X})-\hat{\mu}^{(-k,-k')}_1(\bm{X})\right]\left[\frac{A}{\pi^*(\bm{X})}-1\right]\right|\\
 &\quad\leq \sqrt{\mathbb{E}_{\bm{X}}\left(\hat{{\bm{\beta}}}^{(-k)^\top }\bm{D}\right)^2}\left\{\mathbb{E}_{\bm{X}}\left[\mu_1^*(\bm{X})-\hat{\mu}^{(-k,-k')}_1(\bm{X})\right]^4\right\}^{\frac{1}{4}}\left\{\mathbb{E}_{\bm{X}}\left[\frac{A}{\pi^*(\bm{X})}-1\right]^4\right\}^{\frac{1}{4}}\\
 &\quad\overset{(i)}{=}O_p\left(\sqrt{\frac{s_{{\bm{\alpha}}}\log(d)}{n}}\right),
\end{align*}
where (i) holds from \eqref{40}, Assumption \ref{a7}, and Lemma \ref{l2}. Besides,
\begin{align*}
 &\mathbb{E}_{\bm{X}}\left|\hat{{\bm{\beta}}}^{(-k)^\top }\bm{D}
 \left[\mu_0^*(\bm{X})-\hat{\mu}^{(-k,-k')}_0(\bm{X})\right]\left[\frac{1-A}{1-\pi^*(\bm{X})}-1\right]\right|\\
 &\quad\leq \sqrt{\mathbb{E}_{\bm{X}}\left(\hat{{\bm{\beta}}}^{(-k)^\top }\bm{D}\right)^2}\left\{\mathbb{E}_{\bm{X}}\left[\mu_0^*(\bm{X})-\hat{\mu}^{(-k,-k')}_0(\bm{X})\right]^4\right\}^{\frac{1}{4}}\left\{\mathbb{E}_{\bm{X}}\left[\frac{1-A}{1-\pi^*(\bm{X})}-1\right]^4\right\}^{\frac{1}{4}}\\
 &\quad\overset{(i)}{=}O_p\left(\sqrt{\frac{s_{{\bm{\alpha}}}\log(d)}{n}}\right).
\end{align*}
By Markov’s inequality, $B_{2}
=O_p(\sqrt{s_{{\bm{\alpha}}}\log(d)/n})$ if $\pi^*(\cdot)\neq\pi(\cdot)$.

To sum up,
$$B_{2}
=o_p\left(n^{-\frac{1}{2}}\right)+O_p\left(\sqrt{\frac{s_{{\bm{\alpha}}}\log(d)}{n}}\,
 \mathbbm{1}_{\pi^*(\cdot)\neq\pi(\cdot)}\right).$$

To bound the product error term \(B_3\), H\"older’s inequality implies
\begin{align*}
 &\frac{1}{\tilde{n}}\sum_{i \in I_k}\left|\hat{{\bm{\beta}}}^{(-k)^\top }{\bm{D}_i}
 \left\{A_i\left[\mu_1^*(\bm{X}_i)-\hat{\mu}^{(-k,-k')}_1(\bm{X}_i)\right]\left[\frac{1}{\hat{\pi}^{(-k,-k')}(\bm{X}_i)}-\frac{1}{\pi^*(\bm{X}_i)}\right]\right\}\right|\\
 &\quad\leq \sqrt{\frac{1}{\tilde{n}}\sum_{i \in I_k}\left(\hat{{\bm{\beta}}}^{(-k)^\top }{\bm{D}_i}\right)^2}
 \sqrt{\frac{1}{\tilde{n}}\sum_{i \in I_k}\left[\mu_1^*(\bm{X}_i)-\hat{\mu}^{(-k,-k')}_1(\bm{X}_i)\right]^2\left[\frac{1}{\hat{\pi}^{(-k,-k')}(\bm{X}_i)}-\frac{1}{\pi^*(\bm{X}_i)}\right]^2}\\
 &\quad\leq \sqrt{\frac{1}{\tilde{n}}\sum_{i \in I_k}\left(\hat{{\bm{\beta}}}^{(-k)^\top }{\bm{D}_i}\right)^2}
 \left\{\frac{1}{\tilde{n}}\sum_{i \in I_k}\left[\mu_1^*(\bm{X}_i)-\hat{\mu}^{(-k,-k')}_1(\bm{X}_i)\right]^4\right\}^{\frac{1}{4}}\\
 &\qquad\cdot
 \left\{\frac{1}{\tilde{n}}\sum_{i \in I_k}\left[\frac{1}{\hat{\pi}^{(-k,-k')}(\bm{X}_i)}-\frac{1}{\pi^*(\bm{X}_i)}\right]^4\right\}^{\frac{1}{4}}\\
 &\quad\overset{(i)}{=}O_p\left(\sqrt{\frac{s_{{\bm{\alpha}}}\log(d)}{{n}}}\sqrt{\frac{s_{{\bm{\gamma}}}\log(d)}{N}}\right),
\end{align*}
where (i) holds from \eqref{40}, Lemma \ref{l2}, Lemma \ref{l4}, and Markov's inequality. Then by H\"older’s inequality,
\begin{align*}
 &\frac{1}{\tilde{n}}\sum_{i \in I_k}\left|\hat{{\bm{\beta}}}^{(-k)^\top }{\bm{D}_i}
 \left\{(1-A_i)\left[\mu_0^*(\bm{X}_i)-\hat{\mu}^{(-k,-k')}_0(\bm{X}_i)\right]\left[\frac{1}{1-\hat{\pi}^{(-k,-k')}(\bm{X}_i)}-\frac{1}{1-\pi^*(\bm{X}_i)}\right]\right\}\right|\\
 &\quad=O_p\left(\sqrt{\frac{s_{{\bm{\alpha}}}\log(d)}{{n}}}\sqrt{\frac{s_{{\bm{\gamma}}}\log(d)}{N}}\right).
\end{align*}

Consequently, the overall convergence rate is
\begin{align*}
 &\frac{1}{\tilde{n}}\sum_{i \in I_k}\hat{{\bm{\beta}}}^{(-k)^\top }{\bm{D}_i} \left[\hat{\varphi}^{(-k,-k')}(Z_i
 )-\varphi^*(Z_i
 )\right]=O_p(R_n).
\end{align*}
Since K is finite, 
\begin{align*}
 &\frac{1}{n}\sum_{k=1}^K\sum_{i \in I_k}\hat{{\bm{\beta}}}^{(-k)^\top }{\bm{D}_i} \left[\hat{\varphi}^{(-k)}(Z_i
 )-\varphi^*(Z_i
 )\right]=O_p(R_n).
\end{align*}

Since $\|{\bm{\beta}}^{*^\top }\bm{D}\|_{\psi_2}=O(\|{\bm{\beta}}^*\|_{2})=O(1)$.
Lemma \ref{l1} implies that for any constant $l\geq 0$,
 \begin{equation}\label{109}
 \mathbb{E}_{\bm{X}}\left|{\bm{\beta}}^{*^\top }\bm{D}\right|^{l}=O(1).
 \end{equation}
Consequently, replacing ${\bm{\beta}}^{*\top}\bm{D}_i$ by $\hat{\bm{\beta}}^{(-k)\top}\bm{D}_i$ in the corresponding moment calculation still yields the same upper bound.
\end{proof}

\begin{proof}[Proof of Lemma \ref{l8}]
Throughout the sequel, the following auxiliary bounds will be employed. Since $\bm{D}$ is sub-Gaussian and $\hat{{\bm{\beta}}}^{(-k)}$ constitutes a consistent estimator of ${\bm{\beta}}^*$, one has
 \begin{equation}\label{102}
 \left\|\left(\hat{{\bm{\beta}}}^{(-k)}-{\bm{\beta}}^*\right)^\top \bm{D}\right\|_{\psi_2}=O\left(\left\|\hat{{\bm{\beta}}}^{(-k)}-{\bm{\beta}}^*\right\|_{2}\right)=o_p(1).
 \end{equation}
By Lemma~\ref{l1}, it follows that for any constant $l>0$,
\begin{equation}\label{103}
 \mathbb{E}_{\bm{X}}\left|\left(\hat{\bm{\beta}}^{(-k)}-{\bm{\beta}}^*\right)^\top \bm{D}\right|^l = o_p(1).
\end{equation}
Furthermore, we have \eqref{109} holds and
\begin{equation}\label{92}
 \mathbb{E}_{\bm{X}}\left|\hat{{\bm{\beta}}}^{(-k)^\top }\bm{D}\right|^l=O_p(1).
\end{equation}

For $a \in [0,1]$,
\begin{align}
 &\left\|A(Y-\hat{\mu}_a^{(-k,-k')}(\bm{X}))\right\|_{\psi_2}\leq \|A\left[Y-\mu_a^*(\bm{X})\right]\|_{\psi_2}+\left\|\mu_a^*(\bm{X})-\hat{\mu}_a^{(-k,-k')}(\bm{X})\right\|_{\psi_2}\nonumber\\
 &\qquad=\|A\left[Y-\mu_a^*(\bm{X})\right]\|_{\psi_2}+\left\|\bm{X}^\top ({\bm{\alpha}}_a^*-\hat{{\bm{\alpha}}}^{(-k,-k')}_a)\right\|_{\psi_2}=O_p(1).\label{bound:psi2-residual}
\end{align}
Then for any constant $c>0$,
\begin{align*}
 &\left[\mathbb{E}_{\bm{X}}\left|\hat{\varphi}^{(-k,-k')}(Z)\right|^{4+c}\right]^{\frac{1}{4+c}}\\
 &\quad\overset{(i)}{\leq} \left(\mathbb{E}_{\bm{X}}\left|\frac{A}{\hat{\pi}^{(-k,-k')}(\bm{X})}\left[Y(1)-\hat{\mu}^{(-k,-k')}_1(\bm{X})\right]\right|^{4+c}\right)^{\frac{1}{4+c}}\\
 &\quad\quad+\left(\mathbb{E}_{\bm{X}}\left|\frac{1-A}{1-\hat{\pi}^{(-k,-k')}(\bm{X})}\left[Y(0)-\hat{\mu}^{(-k,-k')}_0(\bm{X})\right]\right|^{4+c}\right)^{\frac{1}{4+c}}\\
 &\quad\quad+\left(\mathbb{E}_{\bm{X}}\left|\hat{\mu}^{(-k,-k')}_1(\bm{X})\right|^{4+c}\right)^{\frac{1}{4+c}}+\left(\mathbb{E}_{\bm{X}}\left|\hat{\mu}^{(-k,-k')}_0(\bm{X})\right|^{4+c}\right)^{\frac{1}{4+c}}\\
 &\quad\leq \left(\mathbb{E}_{\bm{X}}\left|\frac{1}{\hat{\pi}^{(-k,-k')}(\bm{X})}\right|^{8+2c}\right)^{\frac{1}{8+2c}}\left(\mathbb{E}_{\bm{X}}\left|A\left[Y(1)-\hat{\mu}^{(-k,-k')}_1(\bm{X})\right]\right|^{8+2c}\right)^{\frac{1}{8+2c}}\\
 &\quad\quad+\left(\mathbb{E}_{\bm{X}}\left|\frac{1}{1-\hat{\pi}^{(-k,-k')}(\bm{X})}\right|^{8+2c}\right)^{\frac{1}{8+2c}}\left(\mathbb{E}_{\bm{X}}\left|(1-A)\left[Y(0)-\hat{\mu}^{(-k,-k')}_0(\bm{X})\right]\right|^{8+2c}\right)^{\frac{1}{8+2c}}\\
 &\quad\quad+\left(\mathbb{E}_{\bm{X}}\left|\hat{\mu}^{(-k,-k')}_0(\bm{X})\right|^{4+c}\right)^{\frac{1}{4+c}}+
 \left(\mathbb{E}_{\bm{X}}\left|\hat{\mu}^{(-k,-k')}_1(\bm{X})\right|^{4+c}\right)^{\frac{1}{4+c}}\\
 &\quad\leq \left(\mathbb{E}_{\bm{X}}\left|\frac{1}{\hat{\pi}^{(-k,-k')}(\bm{X})}\right|^{8+2c}\right)^{\frac{1}{8+2c}}\left(\mathbb{E}_{\bm{X}}\left|A\left[Y(1)-\hat{\mu}^{(-k,-k')}_1(\bm{X})\right]\right|^{8+2c}\right)^{\frac{1}{8+2c}}\\
 &\notag\quad\quad+\left(\mathbb{E}_{\bm{X}}\left|\bm{X}^\top \hat{{\bm{\alpha}}}^{(-k,-k')}_1\right|^{4+c}\right)^{\frac{1}{4+c}}\left(\mathbb{E}_{\bm{X}}\left|\frac{1}{1-\hat{\pi}^{(-k,-k')}(\bm{X})}\right|^{8+2c}\right)^{\frac{1}{8+2c}}
 \\&\quad\quad\cdot\left(\mathbb{E}_{\bm{X}}\left|(1-A)\left[Y(0)-\hat{\mu}^{(-k,-k')}_0(\bm{X})\right]\right|^{8+2c}\right)^{\frac{1}{8+2c}}
 +\left(\mathbb{E}_{\bm{X}}\left|\bm{X}^\top \hat{{\bm{\alpha}}}^{(-k,-k')}_0\right|^{4+c}\right)^{\frac{1}{4+c}}\\
 &\quad\overset{(ii)}{=}O_p(1),
\end{align*}
where (i) holds by Minkowski's inequality, (ii) holds by Lemma \ref{l4} and \eqref{bound:psi2-residual}. Therefore, for any constant $K$, we have
\begin{align}\label{93}
\left[\mathbb{E}_{\bm{X}}\left|\hat{\varphi}^{(-k)}(Z)\right|^{4+c}\right]^{\frac{1}{4+c}}=O_p(1).
\end{align}
Applying H\"older's inequality to \eqref{93} yields,
 \begin{equation}\label{122}
 \mathbb{E}_{\bm{X}}\left|\hat{\varphi}^{(-k)}(Z)\right|^2\leq \left(\mathbb{E}_{\bm{X}}\left|\hat{\varphi}^{(-k)}(Z)\right|^{4+c}\right)^{\frac{2}{4+c}}\left(\mathbb{E}_{\bm{X}}\left|1\right|^{\frac{4+c}{2+c}}\right)^{\frac{2+c}{4+c}}=O_p(1).
 \end{equation}
 
 For any constants $v,l\geq0, q>0$,
 \begin{align}\label{104}
 \notag&\mathbb{E}_{\bm{X}}\left|\hat{\varphi}^{(-k)}(Z)\left(\hat{{\bm{\beta}}}^{(-k)^\top }\bm{D}\right)^l\left({\bm{\beta}}^{*^\top }\bm{D}\right)^v\left[\left(\hat{{\bm{\beta}}}^{(-k)}-{\bm{\beta}}^*\right)^\top \bm{D}\right]^q\right|\\
 \notag&\quad\leq \sqrt{\mathbb{E}_{\bm{X}}\left|\hat{\varphi}^{(-k)}(Z)\right|^2}\left\{\mathbb{E}_{\bm{X}}\left|\hat{{\bm{\beta}}}^{(-k)^\top }\bm{D}\right|^{4l}\right\}^{\frac{1}{4}}
 \left\{\mathbb{E}_{\bm{X}}\left|{\bm{\beta}}^{*^\top }\bm{D}\right|^{8v}\right\}^{\frac{1}{8}}
 \left\{\mathbb{E}_{\bm{X}}\left|\left(\hat{{\bm{\beta}}}^{(-k)}-{\bm{\beta}}^*\right)^\top \bm{D}\right|^{8q}\right\}^{\frac{1}{8}}\\
 \notag&\quad\overset{(i)}{=} O_p(1)O_p(1)O(1)o_p(1)\\
 &\quad=o_p(1),
 \end{align}
 where (i) holds from  \eqref{109}, \eqref{103}, \eqref{92}, and 
 \eqref{122}. It follows that
 \begin{equation}\label{105}
 \mathbb{E}_{\bm{X}}\left|\left(\hat{{\bm{\beta}}}^{(-k)^\top }\bm{D}\right)^l\left({\bm{\beta}}^{*^\top }\bm{D}\right)^v\left[\left(\hat{{\bm{\beta}}}^{(-k)}-{\bm{\beta}}^*\right)^\top \bm{D}\right]^q\right|=o_p(1).
 \end{equation}

Chebyshev’s inequality implies that 
 \begin{equation}\label{75}
\hat{{\bm{\beta}}}^{(-k)^\top }\left[\mathbb{E}(\bm{W})-|G_k|^{-1}\sum_{i \in G_k}{\bm{W}_i}\right]=O_p\left(N^{-\frac{1}{2}}\right).
 \end{equation}
Furthermore,
\begin{equation}\label{95}
 \tau = \mathbb{E}\left[\varphi^*(Z)\right]=O(1),
\end{equation}
since \(\|\varphi^*(Z)\|_{\psi_2}=O(1)\). Combined with Lemma \ref{l6}, it follows that $\hat{\tau}_{\mathrm{para}}=\tau+\hat{\tau}_{\mathrm{para}}-\tau=O_p(1).$
Next, for any constants \(b,l,v,q,r\ge 0\) with \(q+r>0\) and \(a\in\{0,1,2\}\), we show that
\begin{align*}
 &\frac{1}{\tilde{n}}\sum_{i \in I_k}(\hat{\epsilon}_i^{(k)})^a\left[\hat{{\bm{\beta}}}^{(-k)^\top }\hat{\bm{D}}_i^{(k)}\right]^b\left({\bm{\beta}}^{*^\top }{\bm{D}_i}\right)^v
 \left[\hat{{\bm{\beta}}}^{(-k)^\top }{\bm{D}_i}\right]^l\left\{\hat{{\bm{\beta}}}^{(-k)^\top }\left[\mathbb{E}(\bm{W})-|G_k|^{-1}\sum_{i \in G_k}{\bm{W}_i}\right]\right\}^q \\ &\quad\cdot\left[\left(\hat{{\bm{\beta}}}^{(-k)}-{\bm{\beta}}^*\right)^\top {\bm{D}_i}\right]^r=o_p(1).
\end{align*}
To this end, first demonstrate that ${\tilde{n}^{-1}}\sum_{i \in I_k}(\hat{\epsilon}_i^{(k)})^{2a}=O_p(1).$
By Minkowski’s inequality, for any constant $b\geq 0$,
\begin{align}\label{141}
 &\notag\left\{\frac{1}{\tilde{n}}\sum_{i \in I_k}\left[\hat{{\bm{\beta}}}^{(-k)^\top }\hat{\bm{D}}_i^{(k)}\right]^b\right\}^{\frac{1}{b}}=\left\{\frac{1}{\tilde{n}}\sum_{i \in I_k}\left[\hat{{\bm{\beta}}}^{(-k)^\top }\left({\bm{D}_i}+\mathbb{E}(\bm{W})-|G_k|^{-1}\sum_{i \in G_k}{\bm{W}_i}\right)\right]^b\right\}^{\frac{1}{b}}\notag\\
 &\quad\leq \left\{\frac{1}{\tilde{n}}\sum_{i \in I_k}\left(\hat{{\bm{\beta}}}^{(-k)^\top }{\bm{D}_i}\right)^b\right\}^{\frac{1}{b}}
 +\hat{{\bm{\beta}}}^{(-k)^\top }\left[\mathbb{E}(\bm{W})-|G_k|^{-1}\sum_{i \in G_k}{\bm{W}_i}\right]\overset{(i)}{=}O_p(1),
\end{align}
where (i) follows from \eqref{92}, \eqref{75}, and Markov’s inequality. By the definition of $\hat{\epsilon}_i^{(k)}$,
\begin{align}\label{117}
 \notag&\left\{\frac{1}{\tilde{n}}\sum_{i \in I_k}(\hat{\epsilon}_i^{(k)})^{2a}\right\}^{\frac{1}{2a}}=\left\{\frac{1}{\tilde{n}}\sum_{i \in I_k}\left[\hat{\varphi}^{(-k)}(Z_i)-\hat{\tau}_{\mathrm{para}}-\hat{{\bm{\beta}}}^{(-k)^\top }\hat{\bm{D}}_i^{(k)}\right]^{2a}\right\}^{\frac{1}{2a}}\notag\\
 &\quad\leq \left\{\frac{1}{\tilde{n}}\sum_{i \in I_k}\left[\hat{\varphi}^{(-k)}(Z_i)\right]^{2a}\right\}^{\frac{1}{2a}}
 +\left\{\frac{1}{\tilde{n}}\sum_{i \in I_k}\left[\hat{{\bm{\beta}}}^{(-k)^\top }\hat{\bm{D}}_i^{(k)}\right]^{2a}\right\}^{\frac{1}{2a}}
 +\hat{\tau}_{\mathrm{para}}\overset{(i)}{=}O_p(1),
\end{align}
where (i) follows from  \eqref{93} and \eqref{141}.
 Then, 
 \begin{align}\label{116}
 \notag&\frac{1}{\tilde{n}}\sum_{i \in I_k}(\hat{\epsilon}_i^{(k)})^a\left(\hat{{\bm{\beta}}}^{(-k)^\top }\hat{\bm{D}}_i^{(k)}\right)^b\left({\bm{\beta}}^{*^\top }{\bm{D}_i}\right)^v
 \left(\hat{{\bm{\beta}}}^{(-k)^\top }{\bm{D}_i}\right)^l\left\{\hat{{\bm{\beta}}}^{(-k)^\top }\left[\mathbb{E}(\bm{W})-|G_k|^{-1}\sum_{i \in G_k}{\bm{W}_i}\right]\right\}^q\\
 &\notag\qquad\cdot\left[\left(\hat{{\bm{\beta}}}^{(-k)}-{\bm{\beta}}^*\right)^\top {\bm{D}_i}\right]^r\\
 \notag&\quad\leq \sqrt{\frac{1}{\tilde{n}}\sum_{i \in I_k}(\hat{\epsilon_i}^{(k)})^{2a}}\left\{\frac{1}{\tilde{n}}\sum_{i \in I_k}\left[\hat{{\bm{\beta}}}^{(-k)^\top }\hat{\bm{D}}_i^{(k)}\right]^{4b}\right\}^{\frac{1}{4}}
 \\
 &\notag\qquad\cdot\left\{\frac{1}{\tilde{n}}\sum_{i \in I_k}[\hat{{\bm{\beta}}}^{(-k)^\top }{\bm{D}_i}]^{8l}\right\}^{\frac{1}{8}}
 \left\{\frac{1}{\tilde{n}}\sum_{i \in I_k}\left(\hat{{\bm{\beta}}}^{(-k)^\top }\left[\mathbb{E}(\bm{W})-|G_k|^{-1}\sum_{i \in G_k}{\bm{W}_i}\right]\right)^{16q}\right\}^{\frac{1}{16}}\\
 \notag&\qquad\cdot \left\{\frac{1}{\tilde{n}}\sum_{i \in I_k}\left[\left(\hat{{\bm{\beta}}}^{(-k)}-{\bm{\beta}}^*\right)^\top {\bm{D}_i}\right]^{32r}\right\}^{\frac{1}{32}} \left\{\frac{1}{\tilde{n}}\sum_{i \in I_k}\left({\bm{\beta}}^{*^\top }{\bm{D}_i}\right)^{32v}\right\}^{\frac{1}{32}}\\
 &\quad\overset{(i)}{=}O_p(1)O_p(1)O_p(1)o_p(1)o_p(1)O_p(1)=o_p(1),
 \end{align}
 where (i) follows from  \eqref{109}, \eqref{103}, \eqref{92}, \eqref{75}, \eqref{141}, and \eqref{117}.
 \end{proof}

\begin{proof}[Proof of Lemma \ref{l9}]
 First, we show that $B$ and $C$ are bounded. From \eqref{109} one obtains
\begin{equation*}
 0 \leq B = \mathbb{E}\left({\bm{\beta}}^{*^\top }\bm{D}\right)^4 - \left[\mathbb{E}\left({\bm{\beta}}^{*^\top }\bm{D}\right)^2\right]^2\leq \mathbb{E}\left({\bm{\beta}}^{*^\top }\bm{D}\right)^4 = O(1).
\end{equation*}
Moreover, by Lemma~\ref{l5} $\mathbb{E}|\epsilon|^{4+2c}=O(1)$ for any constant $c>0$. Hence, by \eqref{109} and H\"older’s inequality, 
\begin{equation}\label{18}
 \mathbb{E}|\epsilon\,{\bm{\beta}}^{*^\top}\bm{D}|^{2+c}
 \le \sqrt{\mathbb{E}|\epsilon|^{4+2c}\;\mathbb{E}\left|{\bm{\beta}}^{*^\top}\bm{D}\right|^{4+2c}}
 =O(1).
\end{equation}
Consequently,
\begin{align*}
 \mathbb{E}\left|2\epsilon\,{\bm{\beta}}^{*^\top }\bm{D}\left[\left({\bm{\beta}}^{*^\top }\bm{D}\right)^2 -\theta_{\mathrm{ETH}}\right]\right|
 &\leq \sqrt{\mathbb{E}\left|2\epsilon\,{\bm{\beta}}^{*^\top }\bm{D}\right|^2}\;\sqrt{\mathbb{E}\left|\left({\bm{\beta}}^{*^\top }\bm{D}\right)^2 -\theta_{\mathrm{ETH}}\right|^2}
 \overset{(i)}{=}O(1),
\end{align*}
where (i) follows from \eqref{18} together with the fact that $B=O(1)$. Hence
\begin{equation*}
 C = \mathbb{E}\left\{2\epsilon\,{\bm{\beta}}^{*^\top }\bm{D}\left[\left({\bm{\beta}}^{*^\top }\bm{D}\right)^2 -\theta_{\mathrm{ETH}}\right]\right\} = O(1).
\end{equation*}

Next,  we establish that $\hat{B}^{(k)} = B + o_p(1)$. By definition of $\hat{B}^{(k)}$,
\begin{align*}
 \hat{B}^{(k)} - B
 &= |G_k|^{-1}\sum_{i \in G_k}\left(\hat{{\bm{\beta}}}^{(-k)^\top }\hat{\bm{D}}_i^{(k)}\right)^4
 - \left(\hat{Q}^{(k)}\right)^2
 - \mathbb{E}\left({\bm{\beta}}^{*^\top }\bm{D}\right)^4 + \theta_{\mathrm{ETH}}^2.
\end{align*}

Since \((a+b)^4 = a^4 + 4a^3b + 6a^2b^2 + 4ab^3 + b^4\), we have
 \begin{align}\label{134}
 \notag&|G_k|^{-1}\sum_{i \in G_k}\left(\hat{{\bm{\beta}}}^{(-k)^\top }\hat{\bm{D}}_i^{(k)}\right)^4=|G_k|^{-1}\sum_{i \in G_k}\left\{\hat{{\bm{\beta}}}^{(-k)^\top }\left[{\bm{D}_i}+\mathbb{E}(\bm{W})-|G_k|^{-1}\sum_{i \in G_k}{\bm{W}_i}\right]\right\}^4\\
 \notag&\qquad=|G_k|^{-1}\sum_{i \in G_k}\left(\hat{{\bm{\beta}}}^{(-k)^\top }{\bm{D}_i}\right)^4+\left\{\hat{{\bm{\beta}}}^{(-k)^\top }\left[\mathbb{E}(\bm{W})-|G_k|^{-1}\sum_{i \in G_k}{\bm{W}_i}\right]\right\}^4\\
 &\notag\qquad\qquad+\left\{\hat{{\bm{\beta}}}^{(-k)^\top }\left[\mathbb{E}(\bm{W})-|G_k|^{-1}\sum_{i \in G_k}{\bm{W}_i}\right]\right\}4|G_k|^{-1}\sum_{i \in G_k}\left(\hat{{\bm{\beta}}}^{(-k)^\top }{\bm{D}_i}\right)^3\\
 &\notag \qquad\qquad+4|G_k|^{-1}\sum_{i \in G_k}\left(\hat{{\bm{\beta}}}^{(-k)^\top }{\bm{D}_i}\right)\left\{\hat{{\bm{\beta}}}^{(-k)^\top }\left[\mathbb{E}(\bm{W})-|G_k|^{-1}\sum_{i \in G_k}{\bm{W}_i}\right]\right\}^3\\
 \notag&\qquad\qquad+6\left\{\hat{{\bm{\beta}}}^{(-k)^\top }\left[\mathbb{E}(\bm{W})-|G_k|^{-1}\sum_{i \in G_k}{\bm{W}_i}\right]\right\}^2|G_k|^{-1}\sum_{i \in G_k}\left(\hat{{\bm{\beta}}}^{(-k)^\top }{\bm{D}_i}\right)^2\\
 &\qquad\overset{(i)}{=}|G_k|^{-1}\sum_{i \in G_k}\left(\hat{{\bm{\beta}}}^{(-k)^\top }{\bm{D}_i}\right)^4+o_p(1),
 \end{align}
 where (i) holds from Lemma \ref{l8}.
 Besides,
 \begin{align}
 \notag&|G_k|^{-1}\sum_{i \in G_k}\left(\hat{{\bm{\beta}}}^{(-k)^\top }{\bm{D}_i}\right)^4=|G_k|^{-1}\sum_{i \in G_k}\left[{\bm{\beta}}^{*^\top }{\bm{D}_i}+\left(\hat{{\bm{\beta}}}^{(-k)}-{\bm{\beta}}^*\right)^\top {\bm{D}_i}\right]^4\\
 \notag&\qquad=|G_k|^{-1}\sum_{i \in G_k}\left({\bm{\beta}}^{*^\top }{\bm{D}_i}\right)^4+4|G_k|^{-1}\sum_{i \in G_k}\left({\bm{\beta}}^{*^\top }{\bm{D}_i}\right)^3\left[\left(\hat{{\bm{\beta}}}^{(-k)}-{\bm{\beta}}^*\right)^\top {\bm{D}_i}\right]\\
 &\notag\qquad\qquad+6|G_k|^{-1}\sum_{i \in G_k}\left({\bm{\beta}}^{*^\top }{\bm{D}_i}\right)^2\left[\left(\hat{{\bm{\beta}}}^{(-k)}-{\bm{\beta}}^*\right)^\top {\bm{D}_i}\right]^2\\
 &\notag\qquad\qquad+4|G_k|^{-1}\sum_{i \in G_k}\left({\bm{\beta}}^{*^\top }{\bm{D}_i}\right)\left[\left(\hat{{\bm{\beta}}}^{(-k)}-{\bm{\beta}}^*\right)^\top {\bm{D}_i}\right]^3+|G_k|^{-1}\sum_{i \in G_k}\left[\hat{{\bm{\beta}}}^{(-k)}-{\bm{\beta}}^*)^\top {\bm{D}_i}\right]^4\\
 &\qquad\overset{(i)}{=}|G_k|^{-1}\sum_{i \in G_k}\left({\bm{\beta}}^{*^\top }{\bm{D}_i}\right)^4+o_p(1)\overset{(ii)}{=}\mathbb{E}\left({\bm{\beta}}^{*^\top }\bm{D}\right)^4+o_p(1),\label{90}
 \end{align}
 where (i) holds from Lemma \ref{l8} and (ii) holds from Markov's inequality. 
 
 Moreover, by \eqref{11} and \eqref{75},
\begin{align}\label{44}
 &\notag|G_k|^{-1}\sum_{i \in G_k}\left(\hat{{\bm{\beta}}}^{(-k)^\top }\hat {\bm{D}}_i^{(k)}\right)^2\\
 &\notag\qquad=|G_k|^{-1}\sum_{i \in G_k}\left(\hat{{\bm{\beta}}}^{(-k)^\top }{\bm{D}_i}\right)^2+\left(\hat{{\bm{\beta}}}^{(-k)^\top }\left[|G_k|^{-1}\sum_{i \in G_k}{\bm{W}_i}-\mathbb{E}(\bm{W})\right]\right)^2\\
 &\notag\qquad\qquad-2\hat{{\bm{\beta}}}^{(-k)^\top }\left[|G_k|^{-1}\sum_{i \in G_k}{\bm{W}_i}-\mathbb{E}(\bm{W})\right]|G_k|^{-1}\sum_{i \in G_k}\hat{{\bm{\beta}}}^{(-k)^\top }{\bm{D}_i}\\
 &\qquad=|G_k|^{-1}\sum_{i \in G_k}\left(\hat{{\bm{\beta}}}^{(-k)^\top }{\bm{D}_i}\right)^2+O_p\left(N^{-1}\right).
\end{align}
Meanwhile, by Markov's inequality,
\begin{align}
 &|G_k|^{-1}\sum_{i \in G_k}\left(\left(\hat{{\bm{\beta}}}^{(-k)}-{\bm{\beta}}^*\right)^\top {\bm{D}_i}\right)^2=O_p\left(\mathbb{E}_{\bm{X}}\left[\left(\hat{{\bm{\beta}}}^{(-k)}-{\bm{\beta}}^*\right)^\top \bm{D}\right]^2\right)\nonumber\\
 &\qquad=O_p\left(\left\|\hat{{\bm{\beta}}}^{(-k)}-{\bm{\beta}}^*\right\|_{2}^2\right).\label{27}
\end{align}
By Chebyshev's inequality,
\begin{align}
 &\notag |G_k|^{-1}\sum_{i \in G_k}\left(\hat{{\bm{\beta}}}^{(-k)}-{\bm{\beta}}^*\right)^\top {\bm{D}_i}{\bm{D}_i}^\top {\bm{\beta}}^*\\
 &\quad=\mathbb{E}_{\bm{X}} \left[\left(\hat{{\bm{\beta}}}^{(-k)}-{\bm{\beta}}^*\right)^\top \bm{D}\bm{D}^\top {\bm{\beta}}^*\right]+O_p\left(\left\|\hat{{\bm{\beta}}}^{(-k)}-{\bm{\beta}}^*\right\|_{2}N^{-\frac{1}{2}}\right).\label{19}
\end{align}
Therefore, 
\begin{align}\label{52}
 \notag&|G_k|^{-1}\sum_{i \in G_k}\left(\hat{{\bm{\beta}}}^{(-k)^\top }{\bm{D}_i}\right)^2\\
 &\quad\notag=|G_k|^{-1}\sum_{i \in G_k}\left({\bm{\beta}}^{*^\top }{\bm{D}_i}\right)^2+|G_k|^{-1}\sum_{i \in G_k}\left(\left(\hat{{\bm{\beta}}}^{(-k)}-{\bm{\beta}}^*\right)^\top {\bm{D}_i}\right)^2\\
 &\notag\qquad+2|G_k|^{-1}\sum_{i \in G_k}\left(\hat{{\bm{\beta}}}^{(-k)}-{\bm{\beta}}^*\right)^\top {\bm{D}_i}{\bm{D}_i}^\top {\bm{\beta}}^*\\
 \notag&\quad=|G_k|^{-1}\sum_{i \in G_k}\left({\bm{\beta}}^{*^\top }{\bm{D}_i}\right)^2+2\mathbb{E}_{\bm{X}}\left[\left(\hat{{\bm{\beta}}}^{(-k)}-{\bm{\beta}}^*\right)^\top \bm{D}\bm{D}^\top {\bm{\beta}}^*\right]\\
 &\notag\qquad+O_p\left(\left\|\hat{{\bm{\beta}}}^{(-k)}-{\bm{\beta}}^*\right\|_{2}^2+\left\|\hat{{\bm{\beta}}}^{(-k)}-{\bm{\beta}}^*\right\|_{2}N^{-\frac{1}{2}}\right)\\
 &\quad\overset{(i)}{=}|G_k|^{-1}\sum_{i \in G_k}\left({\bm{\beta}}^{*^\top }{\bm{D}_i}\right)^2+2\mathbb{E}_{\bm{X}}\left[\left(\hat{{\bm{\beta}}}^{(-k)}-{\bm{\beta}}^*\right)^\top \bm{D}\bm{D}^\top {\bm{\beta}}^*\right]+o_p({n}^{-\frac{1}{2}}),
\end{align}
where (i) holds when 
\begin{equation}\label{39}
 \left\|\hat{{\bm{\beta}}}^{(-k)}-{\bm{\beta}}^*\right\|_{2}=o_p({n}^{-\frac{1}{4}}).
\end{equation}

Hence,
\begin{align*}
 &|G_k|^{-1}\sum_{i \in G_k}\left(\hat{{\bm{\beta}}}^{(-k)^\top }\hat{\bm{D}}_i^{(k)}\right)^2
=
 |G_k|^{-1}\sum_{i \in G_k}\left(\hat{{\bm{\beta}}}^{(-k)^\top }{\bm{D}_i}\right)^2 + o_p(1)\\
 &\qquad\overset{(i)}{=}
 |G_k|^{-1}\sum_{i \in G_k}\left({\bm{\beta}}^{*^\top }{\bm{D}_i}\right)^2
 + 2\mathbb{E}_{\bm{X}}\left[\left(\hat{{\bm{\beta}}}^{(-k)}-{\bm{\beta}}^*\right)^\top \bm{D}\bm{D}^\top {\bm{\beta}}^*\right] + o_p(1)\\
 &\qquad\overset{(ii)}{=}
 \theta_{\mathrm{ETH}} + 2\mathbb{E}_{\bm{X}}\left[\left(\hat{{\bm{\beta}}}^{(-k)}-{\bm{\beta}}^*\right)^\top \bm{D} \bm{D}^\top {\bm{\beta}}^*\right] + o_p(1)\overset{(iii)}{=}
 \theta_{\mathrm{ETH}} + o_p(1),
\end{align*}
where (i) holds since \eqref{52}, (ii) holds by Chebyshev's inequality, and (iii) holds since 
\[
\mathbb{E}_{\bm{X}}|\left(\hat{{\bm{\beta}}}^{(-k)}-{\bm{\beta}}^*\right)^\top \bm{D} \bm{D}^\top {\bm{\beta}}^*|
\;\le\;\sqrt{\mathbb{E}_{\bm{X}}\left|\left(\hat{{\bm{\beta}}}^{(-k)}-{\bm{\beta}}^*\right)^\top \bm{D}\right|^2}
 \,\sqrt{\mathbb{E}_{\bm{X}}\left|\bm{D}^\top {\bm{\beta}}^*\right|^2}
\;\overset{(iv)}{=}o_p(1),
\]
with (iv) by \eqref{103} and \eqref{109}. Hence
\begin{equation}\label{173}
\hat{Q}^{(k)}=\theta_{\mathrm{ETH}}+o_p(1).
\end{equation} 
 
Furthermore, by \eqref{109} and Lemma \ref{l5},
\begin{equation}\label{57}
 \theta_{\mathrm{ETH}}=\mathbb{E}\left({\bm{\beta}}^{*^\top}\bm{D}\right)^2
=O(\|{\bm{\beta}}^*\|_2^2)
 =O(1).
\end{equation}
It follows that
 \begin{equation}\label{78}
 \left(\hat{Q}^{(k)}\right)^2-\theta_{\mathrm{ETH}}^2=(\hat{Q}^{(k)}+\theta_{\mathrm{ETH}})(\hat{Q}^{(k)}-\theta_{\mathrm{ETH}})\overset{(i)}{=}O_p(1)o_p(1)=o_p(1),
 \end{equation}
 where (i) holds from \eqref{173} and \eqref{57}.
 According to \eqref{134}, \eqref{90}, and \eqref{78},
 \begin{equation}\label{84}
 \hat{B}^{(k)}=B+o_p(1).
 \end{equation}

Now, we establish the consistency of $\hat{C}^{(k)}$. By Lemma \ref{l8},
 \begin{equation}\label{124}
 \hat{{\bm{\beta}}}^{(-k)^\top }\left[\mathbb{E}(\bm{W})-|G_k|^{-1}\sum_{i \in G_k}{\bm{W}_i}\right]\frac{1}{\tilde{n}}\sum_{i \in I_k}2\hat{\epsilon_i}^{(k)}\left[\left({\bm{\beta}}^{*^\top }{\bm{D}_i}\right)^2-\theta_{\mathrm{ETH}}\right]=o_p(1).
 \end{equation}
Hence,
\begin{align}\label{83}
 \notag\hat{C}^{(k)}&=\frac{1}{\tilde{n}}\sum_{i \in I_k}\left(2\hat{\epsilon_i}^{(k)}\hat{{\bm{\beta}}}^{(-k)^\top }\hat{\bm{D}}_i^{(k)}\right)\left\{\left[\hat{{\bm{\beta}}}^{(-k)^\top }\hat{\bm{D}}_i^{(k)}\right]^2-\theta_{\mathrm{ETH}}+\theta_{\mathrm{ETH}}-\hat{Q}^{(k)}\right\}\\
 \notag&\overset{(i)}{=}\frac{1}{\tilde{n}}\sum_{i \in I_k}\left(2\hat{\epsilon_i}^{(k)}\hat{{\bm{\beta}}}^{(-k)^\top }\hat{\bm{D}}_i^{(k)}\right)\left\{\left[\hat{{\bm{\beta}}}^{(-k)^\top }\hat{\bm{D}}_i^{(k)}\right]^2-\theta_{\mathrm{ETH}}\right\}+o_p(1)\\
 &\notag=\frac{1}{\tilde{n}}\sum_{i \in I_k}\left(2\hat{\epsilon_i}^{(k)}\hat{{\bm{\beta}}}^{(-k)^\top }\hat{\bm{D}}_i^{(k)}\right)\left\{2\hat{{\bm{\beta}}}^{(-k)^\top }{\bm{D}_i}\hat{{\bm{\beta}}}^{(-k)^\top }\left[\mathbb{E}(\bm{W})-|G_k|^{-1}\sum_{i \in G_k}{\bm{W}_i}\right]\right.\\
 \notag&\left.\qquad+\left(\hat{{\bm{\beta}}}^{(-k)^\top }{\bm{D}_i}\right)^2+\left[\hat{{\bm{\beta}}}^{(-k)^\top }\left(\mathbb{E}(\bm{W})-|G_k|^{-1}\sum_{i \in G_k}{\bm{W}_i}\right)\right]^2-\theta_{\mathrm{ETH}}\right\}+o_p(1)\\
 \notag&=\frac{1}{\tilde{n}}\sum_{i \in I_k}\left(2\hat{\epsilon_i}^{(k)}\hat{{\bm{\beta}}}^{(-k)^\top }\hat{\bm{D}}_i^{(k)}\right)\Bigl[2\hat{{\bm{\beta}}}^{(-k)^\top }{\bm{D}_i}\hat{{\bm{\beta}}}^{(-k)^\top }\left[\mathbb{E}(\bm{W})-|G_k|^{-1}\sum_{i \in G_k}{\bm{W}_i}\right]\\
 &\notag\qquad+\left\{\hat{{\bm{\beta}}}^{(-k)^\top }\left[\mathbb{E}(\bm{W})-|G_k|^{-1}\sum_{i \in G_k}{\bm{W}_i}\right]\right\}^2\\
 \notag&\qquad+\left({\bm{\beta}}^{*^\top }{\bm{D}_i}\right)^2+2{\bm{\beta}}^{*^\top }{\bm{D}_i}\left(\hat{{\bm{\beta}}}^{(-k)}-{\bm{\beta}}^*\right)^\top {\bm{D}_i}+\left[\left(\hat{{\bm{\beta}}}^{(-k)}-{\bm{\beta}}^*\right)^\top {\bm{D}_i}\right]^2-\theta_{\mathrm{ETH}}\Bigr]+o_p(1)\\
 \notag&\overset{(ii)}{=}\frac{1}{\tilde{n}}\sum_{i \in I_k}\left[2\hat{\epsilon_i}^{(k)}\hat{{\bm{\beta}}}^{(-k)^\top }\left({\bm{D}_i}+\mathbb{E}(\bm{W})-|G_k|^{-1}\sum_{i \in G_k}{\bm{W}_i}\right)\right]
 \left[\left({\bm{\beta}}^{*^\top }{\bm{D}_i}\right)^2-\theta_{\mathrm{ETH}}\right]+o_p(1)\\
 &\overset{(iii)}{=}\frac{1}{\tilde{n}}\sum_{i \in I_k}\left(2\hat{\epsilon_i}^{(k)}\hat{{\bm{\beta}}}^{(-k)^\top }{\bm{D}_i}\right)
 \left[\left({\bm{\beta}}^{*^\top }{\bm{D}_i}\right)^2-\theta_{\mathrm{ETH}}\right]+o_p(1),
\end{align}
where (i) holds from \eqref{173} and Lemma \ref{l8}, (ii) holds by Lemma \ref{l8},
(iii) holds from \eqref{124}. 
In addition, 
\begin{align}\label{111}
 \notag&\frac{1}{\tilde{n}}\sum_{i \in I_k}\left(\hat{{\bm{\beta}}}^{(-k)^\top }\hat{\bm{D}}_i^{(k)}-{\bm{\beta}}^{*^\top }{\bm{D}_i}\right){\bm{\beta}}^{*^\top }{\bm{D}_i} \left[\left({\bm{\beta}}^{*^\top }{\bm{D}_i}\right)^2-\theta_{\mathrm{ETH}}\right]\\
 \notag&\quad=\frac{1}{\tilde{n}}\sum_{i \in I_k}\left(\hat{{\bm{\beta}}}^{(-k)^\top }\hat{\bm{D}}_i^{(k)}-\hat{{\bm{\beta}}}^{(-k)^\top }{\bm{D}_i}+\hat{{\bm{\beta}}}^{(-k)^\top }{\bm{D}_i}-{\bm{\beta}}^{*^\top }{\bm{D}_i}\right){\bm{\beta}}^{*^\top }{\bm{D}_i} \left[\left({\bm{\beta}}^{*^\top }{\bm{D}_i}\right)^2-\theta_{\mathrm{ETH}}\right]\\
 \notag&\quad=\hat{{\bm{\beta}}}^{(-k)^\top }\left[\mathbb{E}(\bm{W})-|G_k|^{-1}\sum_{i \in G_k}{\bm{W}_i}\right]\frac{1}{\tilde{n}}\sum_{i \in I_k}{\bm{\beta}}^{*^\top }{\bm{D}_i}\left[\left({\bm{\beta}}^{*^\top }{\bm{D}_i}\right)^2-\theta_{\mathrm{ETH}}\right]\\
 \notag&\qquad+\frac{1}{\tilde{n}}\sum_{i \in I_k}\left[\left(\hat{{\bm{\beta}}}^{(-k)}-{\bm{\beta}}^{*}\right)^\top {\bm{D}_i}\right]{\bm{\beta}}^{*^\top }{\bm{D}_i} \left[\left({\bm{\beta}}^{*^\top }{\bm{D}_i}\right)^2-\theta_{\mathrm{ETH}}\right]\\
 &\quad\overset{(i)}{=}o_p(1)O_p(1)+o_p(1)=o_p(1),
\end{align}
where (i) holds from Lemma \ref{l8}.
Thus 
\begin{align}\label{82}
 \notag&\frac{1}{\tilde{n}}\sum_{i \in I_k}\left(2\hat{\epsilon_i}^{(k)}\hat{{\bm{\beta}}}^{(-k)^\top }{\bm{D}_i}\right)
 \left[\left({\bm{\beta}}^{*^\top }{\bm{D}_i}\right)^2-\theta_{\mathrm{ETH}}\right]\\
 \notag&\qquad=\frac{1}{\tilde{n}}\sum_{i \in I_k}\left[2\hat{\epsilon_i}^{(k)}\left({\bm{\beta}}^*+\hat{{\bm{\beta}}}^{(-k)}-{\bm{\beta}}^*\right)^\top {\bm{D}_i}\right]
 \left[\left({\bm{\beta}}^{*^\top }{\bm{D}_i}\right)^2-\theta_{\mathrm{ETH}}\right]\\
 \notag&\qquad=\frac{1}{\tilde{n}}\sum_{i \in I_k}\left(2\hat{\epsilon_i}^{(k)}{\bm{\beta}}^{*^\top }{\bm{D}_i}\right)\left[\left({\bm{\beta}}^{*^\top }{\bm{D}_i}\right)^2-\theta_{\mathrm{ETH}}\right]\\
 &\notag\qquad\qquad+\frac{1}{\tilde{n}}\sum_{i \in I_k}
 \left[2\hat{\epsilon_i}^{(k)}\left(\hat{{\bm{\beta}}}^{(-k)}-{\bm{\beta}}^{*}\right)\right]^\top {\bm{D}_i}\left[\left({\bm{\beta}}^{*^\top }{\bm{D}_i}\right)^2-\theta_{\mathrm{ETH}}\right]\\
 \notag&\qquad\overset{(i)}{=}\frac{1}{\tilde{n}}\sum_{i \in I_k}2\epsilon_i\left({\bm{\beta}}^{*^\top }{\bm{D}_i}\right)\left[\left({\bm{\beta}}^{*^\top }{\bm{D}_i}\right)^2-\theta_{\mathrm{ETH}}\right]+o_p(1)\\
 \notag&\qquad\qquad+\frac{2}{\tilde{n}}\sum_{i \in I_k}
 \left[\hat{\varphi}^{(-k)}(Z_i)-\varphi^*(Z_i)+\theta_{\mathrm{ETH}}-\hat{\tau}_{\mathrm{para}}-\hat{{\bm{\beta}}}^{(-k)^\top }\hat {\bm{D}}_i^{(k)}+{\bm{\beta}}^{*^\top }{\bm{D}_i}\right]{\bm{\beta}}^{*^\top }{\bm{D}_i} \\
 &\notag\qquad\qquad\qquad\cdot\left[\left({\bm{\beta}}^{*^\top }{\bm{D}_i}\right)^2-\theta_{\mathrm{ETH}}\right]\\
 \notag&\qquad\overset{(ii)}{=}\frac{1}{\tilde{n}}\sum_{i \in I_k}2\epsilon_i\left({\bm{\beta}}^{*^\top }{\bm{D}_i}\right)\left[\left({\bm{\beta}}^{*^\top }{\bm{D}_i}\right)^2-\theta_{\mathrm{ETH}}\right]+o_p(1)\\
&\qquad=\mathbb{E}\left[2\epsilon{\bm{\beta}}^{*^\top }\bm{D}\left({\bm{\beta}}^{*^\top }\bm{D}\right)^2-\theta_{\mathrm{ETH}}\right]+o_p(1),
\end{align}
where (i) holds from Lemma \ref{l8}, (ii) holds since \eqref{111}, Lemma \ref{l6}, and Lemma \ref{l7}. Together with \eqref{83}, we have 
$$\hat{C}^{(k)}-C=o_p(1).$$

Next, the convergence of \(\hat{A}^{(k)}\) is established. For any constant \(v\ge0\),
\begin{align}\label{143}
 \notag&\frac{1}{\tilde{n}}\sum_{i \in I_k}\left(\hat{{\bm{\beta}}}^{(-k)^\top }\hat{\bm{D}}_i^{(k)}-{\bm{\beta}}^{*^\top }{\bm{D}_i}\right)^2\left({\bm{\beta}}^{*^\top }{\bm{D}_i}\right)^v\\
 \notag &\quad=\frac{1}{\tilde{n}}\sum_{i \in I_k}(\hat{{\bm{\beta}}}^{(-k)^\top }\hat{\bm{D}}_i^{(k)}-\hat{{\bm{\beta}}}^{(-k)^\top }{\bm{D}_i}+\hat{{\bm{\beta}}}^{(-k)^\top }{\bm{D}_i}-{\bm{\beta}}^{*^\top }{\bm{D}_i})^2\left({\bm{\beta}}^{*^\top }{\bm{D}_i}\right)^v\\
 \notag&\quad=\frac{1}{\tilde{n}}\sum_{i \in I_k}\left\{\left(\hat{{\bm{\beta}}}^{(-k)^\top }\left[\mathbb{E}(\bm{W})-|G_k|^{-1}\sum_{i \in G_k}{\bm{W}_i}\right]\right)^2+
\left[\left(\hat{{\bm{\beta}}}^{(-k)}-{\bm{\beta}}^{*}\right)^\top {\bm{D}_i}\right]^2\right.\\
&\left.\notag\qquad+2\hat{{\bm{\beta}}}^{(-k)^\top }\left[\mathbb{E}(\bm{W})-|G_k|^{-1}\sum_{i \in G_k}{\bm{W}_i}\right]\left(\hat{{\bm{\beta}}}^{(-k)}-{\bm{\beta}}^{*}\right)^\top {\bm{D}_i}\right\}\left({\bm{\beta}}^{*^\top }{\bm{D}_i}\right)^v\\
&\quad\overset{(i)}{=}o_p(1),
\end{align}
where (i) holds from Lemma \ref{l8}.
Additionally,
\begin{equation}\label{144}
 \frac{1}{\tilde{n}}\sum_{i \in I_k}\left[\hat{\varphi}^{(-k)}(Z_i)-\varphi^*(Z_i)\right]^4=o_p(1)
\end{equation}
holds from Lemma \ref{l6} and Markov's inequality. Therefore,
\begin{align}\label{145}
 \notag&\frac{1}{\tilde{n}}\sum_{i \in I_k}\left[2\left(\hat{\epsilon}_i^{(k)}-\epsilon_i\right){\bm{\beta}}^{*^\top }{\bm{D}_i}\right]^2\\
 \notag& =\frac{4}{\tilde{n}}\sum_{i \in I_k}\left[\hat{\varphi}^{(-k)}(Z_i)-\varphi^*(Z_i)+\tau-\hat{\tau}_{\mathrm{para}}-\hat{{\bm{\beta}}}^{(-k)^\top }\hat {\bm{D}}_i^{(k)}+{\bm{\beta}}^{*^\top }{\bm{D}_i}\right]^2\left({\bm{\beta}}^{*^\top }{\bm{D}_i}\right)^2\\
 \notag&=O\left(\frac{4}{\tilde{n}}\sum_{i \in I_k}\left\{\left[\hat{\varphi}^{(-k)}(Z_i)-\varphi^*(Z_i)\right]^2+(\tau-\hat{\tau}_{\mathrm{para}})^2+\left({\bm{\beta}}^{*^\top }{\bm{D}_i}-\hat{{\bm{\beta}}}^{(-k)^\top }\hat {\bm{D}}_i^{(k)}\right)^2\right\}\left({\bm{\beta}}^{*^\top }{\bm{D}_i}\right)^2\right)\\
 &\overset{(i)}{=}O\left(\left\{\frac{4}{\tilde{n}}\sum_{i \in I_k}\left[\hat{\varphi}^{(-k)}(Z_i)-\varphi^*(Z_i)\right]^4\right\}^{\frac{1}{2}} \left\{\frac{4}{\tilde{n}}\sum_{i \in I_k}\left({\bm{\beta}}^{*^\top }{\bm{D}_i}\right)^4\right\}^{\frac{1}{2}}\right)+o_p(1)\overset{(ii)}{=}o_p(1),
\end{align}
where (i) follows from Lemma~\ref{l6}, \eqref{109}, \eqref{143}, and Markov’s inequality, (ii) follows from \eqref{109}, \eqref{144}, and Markov’s inequality. Furthermore,
\begin{align}\label{146}
 &\notag\frac{2}{\tilde{n}}\sum_{i \in I_k}\left(2{\bm{\beta}}^{*^\top }{\bm{D}_i}\right)^2\epsilon_i\left(\hat{\epsilon}_i^{(k)}-\epsilon_i\right)\overset{(i)}{\leq} \left\{\frac{2}{\tilde{n}}\sum_{i \in I_k}\left(2{\bm{\beta}}^{*^\top }{\bm{D}_i}\right)^8\right\}^{\frac{1}{4}}\left(\frac{2}{\tilde{n}}\sum_{i \in I_k}\epsilon_i^4\right)^{\frac{1}{4}}
 \left\{\frac{2}{\tilde{n}}\sum_{i \in I_k}\left(\hat{\epsilon}_i^{(k)}-\epsilon_i\right)^2\right\}^{\frac{1}{2}}\\
 \notag &\quad\overset{(i)}{=}O\Bigl(\left\{\frac{2}{\tilde{n}}\sum_{i \in I_k}\left(2{\bm{\beta}}^{*^\top }{\bm{D}_i}\right)^8\right\}^{\frac{1}{4}}\left(\frac{2}{\tilde{n}}\sum_{i \in I_k}\epsilon_i^4\right)^{\frac{1}{4}}
 \left\{\frac{2}{\tilde{n}}\sum_{i \in I_k}\left[\hat{\varphi}^{(-k)}(Z_i)-\varphi^*(Z_i)\right]^2\right.\\
 &\notag\left.\qquad+(\tau-\hat{\tau}_{\mathrm{para}})^2+\left({\bm{\beta}}^{*^\top }{\bm{D}_i}-\hat{{\bm{\beta}}}^{(-k)^\top }\hat {\bm{D}}_i^{(k)}\right)^2\right\}^{\frac{1}{2}}\Bigr)\\
 &\quad\overset{(ii)}{=}O(1)O(1)o_p(1)=o_p(1),
\end{align}
where (i) holds from H\"older's inequality, (ii) holds from \eqref{109}, Lemma \ref{l5}, Lemma \ref{l6}, \eqref{143}, and Markov's inequality.
Thus,
\begin{align}
 \notag\hat{A}^{(k)}&=\frac{1}{\tilde{n}}\sum_{i \in I_k}\left(2\hat{\epsilon_i}^{(k)}\hat{{\bm{\beta}}}^{(-k)^\top }\hat{\bm{D}}_i^{(k)}\right)^2\\
 \notag&=\frac{1}{\tilde{n}}\sum_{i \in I_k}\left[2\hat{\epsilon_i}^{(k)}\hat{{\bm{\beta}}}^{(-k)^\top }\left({\bm{D}_i}+\hat{\bm{D}}_i^{(k)}-{\bm{D}_i}\right)\right]^2\\
 \notag&=\frac{1}{\tilde{n}}\sum_{i \in I_k}\left(2\hat{\epsilon_i}^{(k)}\hat{{\bm{\beta}}}^{(-k)^\top }{\bm{D}_i}\right)^2+\left\{\hat{{\bm{\beta}}}^{(-k)^\top }\left[\mathbb{E}(\bm{W})-|G_k|^{-1}\sum_{i \in G_k}{\bm{W}_i}\right]\right\}^2\frac{1}{\tilde{n}}\sum_{i \in I_k}\left(2\hat{\epsilon_i}^{(k)}\right)^2\\
 &\notag\qquad+2\hat{{\bm{\beta}}}^{(-k)^\top }\left[\mathbb{E}(\bm{W})-|G_k|^{-1}\sum_{i \in G_k}{\bm{W}_i}\right]\frac{1}{\tilde{n}}\sum_{i \in I_k}\left(2\hat{\epsilon_i}^{(k)}\right)^2\hat{{\bm{\beta}}}^{(-k)^\top }{\bm{D}_i}\\
 \notag &\overset{(i)}{=}\frac{1}{\tilde{n}}\sum_{i \in I_k}\left(2\hat{\epsilon_i}^{(k)}\hat{{\bm{\beta}}}^{(-k)^\top }{\bm{D}_i}\right)^2+o_p(1)\\
 \notag&=\frac{1}{\tilde{n}}\sum_{i \in I_k}\left[2\hat{\epsilon_i}^{(k)}\left({\bm{\beta}}^*+\hat{{\bm{\beta}}}^{(-k)}-{\bm{\beta}}^*\right)^\top {\bm{D}_i}\right]^2+o_p(1)\\
 \notag&=\frac{1}{\tilde{n}}\sum_{i \in I_k}\left(2\hat{\epsilon_i}^{(k)}{\bm{\beta}}^{*^\top }{\bm{D}_i}\right)^2
 +\frac{1}{\tilde{n}}\sum_{i \in I_k}\left[2\hat{\epsilon_i}^{(k)}\left(\hat{{\bm{\beta}}}^{(-k)}-{\bm{\beta}}^{*}\right)^\top {\bm{D}_i}\right]^2\\
 &\notag\qquad+\frac{2}{\tilde{n}}\sum_{i \in I_k}\left(2\hat{\epsilon_i}^{(k)}\right)^2{\bm{\beta}}^{*^\top }{\bm{D}_i}
 \left(\hat{{\bm{\beta}}}^{(-k)}-{\bm{\beta}}^{*}\right)^\top {\bm{D}_i}+o_p(1)\\
 \notag&\overset{(ii)}{=}\frac{1}{\tilde{n}}\sum_{i \in I_k}\left[2\left(\epsilon_i+\hat{\epsilon_i}^{(k)}-\epsilon_i\right){\bm{\beta}}^{*^\top }{\bm{D}_i}\right]^2+o_p(1)\\
 \notag&=\frac{1}{\tilde{n}}\sum_{i \in I_k}\left(2\epsilon_i{\bm{\beta}}^{*^\top }{\bm{D}_i}\right)^2
 +\frac{1}{\tilde{n}}\sum_{i \in I_k}\left[2\left(\hat{\epsilon}_i^{(k)}-\epsilon_i\right){\bm{\beta}}^{*^\top }{\bm{D}_i}\right]^2
 +\frac{2}{\tilde{n}}\sum_{i \in I_k}\left(2{\bm{\beta}}^{*^\top }{\bm{D}_i}\right)^2\epsilon_i\left(\hat{\epsilon}_i^{(k)}-\epsilon_i\right)\\
 &\overset{(iii)}{=}\frac{1}{\tilde{n}}\sum_{i \in I_k}\left(2\epsilon_i{\bm{\beta}}^{*^\top }{\bm{D}_i}\right)^2+o_p(1),
\end{align}
where (i) and (ii) follow from Lemma \ref{l8}, and (iii) is ensured by \eqref{145} and \eqref{146}. Therefore, we have 
$$\hat{A}^{(k)} = A + o_p(1).$$

Combining all the results above, we obtain
\begin{align*}
 &\hat{w}^{(k)}_U-w_U^*=\left(1+\frac{\hat{C}^{(k)}}{\hat{B}^{(k)}}\right)\frac{1}{N}-\left(1+\frac{C}{B}\right)\frac{1}{N}=\left(\frac{\hat{C}^{(k)}}{\hat{B}^{(k)}}
 -\frac{C}{B}\right)\frac{1}{N}\\
 &\qquad=O\left(\frac{B\left(\hat{C}^{(k)}-C\right)+C\left(B-\hat{B}^{(k)}\right)}{\hat{B}^{(k)}B}\frac{1}{N}\right)\overset{(i)}{=}o_p(N^{-1}),
\end{align*}
where (i) holds since \(\hat{B}^{(k)}-B=o_p(1)\) and \(B>c\) for some constant \(c>0\). Moreover, 
\begin{align*}
 \hat{w}^{(k)}_L-w_L^*=\frac{1}{n}-\frac{1}{n}+\frac{m}{n}w_U^*-\frac{m}{n}\hat{w}^{(k)}_U=o_p\left(\frac{m}{nN}\right).
\end{align*}
\end{proof}

\begin{proof}[Proof of Lemma \ref{l10}]
For each $k\leq K$, define
$$(\hat{\sigma}_{\mathrm{OW}}^{(k)})^2:=\hat{A}^{(k)} 
 + \frac{n\hat{B}^{(k)}}{N} +\frac{2n\hat{C}^{(k)} }{N} - \frac{m(\hat{C}^{(k)})^2}{N\hat{B}^{(k)}}.$$
 Then
\begin{align}\label{147}
 \notag(\hat{\sigma}_{\mathrm{OW}}^{(k)})^2 - \sigma_{\mathrm{OW}}^2
 &= \hat{A}^{(k)} - A
 + \frac{1}{N}\left(\hat{B}^{(k)}n - nB + 2n\hat{C}^{(k)} - 2nC
 - m\frac{(\hat{C}^{(k)})^2}{\hat{B}^{(k)}} + m\frac{C^2}{B}\right)\\
 \notag&= \hat{A}^{(k)} - A
 + \frac{n}{N}(\hat{B}^{(k)} - B)
 + \frac{2n}{N}(\hat{C}^{(k)} - C)\\
 \notag&\qquad
 + \frac{m}{N}
 \frac{C^2(\hat{B}^{(k)} - B)
 + B\,(C + \hat{C}^{(k)})(C - \hat{C}^{(k)})}
 {B\,\hat{B}^{(k)}}\\
 &\overset{(i)}{=} o_p(1) + o_p\left(\frac{n}{N} + \frac{m}{N}\right)=o_p(1),
\end{align}
where (i) follows from Lemma \ref{l9}. For any finite \(K\), it holds that $ \hat{\sigma}_{\mathrm{OW}}^2 = \sigma_{\mathrm{OW}}^2 + o_p(1).$
\end{proof}

\section{Proof of the main results}
\begin{proof}{Proof of Theorem \ref{t4}}.
Define 
 \begin{align*}
 \hat{\theta}^{(k)}_{\mathrm{TTH}} &:= |G_k|^{-1} \sum_{i \in G_k} \left[ \hat{h}^{(-k)}(\bm{X}_i) \right]^2 + 2\tilde{n}^{-1} \sum_{i \in I_k} \hat{h}^{(-k)}(\bm{X}_i) \left[\hat{\varphi}^{(-k)}(Z_i) - \hat{\tau} - \hat{h}^{(-k)}(\bm{X}_i) \right],\\
 \tilde{\theta}_{\mathrm{TTH}}^{(k)}&:=|G_k|^{-1} \sum_{i \in G_k} \left[ \hat{\tau}^{(-k)}(\bm{X}_i)-\nu^{(-k)} \right]^2 \\
 &\quad+ 2\tilde{n}^{-1} \sum_{i \in I_k} \left[\hat{\tau}^{(-k)}(\bm{X}_i)-\nu^{(-k)}\right] \left[\varphi^*(Z_i) - \tau - \hat{\tau}^{(-k)}(\bm{X}_i)+\nu^{(-k)} \right],\\
 \check{\theta}_{\mathrm{TTH}}^{(k)}&:=|G_k|^{-1} \sum_{i \in G_k} \left[ h(\bm{X}_i) \right]^2 + 2\tilde{n}^{-1} \sum_{i \in I_k} h(\bm{X}_i)\xi_i,
 \end{align*}
 where $\hat{h}^{(-k)}(\bm{X}_i) = \hat{\tau}^{(-k)}(\bm{X}_i) - |G_k|^{-1} \sum_{i \in G_k} \hat{\tau}^{(-k)}(\bm{X}_i)$, $\nu^{(-k)}=\mathbb{E}_{\bm{X}}[\hat{\tau}^{(-k)}(\bm{X})]$, and $h(\bm{X}_i)=\tau(\bm{X}_i)-\tau$.

 Step 1. We first show that $\hat{\theta}^{(k)}_{\mathrm{TTH}}=\tilde{\theta}_{\mathrm{TTH}}^{(k)}+o_p(n^{-\frac{1}{2}}).$ 
 
 Let $\Delta_1=|G_k|^{-1} \sum_{i \in G_k} [ \hat{\tau}^{(-k)}(\bm{X}_i)-\nu^{(-k)} ]$ and $\delta_i=\hat{\tau}^{(-k)}(\bm{X}_i)-\nu^{(-k)}-h(\bm{X}_i)$. Then 
 \begin{equation}\label{150}
 \Delta_1\overset{(i)}{=}O_p\left(N^{-\frac{1}{2}}\right),
 \end{equation}
 where (i) holds from \eqref{172}.
 Consider the decomposition
 \begin{align*}
 &\hat{\theta}^{(k)}_{\mathrm{TTH}} = |G_k|^{-1} \sum_{i \in G_k} \left[ \hat{h}^{(-k)}(\bm{X}_i) \right]^2 + 2\tilde{n}^{-1} \sum_{i \in I_k} \hat{h}^{(-k)}(\bm{X}_i) \left[\hat{\varphi}^{(-k)}(Z_i) - \hat{\tau} - \hat{h}^{(-k)}(\bm{X}_i) \right]\\
 &\qquad =\underbrace{ |G_k|^{-1} \sum_{i \in G_k} \left\{ \hat{\tau}^{(-k)}(\bm{X}_i)-\nu^{(-k)} - \Delta_1 \right\}^2}_{M_1} \\
 &\qquad\qquad+ \underbrace{2\tilde{n}^{-1} \sum_{i \in I_k} \left\{\hat{\tau}^{(-k)}(\bm{X}_i)-\nu^{(-k)} - \Delta_1 \right\} \left[\hat{\varphi}^{(-k)}(Z_i) - \hat{\tau} - \hat{h}^{(-k)}(\bm{X}_i) \right]}_{M_2}.
 \end{align*}
 Then 
 \begin{align*}
 M_1&=|G_k|^{-1} \sum_{i \in G_k} \left[ \hat{\tau}^{(-k)}(\bm{X}_i)-\nu^{(-k)}\right]^2-\Delta_1^2=|G_k|^{-1} \sum_{i \in G_k} \left[ \hat{\tau}^{(-k)}(\bm{X}_i)-\nu^{(-k)}\right]^2+O_p(N^{-1}),\\
 M_2&=2\tilde{n}^{-1} \sum_{i \in I_k}\left[\hat{\tau}^{(-k)}(\bm{X}_i)-\nu^{(-k)}\right]\left[\varphi^*(Z_i)-\tau-\hat{\tau}^{(-k)}(\bm{X}_i)+\nu^{(-k)}\right]\\
 &\qquad+2\tilde{n}^{-1} \sum_{i \in I_k}\left[\hat{\tau}^{(-k)}(\bm{X}_i)-\nu^{(-k)}\right]\left[\hat{\varphi}^{(-k)}(Z_i)-\varphi^*(Z_i)+\tau-\hat{\tau}+\Delta_1\right]\\
 &\qquad-2\Delta_1\tilde{n}^{-1}\sum_{i \in I_k}\left[\varphi^*(Z_i)-\tau-\hat{\tau}^{(-k)}(\bm{X}_i)+\nu^{(-k)}\right]\\
 &\qquad-2\Delta_1\tilde{n}^{-1}\sum_{i \in I_k}\left[\hat{\varphi}^{(-k)}(Z_i)-\varphi^*(Z_i)+\tau-\hat{\tau}+\Delta_1\right]\\
 &\overset{(i)}{=}2\tilde{n}^{-1}\sum_{i \in I_k}\left[\hat{\tau}^{(-k)}(\bm{X}_i)-\nu^{(-k)}\right]\left[\varphi^*(Z_i)-\tau-\hat{\tau}^{(-k)}(\bm{X}_i)+\nu^{(-k)}\right]+o_p\left(n^{-\frac{1}{2}}\right),
 \end{align*}
where (i) follows from Lemma~\ref{l11}, \eqref{150}, and the assumptions of Theorem \ref{t4}.
 Thus $\hat{\theta}^{(k)}_{\mathrm{TTH}}=\tilde{\theta}_{\mathrm{TTH}}^{(k)}+o_p(n^{-\frac{1}{2}}).$
 
 Step 2. Now we show that $\tilde{\theta}_{\mathrm{TTH}}^{(k)}=\check{\theta}_{\mathrm{TTH}}^{(k)}+o_p(n^{-\frac{1}{2}}).$
 
 Observe that 
 \begin{equation*}
 \tilde{\theta}_{\mathrm{TTH}}^{(k)}=\check{\theta}_{\mathrm{TTH}}^{(k)}
 +|G_k|^{-1} \sum_{i \in G_k}\delta\left[\delta_i+2h(\bm{X}_i)\right]
 +2\tilde{n}^{-1} \sum_{i \in I_k}\delta_i\left[\xi_i-h(\bm{X}_i)-\delta_i\right].
 \end{equation*}
Under Assumption \ref{a4}, $|G_k|^{-1} \sum_{i \in G_k}\delta_i^2=O_p(\mathbb{E}_{\bm{X}}(\delta^2))=o_p(n^{-{1/2}})$.
 Besides, $\mathbb{E}_{\bm{X}}\left[\delta h(\bm{X})\right]^2=O(\mathbb{E}_{\bm{X}}(\delta^2))=o_p(1)$. Together with Chebyshev's inequality, we have
 \begin{equation}
 |G_k|^{-1} \sum_{i \in G_k}\delta_ih(\bm{X}_i)=\mathbb{E}_{\bm{X}}[\delta h(\bm{X})]+o_p\left(n^{-\frac{1}{2}}\right).
 \end{equation}
 Moreover, because $\mathbb{E}(\xi \mid \bm{X})=0$, the law of iterated expectations ensures that $\mathbb{E}_{\bm{X}}(\xi\delta)=0$, and $\mathbb{E}_{\bm{X}}(\delta\xi)^2=O(\mathbb{E}_{\bm{X}}(\delta^2))=o_p(1)$, then $ \tilde{n}^{-1} \sum_{i \in I_k}\delta_i\xi_i=o_p(n^{-{1/2}}).$
 Thus,
 \begin{align*}
 &|G_k|^{-1} \sum_{i \in G_k}\delta_i\left[\delta_i+2h(\bm{X}_i)\right]
 =\mathbb{E}_{\bm{X}}(\delta^2)+2\mathbb{E}_{\bm{X}}[\delta h(\bm{X})]+o_p\left(n^{-\frac{1}{2}}\right),\\
 &\tilde{n}^{-1} \sum_{i \in I_k}\delta_i\left[\xi_i-h(\bm{X}_i)-\delta_i\right]
 =-\mathbb{E}_{\bm{X}}\left[\delta h(\bm{X})\right]-\mathbb{E}_{\bm{X}}(\delta^2)+o_p\left(n^{-\frac{1}{2}}\right),
 \end{align*}
 which means that $ \tilde{\theta}_{\mathrm{TTH}}^{(k)}=\check{\theta}_{\mathrm{TTH}}^{(k)}
 +o_p(n^{-\frac{1}{2}}).$
 
 Step 3. We lastly show that ${\sqrt{n}(\check{\theta}_{\mathrm{TTH}}-\theta_{\mathrm{TTH}})/\sigma_{\mathrm{TTH}}}\rightarrow N(0,1).$
 
 Since 
 \begin{equation*}
\check{\theta}_{\mathrm{TTH}}^{(k)}=|G_k|^{-1} \sum_{i \in G_k} \left[ h(\bm{X}_i) \right]^2 + 2\tilde{n}^{-1} \sum_{i \in I_k} h(\bm{X}_i)\xi_i,
 \end{equation*}
 we have
 \begin{align*}
 \check{\theta}_{\mathrm{TTH}}-\theta_{\mathrm{TTH}}&=N^{-1} \sum_{i=1}^{N} \left[ h(\bm{X}_i) \right]^2 + 2n^{-1} \sum_{i=1}^{n} h(\bm{X}_i)\xi_i-\theta_{\mathrm{TTH}}\\
 &=\frac{1}{n}\sum_{i=1}^{N}\left\{2\mathbbm{1}_{i\leq n}\xi_ih(\bm{X}_i)+\frac{n}{N}\left[h(\bm{X}_i)\right]^2-\frac{n}{N}\theta_{\mathrm{TTH}}\right\}.
 \end{align*}
 Besides,
 \begin{equation}\label{155}
 \|\xi\|_{\psi_2}=\|\varphi^*(Z)-\tau(\bm{X})\|_{\psi_2}\leq \|\varphi^*(Z)\|_{\psi_2}+\|\tau(\bm{X})\|_{\psi_2}=O(1).
 \end{equation}
Then $\mathbb{E}_{\bm{X}}|\xi|^c=O(1)$ for any constant $c>0$.
 By Minkovski's inequality, for any $\delta>0$,
\begin{align}\label{151}
 \notag&\left(\mathbb{E}_{\bm{X}}\left|2\xi h(\bm{X})+\left[\sqrt{\frac{n}{N}}h(\bm{X})\right]^2-\frac{n}{N}\theta_{\mathrm{TTH}}\right|^{2+\delta}\right)^{\frac{1}{2+\delta}}\\
 \notag&\quad\leq \left(\mathbb{E}_{\bm{X}}\left|2\xi h(\bm{X})\right|^{2+\delta}\right)^{\frac{1}{2+\delta}}+\left(\mathbb{E}_{\bm{X}}\left|\left[\sqrt{\frac{n}{N}}h(\bm{X})\right]^2-\frac{n}{N}\theta_{\mathrm{TTH}}\right|^{2+\delta}\right)^{\frac{1}{2+\delta}}\\
 \notag&\quad\leq \left(\mathbb{E}_{\bm{X}}\left|2h(\bm{X})\right|^{4+\delta}\right)^{\frac{1}{4+\delta}}
 \left(\mathbb{E}_{\bm{X}}|\xi|^{4+\delta}\right)^{\frac{1}{4+\delta}}
 +\left(\mathbb{E}_{\bm{X}}\left|\left[\sqrt{\frac{n}{N}}h(\bm{X})\right]^2-\frac{n}{N}\theta_{\mathrm{TTH}}\right|^{2+\delta}\right)^{\frac{1}{2+\delta}}\\
 &\quad\overset{(i)}{=}O(1),
\end{align}
where (i) holds by the assumption of Theorem \ref{t4} and \eqref{155}.
Therefore,
\begin{align}\label{54}
 \notag&\sum_{i=1}^{N}\mathbb{E}\left\{\left|2\mathbbm{1}_{i\leq n}\xi_ih(\bm{X}_i)+\left[\sqrt{\frac{n}{N}}h(\bm{X}_i)\right]^2-\frac{n}{N}\theta_{\mathrm{TTH}}\right|\right\}^{2+\delta}\\
 \notag&\qquad=\sum_{i=1}^{n}\mathbb{E}\left\{\left|2\xi_ih(\bm{X}_i)+\left[\sqrt{\frac{n}{N}}h(\bm{X}_i)\right]^2-\frac{n}{N}\theta_{\mathrm{TTH}}\right|\right\}^{2+\delta}\\
 &\notag\qquad\qquad+\sum_{i=n+1}^{N}\mathbb{E}\left\{\left|\left[\sqrt{\frac{n}{N}}h(\bm{X}_i)\right]^2-\frac{n}{N}\theta_{\mathrm{TTH}}\right|\right\}^{2+\delta}\\
 &\qquad\overset{(i)}{=}O\left\{n+\left[m\left(\frac{n}{N}\right)\right]^{2+\delta}\right\}\overset{(ii)}{=}O(n),
\end{align}
where (i) holds from \eqref{151}, (ii) holds since $m({n/N})^{2+\delta}\leq m({n/N})^{2}\leq n$.
Define
\begin{align*}
 D_{N}^2
 &:=\sum_{i=1}^{N}\Var\left\{2\mathbbm{1}_{i\leq n}\xi_ih(\bm{X}_i)+\left[\sqrt{\frac{n}{N}}h(\bm{X}_i)\right]^2-\frac{n}{N}\theta_{\mathrm{TTH}}\right\},\\
 &=\sum_{i=1}^{n}\Var\left\{2\xi_ih(\bm{X}_i)+\left[\sqrt{\frac{n}{N}}h(\bm{X}_i)\right]^2-\frac{n}{N}\theta_{\mathrm{TTH}}\right\}+\sum_{i=n+1}^{N}\Var\left\{\left[\sqrt{\frac{n}{N}}h(\bm{X}_i)\right]^2-\frac{n}{N}\theta_{\mathrm{TTH}}\right\}\\
 &\geq \sum_{i=1}^{n}\Var\left\{2\xi_ih(\bm{X}_i)+\left[\sqrt{\frac{n}{N}}h(\bm{X}_i)\right]^2-\frac{n}{N}\theta_{\mathrm{TTH}}\right\}=n\sigma_{\mathrm{TTH}}^2.
\end{align*}
Since $\sigma_{\mathrm{TTH}}^2>c$ for some constant $c>0$, for any $\delta>0$,
\begin{equation*}
\frac{\sum_{i=1}^{N}\mathbb{E} |2\mathbbm{1}_{i\leq n}\xi_ih(\bm{X}_i)+\frac{n}{N}\left[h(\bm{X}_i)\right]^2-\frac{n}{N}\theta_{\mathrm{TTH}}|^{2+\delta}}{(D_{N}^2)^{1+\frac{\delta}{2}}}
 =O\left(\frac{n}{(n\sigma_{\mathrm{TTH}}^2)^{1+{\delta/2}}}\right)=o(1).
\end{equation*}
By Lindeberg-Feller central limit theorem and Slutsky's Theorem we have
\begin{equation}
 \frac{\sqrt{n}(\hat{\theta}_{\mathrm{TTH}}-\theta_{\mathrm{TTH}})}{\sigma_{\mathrm{TTH}}}\rightarrow N(0, 1),
\end{equation}
provided that K is finite.
\end{proof}

\begin{proof}[Proof of Theorem \ref{t1}]
Consider the estimator \(\hat{{\bm{\beta}}}^{(-k)}\) for any $k\leq K$. In the following analysis, \(\hat{\mu}_a^{(-k,-k')}(\cdot)\) and \(\hat{\pi}^{(-k,-k')}(\cdot)\) are treated as fixed (i.e., conditioned upon) for all $k'\neq k$. By definition,
\begin{align*}
&\frac{1}{\sum_{k'\neq k}|I_{k'}|}
\sum_{k'\neq k}\sum_{i\in I_{k'}}(\hat{\varphi}^{(-k,-k')}(Z_i)-{\bm{W}_i}^\top \hat{{\bm{\beta}}}^{(-k)})^2+\lambda_{\beta}\left\|\hat{{\bm{\beta}}}^{(-k)}\right\|_1\\
&\qquad\leq \frac{1}{\sum_{k'\neq k}|I_{k'}|}
\sum_{k'\neq k}\sum_{i\in I_{k'}}(\hat{\varphi}^{(-k,-k')}(Z_i)-{\bm{W}_i}^\top {\bm{\beta}}^*)^2+\lambda_{\beta}\|{\bm{\beta}}^*\|_1.
\end{align*}
Denote ${\bm{\Delta}}^{\beta}=\hat{{\bm{\beta}}}^{(-k)}-{\bm{\beta}}^*$. It follows that
\begin{align}
&\notag\frac{1}{\sum_{k'\neq k}|I_{k'}|}
\sum_{k'\neq k}\sum_{i\in I_{k'}}\left({\bm{W}_i}^\top {\bm{\Delta}}^{\beta}\right)^2+\lambda_{\beta}\left\|\hat{{\bm{\beta}}}^{(-k)}\right\|_1\\
&\notag\qquad\leq \frac{2}{\sum_{k'\neq k}|I_{k'}|}
\sum_{k'\neq k}\sum_{i\in I_{k'}}(\hat{\varphi}^{(-k,-k')}(Z_i)-{\bm{W}_i}^\top {\bm{\beta}}^*){\bm{W}_i}^\top {\bm{\Delta}}^{\beta}+\lambda_{\beta}\|{\bm{\beta}}^*\|_1,\\
&\qquad=\frac{2}{\sum_{k'\neq k}|I_{k'}|}
\sum_{k'\neq k}\sum_{i\in I_{k'}}(\Delta_{1i}+\Delta_{2i}+\Delta_{3i}+\Delta_{4i}+\Delta_{5i}+
\Delta_{6i}){\bm{W}_i}^\top {\bm{\Delta}}^{\beta}+\lambda_{\beta}\|{\bm{\beta}}^*\|_1,\label{bound:basic}
\end{align}
where
\begin{align*}
 \Delta_{1i}&:=\frac{A_i-\pi^*(\bm{X}_i)}{\pi^*(\bm{X}_i)[1-\pi^*(\bm{X}_i)]}\left[Y_i-\mu_{A_i}^*(\bm{X}_i)\right]+\mu^*_1(\bm{X}_i)-\mu^*_0(\bm{X}_i)-{\bm{W}_i}^\top {\bm{\beta}}^*,\\
 \Delta_{2i}&:=
 \left\{\left[Y_i(1)-\mu_1^*(\bm{X}_i)\right]\left[\frac{A_i}{\hat{\pi}^{(-k,-k')}(\bm{X}_i)}-\frac{A_i}{\pi^*(\bm{X}_i)}\right]\right.\\
 &\left.\qquad-\left[Y_i(0)-\mu_0^*(\bm{X}_i)\right]\left[\frac{1-A_i}{1-\hat{\pi}^{(-k,-k')}(\bm{X}_i)}-\frac{1-A_i}{1-\pi^*(\bm{X}_i)}\right]\right\}\mathbbm{1}_{\mu_1^*(\cdot)=\mu_1(\cdot),\;\mu_0^*(\cdot)=\mu_0(\cdot)},\\
 \Delta_{3i}&:=
 \left\{\left(\mu_1^*(\bm{X}_i)-\hat{\mu}^{(-k,-k')}_1(\bm{X}_i)\right)\left[\frac{A_i}{\pi^*(\bm{X}_i)}-1\right]\right.\\
 &\left.\qquad-\left(\mu_0^*(\bm{X}_i)-\hat{\mu}^{(-k,-k')}_0(\bm{X}_i)\right)\left[\frac{1-A_i}{1-\pi^*(\bm{X}_i)}-1\right]\right\}\mathbbm{1}_{\pi^*(\cdot)=\pi(\cdot)},\\
 \Delta_{4i}&:=
 \left\{A_i\left(\mu_1^*(\bm{X}_i)-\hat{\mu}^{(-k,-k')}_1(\bm{X}_i)\right)\left[\frac{1}{\hat{\pi}^{(-k,-k')}(\bm{X}_i)}-\frac{1}{\pi^*(\bm{X}_i)}\right]\right.\\
 &\left.\qquad-(1-A_i)\left(\mu_0^*(\bm{X}_i)-\hat{\mu}^{(-k,-k')}_0(\bm{X}_i)\right)\left[\frac{1}{1-\hat{\pi}^{(-k,-k')}(\bm{X}_i)}-\frac{1}{1-\pi^*(\bm{X}_i)}\right]\right\},\\
 \Delta_{5i}&:=
 \left\{\left[Y_i(1)-\mu_1^*(\bm{X}_i)\right]\left[\frac{A_i}{\hat{\pi}^{(-k,-k')}(\bm{X}_i)}-\frac{A_i}{\pi^*(\bm{X}_i)}\right]\right.\\
 &\left.\qquad-\left[Y_i(0)-\mu_0^*(\bm{X}_i)\right]\left[\frac{1-A_i}{1-\hat{\pi}^{(-k,-k')}(\bm{X}_i)}-\frac{1-A_i}{1-\pi^*(\bm{X}_i)}\right]\right\}\mathbbm{1}_{\mu_1^*(\cdot)\neq\mu_1(\cdot)\;\text{or}\;\mu_0^*(\cdot)\neq\mu_0(\cdot)},\\
 \Delta_{6i}&:=
 \left\{\left(\mu_1^*(\bm{X}_i)-\hat{\mu}^{(-k,-k')}_1(\bm{X}_i)\right)\left[\frac{A_i}{\pi^*(\bm{X}_i)}-1\right]\right.\\
 &\left.\qquad-\left(\mu_0^*(\bm{X}_i)-\hat{\mu}^{(-k,-k')}_0(\bm{X}_i)\right)\left[\frac{1-A_i}{1-\pi^*(\bm{X}_i)}-1\right]\right\}\mathbbm{1}_{\pi^*(\cdot)\neq\pi(\cdot)}.
\end{align*}

Let \(\Delta_l\) denote an independent copy of \(\Delta_{li}\) for \(1\le l\le6\). The KKT condition characterizing \({\bm{\beta}}^*\) implies $\mathbb{E}_{\bm{X}}[\Delta_1 \bm{W}]=\mathbf{0}.$ In addition, by the law of iterated expectations, $\mathbb{E}_{\bm{X}}[\Delta_2 \bm{W}]=\mathbf{0}$ and $\mathbb{E}_{\bm{X}}[\Delta_3 \bm{W}]=\mathbf{0}$.

Since $|\Delta_1\,\bm{W}^\top {\bm{e}_j}| = |\epsilon\,\bm{W}^\top {\bm{e}_j}|,$
under Assumption \ref{a6}, by Lemmas~\ref{l1} and \ref{l5}, $\|\Delta_1 \bm{W}^\top {\bm{e}_j}\|_{\psi_1}
 \le \|\bm{W}^\top {\bm{e}_j}\|_{\psi_2}\,\|\epsilon\|_{\psi_2}
\le C_1,$ with some constant $C_1>0$. By Lemma D.4 of \cite{8}, for each $1\leq j\leq p$ and any
$t>0$,
\begin{equation*}
 P\left( \left| \frac{1}{\sum_{k'\neq k}|I_{k'}|}
\sum_{k'\neq k}\sum_{i\in I_{k'}}\Delta_{1i}{\bm{W}_i}^\top {\bm{e}_j}\right|\geq h(t) \right)\leq2\exp\left(-t-\log(p)\right),
\end{equation*}
where $h(t)=C_1\{\sqrt{{[t+\log(p)]/{\sum_{k'\neq k}|I_{k'}}|}}+{[t+\log(p)]/\sum_{k'\neq k}|I_{k'}|}\}$. It follows that
\begin{align*}
 &P\left(\left\|\frac{1}{\sum_{k'\neq k}|I_{k'}|}
\sum_{k'\neq k}\sum_{i\in I_{k'}}\Delta_{1i}{\bm{W}_i}\right\|_{\infty}\geq h(t)\right)\\
&\qquad\leq
 \sum_{j=1}^{p}P\left( \left|\frac{1}{\sum_{k'\neq k}|I_{k'}|}
\sum_{k'\neq k}\sum_{i\in I_{k'}}\Delta_{1i}{\bm{W}_i}^\top {\bm{e}_j}\right|\geq h(t) \right)\\
 &\qquad\leq p \exp\left(-t-\log(p)\right)=\exp(-t).
\end{align*}
Since $\sum_{k'\neq k}|I_{k'}|\asymp n$, it follows that 
\begin{equation*}
 \left\|\frac{1}{\sum_{k'\neq k}|I_{k'}|}
\sum_{k'\neq k}\sum_{i\in I_{k'}}\Delta_{1i}{\bm{W}_i}\right\|_{\infty}=O_p\left(\sqrt{\frac{\log(p)}{{n}}}\right).
\end{equation*}

By Corollary of 2.3 of \cite{5} and $\|\bm{X}\|_{\infty}=O(1)$,
\begin{align*}
 &\mathbb{E}_{\bm{X}}\left[\left\|\frac{1}{\sum_{k'\neq k}|I_{k'}|}
\sum_{k'\neq k}\sum_{i\in I_{k'}}\Delta_{2i}{\bm{W}_i}\right\|_{\infty}^2\right]=O_p\left( \frac{[2e\log(p)-e]\mathbb{E}_{\bm{X}}\left(\|\Delta_{2}\bm{W}\|_{\infty}^2\right)}{\sum_{k'\neq k}|I_{k'}|}\right)\\
&\qquad=O_p\left(\frac{\mathbb{E}_{\bm{X}}(\Delta_2^2)\log(p)}{{n}}\right).
\end{align*}
Since $\mathbb{E}(a-b)^2\leq 2 \mathbb{E}(a^2+b^2)$, we have
\begin{align*}
 \mathbb{E}_{\bm{X}}(\Delta_2^2)
 &\leq 2\mathbb{E}_{\bm{X}}\left|\left[Y(1)-\mu_1^*(\bm{X})\right]\left[\frac{A}{\hat{\pi}^{(-k,-k')}(\bm{X})}-\frac{A}{\pi^*(\bm{X})}\right]\right|^2\\
 &\qquad+2\mathbb{E}_{\bm{X}}\left|\left[Y(0)-\mu_0^*(\bm{X})\right]\left[\frac{1-A}{1-\hat{\pi}^{(-k,-k')}(\bm{X})}-\frac{1-A}{1-\pi^*(\bm{X})}\right]\right|^2\\
 &\leq 2\sqrt{\mathbb{E}_{\bm{X}}|A\left[Y(1)-\mu_1^*(\bm{X})\right]|^4}\sqrt{\mathbb{E}_{\bm{X}}\left|\left[\frac{1}{\hat{\pi}^{(-k,-k')}(\bm{X})}-\frac{1}{\pi^*(\bm{X})}\right]\right|^4}\\
 &\qquad+2\sqrt{\mathbb{E}_{\bm{X}}|(1-A)\left[Y(0)-\mu_0^*(\bm{X})\right]|^4}\sqrt{\mathbb{E}_{\bm{X}}\left|\left[\frac{1}{1-\hat{\pi}^{(-k,-k')}(\bm{X})}-\frac{1}{1-\pi^*(\bm{X})}\right]\right|^4}\\
 &\overset{(i)}{=}O_p\left(\frac{s_{{\bm{\gamma}}}\log(d)}{N}\right),
\end{align*}
where (i) holds since Assumption \ref{a5} and Lemma \ref{l4}.
Hence,
\begin{equation*}
 \mathbb{E}_{\bm{X}}\left[\left\|\frac{1}{\sum_{k'\neq k}|I_{k'}|}
\sum_{k'\neq k}\sum_{i\in I_{k'}}\mathbbm{1}_{A_i=a}\Delta_{2i}{\bm{W}_i}\right\|_{\infty}^2\right]=O_p\left(\frac{s_{{\bm{\gamma}}}\log(d)\log(p)}{nN}\right)=o_p\left(\frac{\log(p)}{n}\right),
\end{equation*}
since $s_{{\bm{\gamma}}}\log(d)=o(N)$. 
By Markov's inequality, one obtains
\begin{equation*}
 \left\|\frac{1}{\sum_{k'\neq k}|I_{k'}|}
\sum_{k'\neq k}\sum_{i\in I_{k'}}\Delta_{2i}{\bm{W}_i}\right\|_{\infty}=o_p\left(\sqrt{\frac{\log(p)}{{n}}}\right).
\end{equation*}

Additionally, observe that
\begin{align*}
 \mathbb{E}_{\bm{X}}(\Delta_3^2)
 &\leq 2\mathbb{E}_{\bm{X}}\left|\left(\mu_1^*(\bm{X})-\hat{\mu}^{(-k,-k')}_1(\bm{X})\right)\left[\frac{A}{\pi^*(\bm{X})}-1\right]\right|^2\\
 &\quad+2\mathbb{E}_{\bm{X}}\left|\left(\mu_0^*(\bm{X})-\hat{\mu}^{(-k,-k')}_0(\bm{X})\right)\left[\frac{1-A}{1-\pi^*(\bm{X})}-1\right]\right|^2\\
 &\leq 2\sqrt{\mathbb{E}_{\bm{X}}\left|\left(\mu_1^*(\bm{X})-\hat{\mu}^{(-k,-k')}_1(\bm{X})\right)\right|^4}\sqrt{\mathbb{E}_{\bm{X}}\left|\left[\frac{A}{\pi^*(\bm{X})}-1\right]\right|^4}\\
 &\quad+2\sqrt{\mathbb{E}_{\bm{X}}\left|\left(\mu_0^*(\bm{X})-\hat{\mu}^{(-k,-k')}_0(\bm{X})\right)\right|^4}\sqrt{\mathbb{E}_{\bm{X}}\left|\left[\frac{1-A}{1-\pi^*(\bm{X})}-1\right]\right|^4}\\
 &\overset{(i)}{=}O_p\left(\frac{s_{{\bm{\alpha}}}\log(d)}{{n}}\right),
\end{align*}
where (i) follows from Lemma \ref{l1}, Lemma \ref{l2}, and Assumption \ref{a7}. Since 
$
\|\Delta_{3i}\,{\bm{W}_i}^\top {\bm{e}_j}\|_{\psi_1}
\le \|{\bm{W}_i}^\top {\bm{e}_j}\|_{\psi_2}\,\|\Delta_{3i}\|_{\psi_2}
=O(1),
$
 an application of Bernstein’s inequality then yields
\begin{equation*}
 \mathbb{E}_{\bm{X}}\left[\left\|\frac{1}{\sum_{k'\neq k}|I_{k'}|}
\sum_{k'\neq k}\sum_{i\in I_{k'}}\Delta_{3i}{\bm{W}_i}\right\|_{\infty}^2\right]=O_p\left(\frac{s_{{\bm{\alpha}}}\log(d)\log(p)}{{{n}}^2}\right)\overset{(i)}{=}o_p\left(\frac{\log(p)}{{n}}\right),
\end{equation*}
where (i) holds when $s_{{\bm{\alpha}}}\log(d)=o({n})$. By Markov's inequality,
\begin{equation*}
 \left\|\frac{1}{\sum_{k'\neq k}|I_{k'}|}
\sum_{k'\neq k}\sum_{i\in I_{k'}}\Delta_{3i}{\bm{W}_i}\right\|_{\infty}=o_p\left(\sqrt{\frac{\log(p)}{{n}}}\right).
\end{equation*}

To sum up, 
\[
\sum_{l=1}^{3}\left\|\frac{1}{\sum_{k'\neq k}|I_{k'}|}
\sum_{k'\neq k}\sum_{i\in I_{k'}}\Delta_{li}{\bm{W}_i}\right\|_{\infty}
=O_p\left(\sqrt{\frac{\log(p)}{{n}}}\right).
\]
Then for any \(t>0\), there exists some 
\(\lambda_{\beta}\asymp \sqrt{{\log(p)/{n}}}\) such that
\[
A_2 := \left\{\sum_{l=1}^{3}\left\|\frac{1}{\sum_{k'\neq k}|I_{k'}|}
\sum_{k'\neq k}\sum_{i\in I_{k'}}\Delta_{li}{\bm{W}_i}\right\|_{\infty}
\le \frac{\lambda_{\beta}}{4}\right\}
\]
satisfies \(P(A_2)\geq1-t\). On the event \(A_2\), 
\begin{align*}
 &\frac{1}{\sum_{k'\neq k}|I_{k'}|}
\sum_{k'\neq k}\sum_{i\in I_{k'}}\left({\bm{W}_i}^\top {\bm{\Delta}}^{\beta}\right)^2+\lambda_{\beta}\left\|\hat{{\bm{\beta}}}^{(-k)}\right\|_1
\\
&\quad\leq \frac{\lambda_{\beta}}{2}\|\bm{\Delta}_{\beta}\|_{1} +\frac{2\sum_{k'\neq k}\sum_{i\in I_{k'}} \left(\sum_{l=4}^{6}\Delta_{li}\right)^2}{\sum_{k'\neq k}|I_{k'}|}
+
\frac{\sum_{k'\neq k}\sum_{i\in I_{k'}}\left({\bm{W}_i}^\top {\bm{\Delta}}^{\beta}\right)^2}{2\sum_{k'\neq k}|I_{k'}|}
+\lambda_{\beta}\|{\bm{\beta}}^*\|_1.
\end{align*}
Since $(\sum_{l=4}^{6}\Delta_{li})^2\leq 4\sum_{l=4}^{6}
\Delta_{li}^2$, which yields
\begin{align*}
 &\frac{1}{\sum_{k'\neq k}|I_{k'}|}
\sum_{k'\neq k}\sum_{i\in I_{k'}}\left({\bm{W}_i}^\top {\bm{\Delta}}^{\beta}\right)^2+2\lambda_{\beta}\left\|\hat{{\bm{\beta}}}^{(-k)}\right\|_1\\
&\quad\leq \lambda_{\beta}\|\bm{\Delta}_{\beta}\|_{1} +\frac{16}{\sum_{k'\neq k}|I_{k'}|}
\sum_{k'\neq k}\sum_{i\in I_{k'}}\left(\sum_{l=4}^{6}\Delta_{li}^2\right)+2\lambda_{\beta}\|{\bm{\beta}}^*\|_1.
\end{align*}
Then,
\begin{align*}
&\frac{1}{\sum_{k'\neq k}|I_{k'}|}
\sum_{k'\neq k}\sum_{i\in I_{k'}}\left({\bm{W}_i}^\top {\bm{\Delta}}^{\beta}\right)^2+2\lambda_{\beta}\left\|\hat{{\bm{\beta}}}^{(-k)}\right\|_1\\
&\quad\leq \lambda_{\beta}\left\|\hat{{\bm{\beta}}}^{(-k)}_S-{\bm{\beta}}^*_S\right\|_{1} +\frac{16}{\sum_{k'\neq k}|I_{k'}|}
\sum_{k'\neq k}\sum_{i\in I_{k'}}\left(\sum_{l=4}^{6}\Delta_{li}^2\right)+2\lambda_{\beta}\left\|{\bm{\beta}}^*_S\right\|_1+\lambda_{\beta}\left\|\hat{{\bm{\beta}}}^{(-k)}_{S^c}\right\|_1,
\end{align*}
where with a slight abuse of notation, consider $S=\{j\leq p:{\beta}^*_j\neq 0 \}$ and 
\begin{equation*}
 \left\|\hat{{\bm{\beta}}}^{(-k)}-{\bm{\beta}}^*\right\|_{1}=\left\|\hat{{\bm{\beta}}}^{(-k)}_S-{\bm{\beta}}^*_S\right\|_{1}+\left\|\hat{{\bm{\beta}}}^{(-k)}_{S^c}-{\bm{\beta}}^*_{S^c}\right\|_{1}=\left\|\hat{{\bm{\beta}}}^{(-k)}_S-{\bm{\beta}}^*_S\right\|_{1}+\left\|\hat{{\bm{\beta}}}^{(-k)}_{S^c}\right\|_{1}.
\end{equation*}
According to the triangle inequality 
\begin{equation*}
 2\left\|\hat{{\bm{\beta}}}^{(-k)}\right\|_{1}=2\| \hat{{\bm{\beta}}}^{(-k)}_S\|_{1}+2\left\|\hat{{\bm{\beta}}}^{(-k)}_{S^c}\right\|_{1}\geq 2\left\|{\bm{\beta}}^*_S\right\|_{1}-\left\|\hat{{\bm{\beta}}}^{(-k)}_S-{\bm{\beta}}^*_S\right\|_{1}+2\left\|\hat{{\bm{\beta}}}^{(-k)}_{S^c}\right\|_{1}.
\end{equation*}
Hence,
\begin{align}\label{7}
 \notag&\frac{1}{\sum_{k'\neq k}|I_{k'}|}
\sum_{k'\neq k}\sum_{i\in I_{k'}}\left[{\bm{W}_i}^\top {\bm{\Delta}}^{\beta}\right]^2+\lambda_{\beta}\left\|\hat{{\bm{\beta}}}^{(-k)}_{S^c}\right\|_1\\
&\quad\leq 3\lambda_{\beta}\left\|\hat{{\bm{\beta}}}^{(-k)}_S-{\bm{\beta}}^*_S\right\|_{1} +\frac{16}{\sum_{k'\neq k}|I_{k'}|}
\sum_{k'\neq k}\sum_{i\in I_{k'}}\left(\sum_{l=4}^{6}\Delta_{li}^2\right).
\end{align}
Using the Lemma D.6 of \cite{6}, there exist constants $\kappa_1,\kappa_2>0$ such that for all $\|\bm{a}\|_{2}\leq 1$,
\begin{equation}\label{61}
 \frac{1}{\sum_{k'\neq k}|I_{k'}|}
\sum_{k'\neq k}\sum_{i\in I_{k'}}({\bm{W}_i}^\top \bm{a})^2\geq \kappa_1\|\bm{a}\|_{2}\left[\|\bm{a}\|_{2}-\kappa_2\sqrt{\frac{\log(p)}{\sum_{k'\neq k}|I_{k'}|}}\|\bm{a}\|_{1}\right],
\end{equation}
 with probability at least $1 - c_1\exp(-c_2\sum_{k'\neq k}|I_{k'}|),$
for some constants \(c_1,c_2>0\). Although this original lemma is stated for the logistic loss, an analogous bound for the least squares loss follows by retracing the argument of Lemma~4.5 of \cite{6}. Next, define the event 
\begin{equation*}
 A_3:=\left\{\frac{1}{\sum_{k'\neq k}|I_{k'}|}
\sum_{k'\neq k}\sum_{i\in I_{k'}}\left({\bm{W}_i}^\top {\bm{\Delta}}^{\beta}\right)^2\geq \kappa_1\left\|{\bm{\Delta}}^{\beta}\right\|_{2}^2-\kappa_1 \kappa_2\sqrt{\frac{\log(p)}{\sum_{k'\neq k}|I_{k'}|}}\left\|{\bm{\Delta}}^{\beta}\right\|_{1}\left\|{\bm{\Delta}}^{\beta}\right\|_{2}\right\}.
\end{equation*}
By \eqref{61}, it follows that $P(A_3)\ge 1 - c_1\exp(-c_2\sum_{k'\neq k}|I_{k'}|).$
Proceeding by distinguishing two cases yields the desired bound.

Case 1. If $\left\|{\bm{\Delta}}^{\beta}_S\right\|_{1}\leq \lambda_{\beta}^{-1}{16(\sum_{k'\neq k}|I_{k'}|)^{-1}}
\sum_{k'\neq k}\sum_{i\in I_{k'}}\left(\sum_{l=4}^{6}\Delta_{li}^2\right)$, then by \eqref{7}, 
\begin{equation*}
 \left\|{\bm{\Delta}}^{\beta}_{S^c}\right\|_{1}\leq 3\left\|{\bm{\Delta}}^{\beta}_S\right\|_{1}+\lambda_{\beta}^{-1}\frac{16}{\sum_{k'\neq k}|I_{k'}|}
\sum_{k'\neq k}\sum_{i\in I_{k'}}\left(\sum_{l=4}^{6}\Delta_{li}^2\right)\leq \lambda_{\beta}^{-1}\frac{64}{\sum_{k'\neq k}|I_{k'}|}
\sum_{k'\neq k}\sum_{i\in I_{k'}}\left(\sum_{l=4}^{6}\Delta_{li}^2\right).
\end{equation*}
Hence,
\begin{equation*}
 \left\|{\bm{\Delta}}^{\beta}\right\|_{1}=\left\|{\bm{\Delta}}^{\beta}_S\right\|_{1}+\left\|{\bm{\Delta}}^{\beta}_{S^c}\right\|_{1}\leq \lambda_{\beta}^{-1}\frac{80}{\sum_{k'\neq k}|I_{k'}|}
\sum_{k'\neq k}\sum_{i\in I_{k'}}\left(\sum_{l=4}^{6}\Delta_{li}^2\right)
\end{equation*}
and
\begin{align*}
 &\frac{1}{\sum_{k'\neq k}|I_{k'}|}
\sum_{k'\neq k}\sum_{i\in I_{k'}}\left({\bm{W}_i}^\top {\bm{\Delta}}^{\beta}\right)^2\leq 3\lambda_{\beta}\left\|{\bm{\Delta}}^{\beta}_S\right\|_{1}+\frac{16}{\sum_{k'\neq k}|I_{k'}|}
\sum_{k'\neq k}\sum_{i\in I_{k'}}\left(\sum_{l=4}^{6}\Delta_{li}^2\right)\\
&\qquad\leq \frac{64}{\sum_{k'\neq k}|I_{k'}|}
\sum_{k'\neq k}\sum_{i\in I_{k'}}\left(\sum_{l=4}^{6}\Delta_{li}^2\right).
\end{align*}
In addition, conditioning on the event $A_3$,
\begin{align*}
 &\kappa_1\left\|{\bm{\Delta}}^{\beta}\right\|_{2}^2-\kappa_1 \kappa_2\sqrt{\frac{\log(p)}{\sum_{k'\neq k}|I_{k'}|}}\left\|{\bm{\Delta}}^{\beta}\right\|_{1}\left\| {\bm{\Delta}}^{\beta}\right\|_{2}\leq \frac{1}{\sum_{k'\neq k}|I_{k'}|}
\sum_{k'\neq k}\sum_{i\in I_{k'}}\left({\bm{W}_i}^\top {\bm{\Delta}}^{\beta}\right)^2\\
&\qquad\leq \frac{64}{\sum_{k'\neq k}|I_{k'}|}
\sum_{k'\neq k}\sum_{i\in I_{k'}}\left(\sum_{l=4}^{6}\Delta_{li}^2\right).
\end{align*}
It follows that,
\begin{align*}
 &\left\|{\bm{\Delta}}^{\beta}\right\|_{2}\\
 &\quad\leq \frac{\kappa_1\kappa_2\sqrt{\frac{\log(p)}{\sum_{k'\neq k}|I_{k'}|}}\left\|{\bm{\Delta}}^{\beta}\right\|_{1}+\sqrt{\frac{\kappa_1^2\kappa_2^2\log(p)}{\sum_{k'\neq k}|I_{k'}|}\left\|{\bm{\Delta}}^{\beta}\right\|_{1}^2+\frac{256\kappa_1}{\sum_{k'\neq k}|I_{k'}|}
\sum_{k'\neq k}\sum_{i\in I_{k'}}\left(\sum_{l=4}^{6}\Delta_{li}^2\right)}}{2\kappa_1}\\
 &\quad\overset{(i)}{\leq} \kappa_2\sqrt{\frac{\log(p)}{\sum_{k'\neq k}|I_{k'}|}}\left\|{\bm{\Delta}}^{\beta}\right\|_{1}+8\kappa_1^{-\frac{1}{2}}\left(\frac{1}{\sum_{k'\neq k}|I_{k'}|}
\sum_{k'\neq k}\sum_{i\in I_{k'}}\left(\sum_{l=4}^{6}\Delta_{li}^2\right)\right)^{\frac{1}{2}}\\
 &\quad\leq \kappa_2\sqrt{\frac{\log(p)}{\sum_{k'\neq k}|I_{k'}|}}\lambda_{\beta}^{-1}\frac{80}{\sum_{k'\neq k}|I_{k'}|}
\sum_{k'\neq k}\sum_{i\in I_{k'}}\left(\sum_{l=4}^{6}\Delta_{li}^2\right)\\
&\qquad+8\kappa_1^{-\frac{1}{2}}\left(\frac{1}{\sum_{k'\neq k}|I_{k'}|}
\sum_{k'\neq k}\sum_{i\in I_{k'}}\left(\sum_{l=4}^{6}\Delta_{li}^2\right)\right)^{\frac{1}{2}}\\
 &\quad\leq \kappa_2\frac{80}{\sum_{k'\neq k}|I_{k'}|}
\sum_{k'\neq k}\sum_{i\in I_{k'}}\left(\sum_{l=4}^{6}\Delta_{li}^2\right)+8\kappa_1^{-\frac{1}{2}}\left(\frac{1}{\sum_{k'\neq k}|I_{k'}|}
\sum_{k'\neq k}\sum_{i\in I_{k'}}\left(\sum_{l=4}^{6}\Delta_{li}^2\right)\right)^{\frac{1}{2}}\\
 &\quad=O_p\left(\left(\frac{1}{\sum_{k'\neq k}|I_{k'}|}
\sum_{k'\neq k}\sum_{i\in I_{k'}}\left(\sum_{l=4}^{6}\Delta_{li}^2\right)\right)^{\frac{1}{2}}\right),
\end{align*}
where (i) follows from the fact that $|a|+|b|\geq \sqrt{a^2+b^2}$.

Case 2. If $\|{\bm{\Delta}}^{\beta}_S\|_{1}\geq \lambda_{\beta}^{-1}16(\sum_{k'\neq k}|I_{k'}|)^{-1}
\sum_{k'\neq k}\sum_{i\in I_{k'}}(\sum_{l=4}^{6}\Delta_{li}^2)$, then by \eqref{7},
\begin{align}\label{8}
 &\frac{1}{\sum_{k'\neq k}|I_{k'}|}
\sum_{k'\neq k}\sum_{i\in I_{k'}}\left[{\bm{W}_i}^\top {\bm{\Delta}}^{\beta}\right]^2+\lambda_{\beta}\left\|{\bm{\Delta}}^{\beta}_{S^c}\right\|_1\\
&\quad\notag\leq \lambda_{\beta}\left(3\left\|{\bm{\Delta}}^{\beta}_S\right\|_{1} +\lambda_{\beta}^{-1}\frac{16}{\sum_{k'\neq k}|I_{k'}|}
\sum_{k'\neq k}\sum_{i\in I_{k'}}\left(\sum_{l=4}^{6}\Delta_{li}^2\right)\right)\leq 4\lambda_{\beta}\left\|{\bm{\Delta}}^{\beta}_S\right\|_{1},
\end{align}
thus $\|{\bm{\Delta}}^{\beta}_{S^c}\|_{1}\leq 4\|{\bm{\Delta}}^{\beta}_{s}\|_{1}.$
Since $\|{\bm{\Delta}}^{\beta}_S\|_{1}\leq \sqrt{s}\|\bm{\Delta}_S^{{{\beta}}}\|_{2}$, we have
\begin{equation*}
 \left\|{\bm{\Delta}}^{\beta}\right\|_{1}=\left\|{\bm{\Delta}}^{\beta}_S\right\|_{1}+\left\|{\bm{\Delta}}^{\beta}_{S^c}\right\|\leq 5\left\|{\bm{\Delta}}^{\beta}_S\right\|_{1}\leq 5\sqrt{s}\left\|\bm{\Delta}_S^{{{\beta}}}\right\|_{2}.
\end{equation*}
When $\sum_{k'\neq k}|I_{k'}|>100\kappa_2^2s_{\beta}\log(p)$, under the event $A_3$,
\begin{align*}
 \frac{1}{\sum_{k'\neq k}|I_{k'}|}
\sum_{k'\neq k}\sum_{i\in I_{k'}}\left({\bm{W}_i}^\top {\bm{\Delta}}^{\beta}\right)^2&\geq \kappa_1\left\|{\bm{\Delta}}^{\beta}\right\|_{2}^2-5\kappa_1 \kappa_2\sqrt{\frac{s_{\beta}\log(p)}{\sum_{k'\neq k}|I_{k'}|}}\left\|{\bm{\Delta}}^{\beta}\right\|_{2}^2\\
 &\geq \frac{\kappa_1}{2}\left\|{\bm{\Delta}}^{\beta}\right\|_{2}^2\geq \frac{\kappa_1}{2s}\left\|{\bm{\Delta}}^{\beta}_S\right\|_{1}^2.
\end{align*}
Together with \eqref{8}, one obtains
\begin{equation}\label{62}
 \frac{\kappa_1}{2s}\left\|{\bm{\Delta}}^{\beta}_S\right\|_{1}^2\leq
 \frac{1}{\sum_{k'\neq k}|I_{k'}|}
\sum_{k'\neq k}\sum_{i\in I_{k'}}\left({\bm{W}_i}^\top {\bm{\Delta}}^{\beta}\right)^2\leq 4\lambda_{\beta}\left\|{\bm{\Delta}}^{\beta}_S\right\|_{1}.
\end{equation}
Hence, $ \|{\bm{\Delta}}^{\beta}_S\|_{1}\leq 8\kappa_1^{-1}s\lambda_{\beta} $, $ \|{\bm{\Delta}}^{\beta}\|_{1}\leq 5\|{\bm{\Delta}}^{\beta}_S\|_{1}\leq 40\kappa_1^{-1}s\lambda_{\beta}$
and 
\begin{equation*}
 \frac{1}{\sum_{k'\neq k}|I_{k'}|}
\sum_{k'\neq k}\sum_{i\in I_{k'}}\left({\bm{W}_i}^\top {\bm{\Delta}}^{\beta}\right)^2\leq 4\lambda_{\beta}\left\|{\bm{\Delta}}^{\beta}_S\right\|_{1} \leq 32\kappa_1^{-1}\sqrt{s}\lambda_{\beta}.
\end{equation*}
When \(\sum_{k'\neq k}|I_{k'}| > 100\,\kappa_2^2 s_{\beta}\log(p)\), \eqref{62} implies that 
\begin{equation*}
 \left\|{\bm{\Delta}}^{\beta}\right\|_{2}\leq \sqrt{\frac{2}{\kappa_1\sum_{k'\neq k}|I_{k'}|}
\sum_{k'\neq k}\sum_{i\in I_{k'}}\left({\bm{W}_i}^\top {\bm{\Delta}}^{\beta}\right)^2}\leq 8\kappa^{-1}\sqrt{s}\lambda_{\beta}.
\end{equation*}
 To sum up ,when $\sum_{k'\neq k}|I_{k'}|\gg \max [s_{\beta}\log(p),s_{{\bm{\gamma}}}\log(d),s_{{\bm{\alpha}}}\log(d)]$, with $\lambda_{\beta}\asymp \sqrt{{\log(p)/{n}}}$, we have
\begin{equation}\label{rate:betahat-inter}
 \left\|\hat{{\bm{\beta}}}^{(-k)}-{\bm{\beta}}^*\right\|_{2}=\|\bm{\Delta}_{\beta}\|_{2}=O_p\left(\sqrt{\frac{s_{\beta}\log(p)}{n}}+\left(\frac{1}{\sum_{k'\neq k}|I_{k'}|}
\sum_{k'\neq k}\sum_{i\in I_{k'}}\sum_{l=4}^{6}\Delta_{li}^2\right)^{\frac{1}{2}}\right).
\end{equation}

In the following, we further control the error term $(\sum_{k'\neq k}|I_{k'}|)^{-1}
\sum_{k'\neq k}\sum_{i\in I_{k'}}\Delta_{li}^2$ for each $l \in \{4,5,6\}$. We first observe that
\begin{align}
 &\notag\mathbb{E}_{\bm{X}}\left|\frac{1}{\sum_{k'\neq k}|I_{k'}|}
\sum_{k'\neq k}\sum_{i\in I_{k'}}\Delta_{4i}^2\right|=\mathbb{E}_{\bm{X}}(\Delta_4^2)\\
 &\notag\qquad\leq
 2\mathbb{E}_{\bm{X}}\left|(1-A)\left(\mu_0^*(\bm{X})-\hat{\mu}^{(-k,-k')}_0(\bm{X})\right)\left[\frac{1}{1-\hat{\pi}^{(-k,-k')}(\bm{X})}-\frac{1}{1-\pi^*(\bm{X})}\right]\right|^2\\
 &\qquad\qquad+2\mathbb{E}_{\bm{X}}\left|A\left(\mu_1^*(\bm{X})-\hat{\mu}^{(-k,-k')}_1(\bm{X})\right)\left[\frac{1}{\hat{\pi}^{(-k,-k')}(\bm{X})}-\frac{1}{\pi^*(\bm{X})}\right]\right|^2\label{bound:Delta4}\\
 &\notag\qquad\leq 2\sqrt{\mathbb{E}_{\bm{X}}\left|\left(\mu_1^*(\bm{X})-\hat{\mu}^{(-k,-k')}_1(\bm{X})\right)\right|^4}\sqrt{\mathbb{E}_{\bm{X}}\left|\left[\frac{1}{\hat{\pi}^{(-k,-k')}(\bm{X})}-\frac{1}{\pi^*(\bm{X})}\right]\right|^4}\\
 &\notag\qquad\qquad+2\sqrt{\mathbb{E}_{\bm{X}}\left|\left(\mu_0^*(\bm{X})-\hat{\mu}^{(-k,-k')}_0(\bm{X})\right)\right|^4}\sqrt{\mathbb{E}_{\bm{X}}\left|\left[\frac{1}{1-\hat{\pi}^{(-k,-k')}(\bm{X})}-\frac{1}{1-\pi^*(\bm{X})}\right]\right|^4}\\
 &\notag\qquad\overset{(i)}{=}O_p\left(\frac{s_{{\bm{\alpha}}}\log(d)}{{n}}\frac{s_{{\bm{\gamma}}}\log(d)}{N}\right),
\end{align}
where (i) holds from Lemma \ref{l2} and Lemma \ref{l4}. Furthermore,
\begin{align}
 &\notag\mathbb{E}_{\bm{X}}\left|\frac{1}{\sum_{k'\neq k}|I_{k'}|}
\sum_{k'\neq k}\sum_{i\in I_{k'}}\Delta_{5i}^2\right|=\mathbb{E}_{\bm{X}}(\Delta_5^2)\\
 &\notag\qquad\leq
 2\mathbb{E}_{\bm{X}}\left|A\left[Y(1)-\mu_1^*(\bm{X})\right]\left[\frac{1}{\hat{\pi}^{(-k,-k')}(\bm{X})}-\frac{1}{\pi^*(\bm{X})}\right]\right|^2\\
 &\qquad\qquad+2\mathbb{E}_{\bm{X}}\left|(1-A)\left[Y(0)-\mu_0^*(\bm{X})\right]\left[\frac{1}{1-\hat{\pi}^{(-k,-k')}(\bm{X})}-\frac{1}{1-\pi^*(\bm{X})}\right]\right|^2\label{bound:Delta5}\\
 &\notag\qquad\leq 2\sqrt{\mathbb{E}_{\bm{X}}\left|A\left[Y(1)-\mu_1^*(\bm{X})\right]\right|^4}\sqrt{\mathbb{E}_{\bm{X}}\left|\left[\frac{1}{\hat{\pi}^{(-k,-k')}(\bm{X})}-\frac{1}{\pi^*(\bm{X})}\right]\right|^4}\\
 &\notag\qquad\qquad+2\sqrt{\mathbb{E}_{\bm{X}}|(1-A)\left[Y(0)-\mu_0^*(\bm{X})\right]|^4}\sqrt{\mathbb{E}_{\bm{X}}\left|\left[\frac{1}{1-\hat{\pi}^{(-k,-k')}(\bm{X})}-\frac{1}{1-\pi^*(\bm{X})}\right]\right|^4}\\
 &\notag\qquad\overset{(i)}{=}O_p\left(\frac{s_{{\bm{\gamma}}}\log(d)}{N}\right),
\end{align}
where (i) holds from Assumption \ref{a2} and Lemma \ref{l4}. Besides,
\begin{align*}
 &\mathbb{E}_{\bm{X}}\left|\frac{1}{\sum_{k'\neq k}|I_{k'}|}
\sum_{k'\neq k}\sum_{i\in I_{k'}}\Delta_{6i}^2\right|=\mathbb{E}_{\bm{X}}(\Delta_6^2)\\
 &\qquad\leq
 2\mathbb{E}_{\bm{X}}\left|\left(\mu_1^*(\bm{X})-\hat{\mu}^{(-k,-k')}_1(\bm{X})\right)\left[\frac{A}{\pi^*(\bm{X})}-1\right]\right|^2\\
 &\qquad\qquad
 +2\mathbb{E}_{\bm{X}}\left|\left(\mu_0^*(\bm{X})-\hat{\mu}^{(-k,-k')}_0(\bm{X})\right)\left[\frac{1-A}{1-\pi^*(\bm{X})}-1\right]\right|^2\\
 &\qquad\leq 2\sqrt{\mathbb{E}_{\bm{X}}\left|\left(\mu_1^*(\bm{X})-\hat{\mu}^{(-k,-k')}_1(\bm{X})\right)\right|^4}\sqrt{\mathbb{E}_{\bm{X}}\left|\left[\frac{A}{\pi^*(\bm{X})}-1\right]\right|^4}\\
 &\qquad\qquad
 +2\sqrt{\mathbb{E}_{\bm{X}}\left|\left(\mu_0^*(\bm{X})-\hat{\mu}^{(-k,-k')}_0(\bm{X})\right)\right|^4}\sqrt{\mathbb{E}_{\bm{X}}\left|\left[\frac{1-A}{1-\pi^*(\bm{X})}-1\right]\right|^4}\\
 &\qquad\overset{(i)}{=}O_p\left(\frac{s_{{\bm{\alpha}}}\log(d)}{{n}}\right),
\end{align*}
where (i) holds by Assumption \ref{a4} and Lemma \ref{l2}.

By Markov's inequality,
\begin{align*}
 &\frac{1}{\sum_{k'\neq k}|I_{k'}|}
\sum_{k'\neq k}\sum_{i\in I_{k'}}\sum_{l=4}^{6}\Delta_{li}^2\\
&\qquad=O_p\left(\frac{s_{{\bm{\alpha}}}s_{{\bm{\gamma}}}\log^2(d)}{{{n}}^2}+\frac{s_{{\bm{\gamma}}}\log(d)}{N}\left(\mathbbm{1}_{\mu_1^*(\cdot)\neq \mu_1(\cdot)}+\mathbbm{1}_{\mu_0^*(\cdot)\neq \mu_0(\cdot)}\right)+\frac{s_{{\bm{\alpha}}}\log(d)}{{n}}\mathbbm{1}_{\pi^*(\cdot)\neq \pi(\cdot)}\right).
\end{align*} 
Consequently, when $K$ is finite
\begin{equation*}
 \left\|\hat{{\bm{\beta}}}^{(-k)}-{\bm{\beta}}^*\right\|_{2}=O_p\left(R_n+\sqrt{\frac{s_{\beta}\log(p)}{n}}\right).
\end{equation*}
\end{proof}

\begin{proof}[Proof of Theorem \ref{t2}]
 Define 
\begin{equation*}
 \hat{\theta}^{(k)}_{\mathrm{ETH}}:=\frac{1}{|G_k|}\sum_{i \in G_k}\left(\hat{{\bm{\beta}}}^{(-k)^\top }\hat {\bm{D}}_i^{(k)}\right)^2+\frac{2}{\tilde{n}}\sum_{i\in I_k}\hat{{\bm{\beta}}}^{(-k)^\top }\hat {\bm{D}}_i^{(k)}\left(\hat{\varphi}^{(-k)}(Z_i
 )-\hat{\tau}_{\mathrm{para}}-\hat{{\bm{\beta}}}^{(-k)^\top }\hat {\bm{D}}_i^{(k)}\right).
\end{equation*} 
For any constant $c>0$, by \eqref{93} and H\"older's inequality,
 \begin{align}\label{26}
 &\notag\mathbb{E}_{\bm{X}}\left|\hat{{\bm{\beta}}}^{(-k)^\top }\bm{D}(\hat{\varphi}^{(-k)}(Z)-\tau)\right|^{2+\frac{c}{2}}\\
 &\quad\leq\sqrt{\mathbb{E}_{\bm{X}}\left|\left(\hat{{\bm{\beta}}}^{(-k)^\top }\bm{D}\right)\right|^{4+c}}\sqrt{\mathbb{E}_{\bm{X}}\left|(\hat{\varphi}^{(-k)}(Z)-\tau)\right|^{4+c}}\overset{(i)}{=}O_p(1),
 \end{align}
 where (i) holds by \eqref{92}, \eqref{93}, and \eqref{95}.
Chebyshev’s inequality implies that 
\begin{align}
 &|G_k|^{-1}\sum_{i \in G_k}\left(\hat{{\bm{\beta}}}^{(-k)^\top }
 {{\bm{D}_i}}\right)^2=\mathbb{E}_{\bm{X}}\left(\hat{{\bm{\beta}}}^{(-k)^\top } \bm{D}\right)^2+O_p\left(N^{-\frac{1}{2}}\right),\label{45}\\
 &\tilde{n}^{-1}\sum_{i \in I_k}(\hat{{\bm{\beta}}}^{(-k)^\top }
 {{\bm{D}_i}})^2=\mathbb{E}_{\bm{X}}\left(\hat{{\bm{\beta}}}^{(-k)^\top } \bm{D}\right)^2+O_p\left({{n}}^{-\frac{1}{2}}\right),\label{47}\\
 &\sum_{i \in J_k}(\hat{{\bm{\beta}}}^{(-k)^\top }
 {{\bm{D}_i}})^2=\tilde{m}\mathbb{E}_{\bm{X}}\left(\hat{{\bm{\beta}}}^{(-k)^\top } \bm{D}\right)^2+O_p\left({{m}}^{\frac{1}{2}}\right)\label{126}.
\end{align}
We now demonstrate that 
\begin{equation}\label{29}
 \frac{1}{\tilde{n}}\sum_{i \in I_k}\hat{{\bm{\beta}}}^{(-k)^\top }{\bm{D}_i}(\varphi^*(Z_i
 )-\tau)
 =\mathbb{E}_{\bm{X}}\left(\hat{{\bm{\beta}}}^{(-k)^\top }\bm{D}\bm{D}^\top {\bm{\beta}}^*\right)+O_p\left({{n}}^{-\frac{1}{2}}\right).
\end{equation}
Since 
\begin{equation}\label{30}
 \mathbb{E}_{\bm{X}}\left(\hat{{\bm{\beta}}}^{(-k)^\top }
 \bm{D}\bm{D}^\top {\bm{\beta}}^*\right)^2\overset{(i)}{\leq} \sqrt{\mathbb{E}_{\bm{X}}\left(\hat{{\bm{\beta}}}^{(-k)^\top }
 \bm{D}\right)^4\mathbb{E}\left(\bm{D}^\top {\bm{\beta}}^*\right)^4}\overset{(ii)}{=}O_p(1),
\end{equation}
where (i) holds by Cauchy-Schwarz inequality, (ii) holds since \eqref{92} and \eqref{109}. Besides, according to Cauchy-Schwarz inequality,
\begin{equation}\label{37}
 \mathbb{E}_{\bm{X}}\left(\hat{{\bm{\beta}}}^{(-k)^\top }\bm{D}\epsilon\right)^2\leq \sqrt{\mathbb{E}_{\bm{X}}\left(\hat{{\bm{\beta}}}^{(-k)^\top }\bm{D}\right)^4}\sqrt{\mathbb{E}(\epsilon^4)}\overset{(i)}{=}O_p(1),
\end{equation}
where (i) follows from \eqref{92} and the sub-Gaussianity of \(\epsilon\) in Lemma \ref{l5}, which guarantees bounded moments. By combining \eqref{30} and \eqref{37} with Chebyshev’s inequality, one obtains 
\begin{align} 
 \frac{1}{\tilde{n}}\sum_{i \in I_k}\hat{{\bm{\beta}}}^{(-k)^\top }{\bm{D}_i}{\bm{D}_i}^\top {\bm{\beta}}^*&=\mathbb{E}_{\bm{X}}\left(\hat{{\bm{\beta}}}^{(-k)^\top }\bm{D}\bm{D}^\top {\bm{\beta}}^*\right)+O_p\left({{n}}^{-\frac{1}{2}}\right),\label{31}\\
 \frac{1}{\tilde{n}}\sum_{i \in I_k}\hat{{\bm{\beta}}}^{(-k)^\top }{\bm{D}_i}\epsilon_i&=O_p\left({{n}}^{-\frac{1}{2}}\right).\label{32}
\end{align}
We can derive that 
\begin{align}\label{46}
 &\notag\frac{1}{\tilde{n}}\sum_{i \in I_k}\hat{{\bm{\beta}}}^{(-k)^\top }{\bm{D}_i}(\varphi^*(Z_i
 )-\tau)=\frac{1}{\tilde{n}}\sum_{i \in I_k}\hat{{\bm{\beta}}}^{(-k)^\top }{\bm{D}_i}\left({\bm{W}_i}^\top {\bm{\beta}}^*+\epsilon_i-\mathbb{E}(\bm{W})^\top {\bm{\beta}}^*\right)\\
 &\quad=\frac{1}{\tilde{n}}\sum_{i \in I_k}\hat{{\bm{\beta}}}^{(-k)^\top }{\bm{D}_i}\left({\bm{D}_i}^\top {\bm{\beta}}^*+\epsilon_i\right)\overset{(i)}{=}\mathbb{E}_{\bm{X}}\left(\hat{{\bm{\beta}}}^{(-k)^\top }\bm{D}\bm{D}^\top {\bm{\beta}}^*\right)+O_p\left({{n}}^{-\frac{1}{2}}\right),
\end{align}
where (i) holds by \eqref{31} and \eqref{32}. By Lemma \ref{l5}, \(\|\varphi^*(Z)\|_{\psi_2}=O(1)\). An application of Lemma \ref{l1} together with Chebyshev’s inequality then implies 
\begin{equation}\label{48}
 \frac{1}{\tilde{n}}\sum_{i \in I_k}(\varphi^*(Z_i
 )-\tau)=O_p\left(n^{-\frac{1}{2}}\right).
\end{equation}

Besides, consider the decomposition $ab = cd + (a-c)(b-d) + (a-c)d + (b-d)c$,
\begin{align}\label{43}
 &\notag\frac{1}{\tilde{n}}\sum_{i \in I_k}\hat{{\bm{\beta}}}^{(-k)^\top }\hat {\bm{D}}_i^{(k)}\left[\hat{\varphi}^{(-k)}(Z_i)-\hat{\tau}_{\mathrm{para}}\right]\\
 &\notag\qquad=\frac{1}{\tilde{n}}\sum_{i \in I_k}\hat{{\bm{\beta}}}^{(-k)^\top }{\bm{D}_i}[\varphi^*(Z_i
 )-\tau]+\frac{1}{\tilde{n}}\sum_{i \in I_k}\hat{{\bm{\beta}}}^{(-k)^\top }{\bm{D}_i} \left[\tau-\hat{\tau}_{\mathrm{para}} +\hat{\varphi}^{(-k)}(Z_i)-\varphi^*(Z_i)\right]\\
 &\notag\qquad\qquad+\hat{{\bm{\beta}}}^{(-k)^\top }\left[\mathbb{E}(\bm{W})-|G_k|^{-1}\sum_{i \in G_k}{\bm{W}_i}\right]\\
 &\notag\qquad\qquad\qquad\cdot\frac{1}{\tilde{n}}\sum_{i \in I_k}\left[\tau-\hat{\tau}_{\mathrm{para}} +\hat{\varphi}^{(-k)}(Z_i)-\varphi^*(Z_i)+\varphi^*(Z_i)-\tau\right]\\
 \notag&\qquad\overset{(i)}{=}\frac{1}{\tilde{n}}\sum_{i \in I_k}\hat{{\bm{\beta}}}^{(-k)^\top }{\bm{D}_i}(\varphi^*(Z_i)-\tau)+O_p(n^{-1/2})o_p(1)+O_p(R_n)\\
 &\notag\qquad\qquad+O_p(N^{-1/2})[o_p(1)+o_p(1)+O_p(n^{-1/2})]\\
 &\qquad=\frac{1}{\tilde{n}}\sum_{i \in I_k}\hat{{\bm{\beta}}}^{(-k)^\top }{\bm{D}_i}(\varphi^*(Z_i
 )-\tau)+o_p\left({{n}}^{-{1/2}}\right)+O_p(R_n),
\end{align}
where (i) holds by \eqref{11}, \eqref{75}, \eqref{48}, and the results from Lemmas \ref{l6} and \ref{l7}. Analogously as in \eqref{44}, we have 
\begin{align}\label{bound:ave-square}
 &\notag|I_k|^{-1}\sum_{i \in I_k}\left(\hat{{\bm{\beta}}}^{(-k)^\top }\hat {\bm{D}}_i^{(k)}\right)^2\\
 &\notag\qquad=|I_k|^{-1}\sum_{i \in I_k}\left(\hat{{\bm{\beta}}}^{(-k)^\top }{\bm{D}_i}\right)^2+\left(\hat{{\bm{\beta}}}^{(-k)^\top }\left[|G_k|^{-1}\sum_{i \in G_k}{\bm{W}_i}-\mathbb{E}(\bm{W})\right]\right)^2\\
 &\notag\qquad\qquad-2\hat{{\bm{\beta}}}^{(-k)^\top }\left[|G_k|^{-1}\sum_{i \in G_k}{\bm{W}_i}-\mathbb{E}(\bm{W})\right]|I_k|^{-1}\sum_{i \in I_k}\hat{{\bm{\beta}}}^{(-k)^\top }{\bm{D}_i}\\
 &\qquad=|I_k|^{-1}\sum_{i \in I_k}\left(\hat{{\bm{\beta}}}^{(-k)^\top }{\bm{D}_i}\right)^2+O_p\left(N^{-1/2}n^{-1/2}\right).
\end{align}
Therefore, 
 \begin{align}\label{50}
 \notag\hat{\theta}^{(k)}_{\mathrm{ETH}}&=\frac{1}{|G_k|}\sum_{i \in G_k}\left(\hat{{\bm{\beta}}}^{(-k)^\top }\hat {\bm{D}}_i^{(k)}\right)^2+\frac{2}{\tilde{n}}\sum_{i\in I_k}\hat{{\bm{\beta}}}^{(-k)^\top }\hat {\bm{D}}_i^{(k)}\left(\hat{\varphi}^{(-k)}(Z_i)-\hat{\tau}_{\mathrm{para}}-\hat{{\bm{\beta}}}^{(-k)^\top }\hat {\bm{D}}_i^{(k)}\right)\\
 \notag&\overset{(i)}{=}|G_k|^{-1}\sum_{i \in G_k}\left(\hat{{\bm{\beta}}}^{(-k)^\top }{\bm{D}_i}\right)^2+O_p(N^{-1})-\frac{2}{\tilde{n}}\sum_{i \in I_k}\left(\hat{{\bm{\beta}}}^{(-k)^\top }{\bm{D}_i}\right)^2+O_p\left(N^{-\frac{1}{2}}{{n}}^{-\frac{1}{2}}\right)\\
 &\qquad+\frac{2}{\tilde{n}}\sum_{i \in I_k}\hat{{\bm{\beta}}}^{(-k)^\top }{\bm{D}_i}(\varphi^*(Z_i
 )-\tau)+o_p\left({{n}}^{-\frac{1}{2}}\right)+O_p(R_n)\\
 &\notag\overset{(ii)}{=}\mathbb{E}_{\bm{X}}\left(\hat{{\bm{\beta}}}^{(-k)^\top } \bm{D}\right)^2+O_p\left({N}^{-\frac{1}{2}}\right)-2\mathbb{E}_{\bm{X}}\left(\hat{{\bm{\beta}}}^{(-k)^\top } \bm{D}\right)^2+O_p\left({{n}}^{-\frac{1}{2}}\right)\\
 &\notag\qquad+2\mathbb{E}_{\bm{X}}\left(\hat{{\bm{\beta}}}^{(-k)^\top }\bm{D}\bm{D}^\top {\bm{\beta}}^*\right) +O_p\left({{n}}^{-\frac{1}{2}}\right)+O_p(R_n),
\end{align} 
 where (i) holds by \eqref{44}, \eqref{43}, and \eqref{bound:ave-square}, (ii) holds by \eqref{45}, \eqref{47}, and \eqref{46}. Then,
\begin{align*}
\hat{\theta}^{(k)}_{\mathrm{ETH}}&=\theta_{\mathrm{ETH}}-\mathbb{E}_{\bm{X}}\left[\left(\hat{{\bm{\beta}}}^{(-k)}-{\bm{\beta}}^*\right)^\top \bm{D}\right]^2 +O_p\left({{n}}^{-\frac{1}{2}}+R_n\right)\\
 &\overset{(i)}{=}\theta_{\mathrm{ETH}}+O_p\left(\left\|\hat{{\bm{\beta}}}^{(-k)}-{\bm{\beta}}^*\right\|_{2}^2+{{n}}^{-\frac{1}{2}}+R_n\right)\\
 &\overset{(ii)}{=}\theta_{\mathrm{ETH}}+O_p\left(\left(R_n+\sqrt{\frac{s_{\beta}\log(p)}{{n}}}\right)^2+{{n}}^{-\frac{1}{2}}+R_n\right)\\
 &=\theta_{\mathrm{ETH}}+O_p\left(R_n+\frac{s_{\beta}\log(p)}{{n}}+{{n}}^{-\frac{1}{2}}\right),
\end{align*}
where (i) holds by \eqref{40}, (ii) holds by Theorem \ref{t1}.
When $K$ is finite,
$$
 \hat{\theta}_{\mathrm{ETH}}=\theta_{\mathrm{ETH}}+O_p\left(R_n+\frac{s_{\beta}\log(p)}{n}+n^{-\frac{1}{2}}\right).
$$

Lastly, we establish the asymptotic normality of $\hat{\theta}_{\mathrm{ETH}}$. By Chebyshev's inequality,
\begin{align}
 |I_k|^{-1}\sum_{i \in I_k}\left(\hat{{\bm{\beta}}}^{(-k)}-{\bm{\beta}}^*\right)^\top {\bm{D}_i}{\bm{D}_i}^\top {\bm{\beta}}^*&=\mathbb{E}_{\bm{X}} \left[\left(\hat{{\bm{\beta}}}^{(-k)}-{\bm{\beta}}^*\right)^\top \bm{D}\bm{D}^\top {\bm{\beta}}^*\right]+o_p\left(n^{-\frac{1}{2}}\right),\label{19'}\\
 \frac{1}{\tilde{n}}\sum_{i \in I_k}\left(\hat{{\bm{\beta}}}^{(-k)}-{\bm{\beta}}^*\right)^\top {\bm{D}_i}\epsilon_i&=o_p\left({{n}}^{-\frac{1}{2}}\right).\label{49}
\end{align}
Thus
\begin{align}\label{51}
 &\notag\frac{1}{\tilde{n}}\sum_{i \in I_k}\hat{{\bm{\beta}}}^{(-k)^\top }{\bm{D}_i}(\varphi^*(Z_i
 )-\tau)\\
 &\quad\notag=\frac{1}{\tilde{n}}\sum_{i \in I_k}\left(\hat{{\bm{\beta}}}^{(-k)}-{\bm{\beta}}^*\right)^\top {\bm{D}_i}(\varphi^*(Z_i
 )-\tau)+\frac{1}{\tilde{n}}\sum_{i \in I_k}{\bm{\beta}}^{*^\top }{\bm{D}_i}(\varphi^*(Z_i
 )-\tau),\\
 \notag&\quad=\frac{1}{\tilde{n}}\sum_{i \in I_k}\left(\hat{{\bm{\beta}}}^{(-k)}-{\bm{\beta}}^*\right)^\top {\bm{D}_i}\left({\bm{D}_i}^\top {\bm{\beta}}^*+\epsilon_i\right)+\frac{1}{\tilde{n}}\sum_{i \in I_k}{\bm{\beta}}^{*^\top }{\bm{D}_i}(\varphi^*(Z_i
 )-\tau),\\
 &\quad\overset{(i)}{=}\frac{1}{\tilde{n}}\sum_{i \in I_k}{\bm{\beta}}^{*^\top }{\bm{D}_i}(\varphi^*(Z_i
 )-\tau)+\mathbb{E}_{\bm{X}}\left[\left(\hat{{\bm{\beta}}}^{(-k)}-{\bm{\beta}}^*\right)^\top \bm{D}\bm{D}^\top {\bm{\beta}}^*\right]+o_p\left({{n}}^{-\frac{1}{2}}\right),
\end{align}
where (i) holds from \eqref{19'} and \eqref{49}. Furthermore,
\begin{align}\label{52'}
 \notag&|I_k|^{-1}\sum_{i \in I_k}\left(\hat{{\bm{\beta}}}^{(-k)^\top }{\bm{D}_i}\right)^2\\
 &\quad\notag=|I_k|^{-1}\sum_{i \in I_k}\left({\bm{\beta}}^{*^\top }{\bm{D}_i}\right)^2+|I_k|^{-1}\sum_{i \in I_k}\left(\left(\hat{{\bm{\beta}}}^{(-k)}-{\bm{\beta}}^*\right)^\top {\bm{D}_i}\right)^2\\
 &\notag\qquad+2|I_k|^{-1}\sum_{i \in I_k}\left(\hat{{\bm{\beta}}}^{(-k)}-{\bm{\beta}}^*\right)^\top {\bm{D}_i}{\bm{D}_i}^\top {\bm{\beta}}^*\\
 \notag&\quad=|I_k|^{-1}\sum_{i \in I_k}\left({\bm{\beta}}^{*^\top }{\bm{D}_i}\right)^2+2\mathbb{E}_{\bm{X}}\left[\left(\hat{{\bm{\beta}}}^{(-k)}-{\bm{\beta}}^*\right)^\top \bm{D}\bm{D}^\top {\bm{\beta}}^*\right]\\
 &\notag\qquad+O_p\left(\left\|\hat{{\bm{\beta}}}^{(-k)}-{\bm{\beta}}^*\right\|_{2}^2+\left\|\hat{{\bm{\beta}}}^{(-k)}-{\bm{\beta}}^*\right\|_{2}n^{-\frac{1}{2}}\right)\\
 &\quad\overset{(i)}{=}|I_k|^{-1}\sum_{i \in I_k}\left({\bm{\beta}}^{*^\top }{\bm{D}_i}\right)^2+2\mathbb{E}_{\bm{X}}\left[\left(\hat{{\bm{\beta}}}^{(-k)}-{\bm{\beta}}^*\right)^\top \bm{D}\bm{D}^\top {\bm{\beta}}^*\right]+o_p({n}^{-\frac{1}{2}}),
\end{align}
where (i) holds under \eqref{39}. Hence,
\begin{align}
 \notag\hat{\theta}^{(k)}_{\mathrm{ETH}}&\overset{(i)}{=}|G_k|^{-1}\sum_{i \in G_k}\left(\hat{{\bm{\beta}}}^{(-k)^\top }{\bm{D}_i}\right)^2+O_p(N^{-1})-\frac{2}{\tilde{n}}\sum_{i \in I_k}\left(\hat{{\bm{\beta}}}^{(-k)^\top }{\bm{D}_i}\right)^2+O_p\left(N^{-\frac{1}{2}}{{n}}^{-\frac{1}{2}}\right)\\
 \notag&\qquad+\frac{2}{\tilde{n}}\sum_{i \in I_k}\hat{{\bm{\beta}}}^{(-k)^\top }{\bm{D}_i}(\varphi^*(Z_i
 )-\tau)+o_p\left({{n}}^{-\frac{1}{2}}\right)+O_p(R_n)\\
 &\notag \overset{(ii)}{=}|G_k|^{-1}\sum_{i \in G_k}\left({\bm{\beta}}^{*^\top }{\bm{D}_i}\right)^2+2\mathbb{E}_{\bm{X}}\left[\left(\hat{{\bm{\beta}}}^{(-k)}-{\bm{\beta}}^*\right)^\top \bm{D}\bm{D}^\top {\bm{\beta}}^*\right]+o_p({n}^{-\frac{1}{2}})\\
 &\notag \qquad-\frac{2}{\tilde{n}}\sum_{i \in I_k}\left({\bm{\beta}}^{*^\top }{\bm{D}_i}\right)^2-4\mathbb{E}_{\bm{X}}\left[\left(\hat{{\bm{\beta}}}^{(-k)}-{\bm{\beta}}^*\right)^\top \bm{D}\bm{D}^\top {\bm{\beta}}^*\right]+o_p\left({{n}}^{-\frac{1}{2}}\right)\\
 &\notag \qquad+\frac{2}{\tilde{n}}\sum_{i \in {I_k}}(\varphi^*(Z_i
 )-\tau){\bm{\beta}}^{*^\top }{\bm{D}_i}+2\mathbb{E}_{\bm{X}}\left[\left(\hat{{\bm{\beta}}}^{(-k)}-{\bm{\beta}}^*\right)^\top \bm{D}\bm{D}^\top {\bm{\beta}}^*\right]+o_p\left({{n}}^{-\frac{1}{2}}\right)\\
 &=\frac{2}{\tilde{n}}\sum_{i \in {I_k}}\epsilon_i{\bm{\beta}}^{*^\top }{\bm{D}_i}+|G_k|^{-1}\sum_{i \in G_k}\left({\bm{\beta}}^{*^\top }{\bm{D}_i}\right)^2 +o_p\left({{n}}^{-\frac{1}{2}}\right), \label{128}
\end{align}
where (i) follows from \eqref{50}, (ii) from \eqref{52} and \eqref{51}.

Next we can decompose $\hat{\theta}_{\mathrm{ETH}}-\theta_{\mathrm{ETH}}$ as
\begin{align}\label{20}
 \hat{\theta}_{\mathrm{ETH}}-\theta_{\mathrm{ETH}}&=\frac{2}{n}\sum_{i=1}^{n}\epsilon_i{\bm{\beta}}^{*^\top }{\bm{D}_i}+\frac{1}{N}\sum_{i=1}^{N}\left({\bm{\beta}}^{*^\top }{\bm{D}_i}\right)^2 -\theta_{\mathrm{ETH}}+o_p\left(n^{-\frac{1}{2}}\right)\\
 \notag&=\frac{1}{n}\sum_{i=1}^{N}\left[2\mathbbm{1}_{i\leq n}\epsilon_i{\bm{\beta}}^{*^\top }{\bm{D}_i}+\left(\sqrt{\frac{n}{N}}{\bm{\beta}}^{*^\top }{\bm{D}_i}\right)^2-\frac{n}{N}\theta_{\mathrm{ETH}}\right]+o_p\left(n^{-\frac{1}{2}}\right).
\end{align}

By Minkovski's inequality, for any constant $\delta>0$,
\begin{align}\label{53}
 \notag&\left(\mathbb{E}_{\bm{X}}\left|2\epsilon{\bm{\beta}}^{*^\top }\bm{D}+\left(\sqrt{\frac{n}{N}}{\bm{\beta}}^{*^\top }\bm{D}\right)^2-\frac{n}{N}\theta_{\mathrm{ETH}}\right|^{2+\delta}\right)^{\frac{1}{2+\delta}}\\
 &\qquad\leq \left(\mathbb{E}_{\bm{X}}\left|2\epsilon{\bm{\beta}}^{*^\top }\bm{D}\right|^{2+\delta}\right)^{\frac{1}{2+\delta}}+\left(\mathbb{E}_{\bm{X}}\left|\left(\sqrt{\frac{n}{N}}{\bm{\beta}}^{*^\top }\bm{D}\right)^2-\frac{n}{N}\theta_{\mathrm{ETH}}\right|^{2+\delta}\right)^{\frac{1}{2+\delta}}\overset{(i)}{=}O(1),
\end{align}
where (i) holds by \eqref{109}, \eqref{18}, and \eqref{57}.
Then
\begin{align}\label{54}
 \notag&\sum_{i=1}^{N}\mathbb{E}_{\bm{X}}\left[\left|2\mathbbm{1}_{i\leq n}\epsilon_i{\bm{\beta}}^{*^\top }{\bm{D}_i}+\left(\sqrt{\frac{n}{N}}{\bm{\beta}}^{*^\top }{\bm{D}_i}\right)^2-\frac{n}{N}\theta_{\mathrm{ETH}}\right|\right]^{2+\delta}\\
 \notag&\qquad=\sum_{i=1}^{n}\mathbb{E}_{\bm{X}}\left[\left|2\epsilon_i{\bm{\beta}}^{*^\top }{\bm{D}_i}+\left(\sqrt{\frac{n}{N}}{\bm{\beta}}^{*^\top }{\bm{D}_i}\right)^2-\frac{n}{N}\theta_{\mathrm{ETH}}\right|\right]^{2+\delta}\\
 &\notag\qquad\qquad+\sum_{i=n+1}^{N}\mathbb{E}_{\bm{X}}\left[\left|\left(\sqrt{\frac{n}{N}}{\bm{\beta}}^{*^\top }{\bm{D}_i}\right)^2-\frac{n}{N}\theta_{\mathrm{ETH}}\right|\right]^{2+\delta}\\
 \notag&\qquad\overset{(i)}{=}O\left(n+m\left(\frac{n}{N}\right)^{2+\delta}\right)\overset{(ii)}{=}O(n),
\end{align}
where (i) holds from \eqref{53} and the sub-Gaussian property of $\bm{D}$, (ii) holds since $m({n/N})^{2+\delta}\leq m({n/N})^{2}\leq n$.
Define
\begin{align*}
 B_{N}^2
 &:=\sum_{i=1}^{N}\Var\left[2\mathbbm{1}_{i\leq n}\epsilon_i{\bm{\beta}}^{*^\top }{\bm{D}_i}+\left(\sqrt{\frac{n}{N}}{\bm{\beta}}^{*^\top }{\bm{D}_i}\right)^2-\frac{n}{N}\theta_{\mathrm{ETH}}\right],\\
 &=\sum_{i=1}^{n}\Var\left[2\epsilon_i{\bm{\beta}}^{*^\top }{\bm{D}_i}+\left(\sqrt{\frac{n}{N}}{\bm{\beta}}^{*^\top }{\bm{D}_i}\right)^2-\frac{n}{N}\theta_{\mathrm{ETH}}\right]\\
 &\qquad+\sum_{i=n+1}^{N}\Var\left[\left(\sqrt{\frac{n}{N}}{\bm{\beta}}^{*^\top }{\bm{D}_i}\right)^2-\frac{n}{N}\theta_{\mathrm{ETH}}\right]\\
 &\geq \sum_{i=1}^{n}\Var\left[2\epsilon_i{\bm{\beta}}^{*^\top }{\bm{D}_i}+\left(\sqrt{\frac{n}{N}}{\bm{\beta}}^{*^\top }{\bm{D}_i}\right)^2-\frac{n}{N}\theta_{\mathrm{ETH}}\right]=n\sigma_{\mathrm{para}}^2.
\end{align*}
Under the assumptions of Theorem \ref{t2}, we have $\sigma_{\mathrm{para}}^2>c$ with some constant $c>0$. Hence,
\begin{equation*}
\frac{\sum_{i=1}^{N}\mathbb{E} \left|2\mathbbm{1}_{i\leq n}E\epsilon_i{\bm{\beta}}^{*^\top }{\bm{D}_i}+\left(\sqrt{\frac{n}{N}}{\bm{\beta}}^{*^\top }{\bm{D}_i}\right)^2-\frac{n}{N}\theta_{\mathrm{ETH}}\right|^{2+\delta}}{\left(B_{N}^2\right)^{1+\frac{\delta}{2}}}
 =O\left(\frac{n}{(n\sigma_{\mathrm{para}}^2)^{1+{\delta/2}}}\right)
 =o(1).
\end{equation*}
By Lindeberg-Feller central limit theorem and Slutsky's Theorem,
\begin{equation}
 \frac{\sqrt{n}(\hat{\theta}_{\mathrm{ETH}}-\theta_{\mathrm{ETH}})}{\sigma_{\mathrm{para}}}\rightarrow N(0, 1).
\end{equation}
\end{proof}

\begin{proof}[Proof of Theorem \ref{t3}]
Define
\begin{align}\label{133}
 \notag\hat{\theta}^{(k)}_{\mathrm{OW}}
 &:=\hat{w}^{(k)}_U\sum_{i \in J_k}\left(\hat{{\bm{\beta}}}^{(-k)^\top }\hat {\bm{D}}_i^{(k)}\right)^2+\hat{w}^{(k)}_L\sum_{i \in I_k}\left(\hat{{\bm{\beta}}}^{(-k)^\top }\hat {\bm{D}}_i^{(k)}\right)^2\\
 &\qquad+\frac{2}{{n}}\sum_{i\in I_k}\hat{{\bm{\beta}}}^{(-k)^\top }\hat {\bm{D}}_i^{(k)}\left(\hat{\varphi}^{(-k)}(Z_i
 )-\hat{\tau}_{\mathrm{para}}-\hat{{\bm{\beta}}}^{(-k)^\top }\hat {\bm{D}}_i^{(k)}\right).
\end{align}
Then
\begin{align}
 &\notag\frac{2}{\tilde{n}}\sum_{i\in I_k}\hat{{\bm{\beta}}}^{(-k)^\top }\hat {\bm{D}}_i^{(k)}\left[\hat{\varphi}^{(-k)}(Z_i)-\hat{\tau}_{\mathrm{para}}-\hat{{\bm{\beta}}}^{(-k)^\top }\hat {\bm{D}}_i^{(k)}\right]\\
 &\notag\qquad\overset{(i)}{=}\frac{2}{\tilde{n}}\sum_{i \in I_k}\hat{{\bm{\beta}}}^{(-k)^\top }{\bm{D}_i}[\varphi^*(Z_i
 )-\tau]-\frac{2}{\tilde{n}}\sum_{i \in I_k}\left(\hat{{\bm{\beta}}}^{(-k)^\top }{\bm{D}_i}\right)^2+o_p({{n}}^{-{1/2}})+O_p(R_n)\\
 &\notag\qquad\overset{(ii)}{=}\frac{2}{\tilde{n}}\sum_{i \in I_k}{\bm{\beta}}^{*^\top }{\bm{D}_i}[\varphi^*(Z_i
 )-\tau]-\frac{2}{\tilde{n}}\sum_{i \in I_k}\left({\bm{\beta}}^{*^\top }{\bm{D}_i}\right)^2\\
 &\qquad\qquad-2\mathbb{E}_{\bm{X}}\left[\left(\hat{{\bm{\beta}}}^{(-k)}-{\bm{\beta}}^*\right)^\top \bm{D}\bm{D}^\top {\bm{\beta}}^*\right]+o_p({{n}}^{-{1/2}}),\label{rate:debias}
\end{align}
where (i) holds by \eqref{43} and \eqref{bound:ave-square}, (ii) holds by \eqref{51} and \eqref{52'} when $R_n=o(n^{-1/2})$. Here, $R_n=o(n^{-1/2})$ occurs since
\begin{align*}
R_n&= \sqrt{\frac{s_{\alpha}s_{\gamma}\log^2(d)}{nN}}
 + \sqrt{\frac{s_{\gamma}\log(d)}{N}}\Bigl(\mathbbm{1}_{\mu_1^*(\cdot)\neq \mu_1(\cdot)}+\mathbbm{1}_{\mu_0^*(\cdot)\neq \mu_0(\cdot)}\Bigr)
 + \sqrt{\frac{s_{\alpha}\log(d)}{n}}\mathbbm{1}_{\pi^*(\cdot)\neq \pi(\cdot)}\\
 &=o\left(\sqrt\frac{N}{nN}+\sqrt\frac{N/n}{N}+0\right)=o(n^{-1/2})
\end{align*}
under the assumed conditions. Moreover,
\begin{align}\label{125}
 &\notag|J_k|\sum_{i \in J_k}\left(\hat{{\bm{\beta}}}^{(-k)^\top }\hat {\bm{D}}_i^{(k)}\right)^2\\
 &\notag\qquad=|J_k|\sum_{i \in J_k}\left(\hat{{\bm{\beta}}}^{(-k)^\top }{\bm{D}_i}\right)^2+\left(\hat{{\bm{\beta}}}^{(-k)^\top }\left[|G_k|^{-1}\sum_{i \in G_k}{\bm{W}_i}-\mathbb{E}(\bm{W})\right]\right)^2\\
 &\notag\qquad\qquad-2\hat{{\bm{\beta}}}^{(-k)^\top }\left[|G_k|^{-1}\sum_{i \in G_k}{\bm{W}_i}-\mathbb{E}(\bm{W})\right]|J_k|\sum_{i\in J_k}\hat{{\bm{\beta}}}^{(-k)^\top }{\bm{D}_i}\\
 &\qquad=|J_k|\sum_{i \in J_k}\left(\hat{{\bm{\beta}}}^{(-k)^\top }{\bm{D}_i}\right)^2+O_p\left(N^{-\frac{1}{2}}{{m}}^{-\frac{1}{2}}\right).
\end{align}
Analogous to \eqref{52} and \eqref{52'},
\begin{align}\label{52''}
 \notag&|J_k|^{-1}\sum_{i \in J_k}\left(\hat{{\bm{\beta}}}^{(-k)^\top }{\bm{D}_i}\right)^2\\
 &\quad\notag=|J_k|^{-1}\sum_{i \in J_k}\left({\bm{\beta}}^{*^\top }{\bm{D}_i}\right)^2+|J_k|^{-1}\sum_{i \in J_k}\left(\left(\hat{{\bm{\beta}}}^{(-k)}-{\bm{\beta}}^*\right)^\top {\bm{D}_i}\right)^2\\
 &\notag\qquad+2|J_k|^{-1}\sum_{i \in J_k}\left(\hat{{\bm{\beta}}}^{(-k)}-{\bm{\beta}}^*\right)^\top {\bm{D}_i}{\bm{D}_i}^\top {\bm{\beta}}^*\\
 \notag&\quad=|J_k|^{-1}\sum_{i \in J_k}\left({\bm{\beta}}^{*^\top }{\bm{D}_i}\right)^2+2\mathbb{E}_{\bm{X}}\left[\left(\hat{{\bm{\beta}}}^{(-k)}-{\bm{\beta}}^*\right)^\top \bm{D}\bm{D}^\top {\bm{\beta}}^*\right]\\
 &\notag\qquad+O_p\left(\left\|\hat{{\bm{\beta}}}^{(-k)}-{\bm{\beta}}^*\right\|_{2}^2+\left\|\hat{{\bm{\beta}}}^{(-k)}-{\bm{\beta}}^*\right\|_{2}n^{-\frac{1}{2}}\right)\\
 &\quad\overset{(i)}{=}|J_k|^{-1}\sum_{i \in J_k}\left({\bm{\beta}}^{*^\top }{\bm{D}_i}\right)^2+2\mathbb{E}_{\bm{X}}\left[\left(\hat{{\bm{\beta}}}^{(-k)}-{\bm{\beta}}^*\right)^\top \bm{D}\bm{D}^\top {\bm{\beta}}^*\right]+o_p({m}^{-\frac{1}{2}}).
\end{align}

Based on the definitions in \eqref{def:omega},
\begin{align}\label{bound:weights}
|\omega_L^*|\leq\frac{1}{N}+\frac{m|C|}{nNB}=O\left(\frac{1}{N}+\frac{m}{nN}\right)=O(n^{-1}),\quad|\omega_U^*|\leq\frac{1}{N}+\frac{|C|}{NB}=O(N^{-1}).
\end{align}
Therefore,
\begin{align}
 \notag&\hat{w}^{(k)}_U\sum_{i \in J_k}\left(\hat{{\bm{\beta}}}^{(-k)^\top }\hat {\bm{D}}_i^{(k)}\right)^2+\hat{w}^{(k)}_L\sum_{i \in I_k}\left(\hat{{\bm{\beta}}}^{(-k)^\top }\hat {\bm{D}}_i^{(k)}\right)^2
 \\
 &\qquad\notag=w_U^*\sum_{i \in J_k}\left(\hat{{\bm{\beta}}}^{(-k)^\top }\hat {\bm{D}}_i^{(k)}\right)^2+\left(\hat{w}^{(k)}_U-w_U^*\right)\sum_{i \in J_k}\left(\hat{{\bm{\beta}}}^{(-k)^\top }\hat {\bm{D}}_i^{(k)}\right)^2\\
 &\qquad\qquad\notag+w_L^*\sum_{i \in I_k}\left(\hat{{\bm{\beta}}}^{(-k)^\top }\hat {\bm{D}}_i^{(k)}\right)^2+\left(\hat{w}^{(k)}_L-w_L^*\right)\sum_{i \in I_k}\left(\hat{{\bm{\beta}}}^{(-k)^\top }\hat {\bm{D}}_i^{(k)}\right)^2,\\
 &\notag\qquad\overset{(i)}{=}w_U^*\sum_{i \in J_k}\left(\hat{{\bm{\beta}}}^{(-k)^\top }{\bm{D}_i}\right)^2+O_p\left(\sqrt\frac{m}{N}N^{-1}\right)+w_L^*\sum_{i \in I_k}\left(\hat{{\bm{\beta}}}^{(-k)^\top }{\bm{D}_i}\right)^2+O_p\left(\sqrt\frac{n}{N}n^{-1}\right)\\
 &\notag\qquad\qquad+\left(\hat{w}^{(k)}_U-w_U^*\right)\sum_{i \in J_k}\left(\hat{{\bm{\beta}}}^{(-k)^\top }\hat {\bm{D}}_i^{(k)}\right)^2+\left(\hat{w}^{(k)}_L-w_L^*\right)\sum_{i \in I_k}\left(\hat{{\bm{\beta}}}^{(-k)^\top }\hat {\bm{D}}_i^{(k)}\right)^2,\\
 &\notag\qquad\overset{(ii)}{=}w_U^*\sum_{i \in J_k}\left({\bm{\beta}}^{*^\top }{\bm{D}_i}\right)^2+2w_U^*\tilde{m}\mathbb{E}_{\bm{X}}\left[\left(\hat{{\bm{\beta}}}^{(-k)}-{\bm{\beta}}^*\right)^\top \bm{D}\bm{D}^\top {\bm{\beta}}^*\right]+o_p(w_U^*{m}^{\frac{1}{2}})\\
 &\notag\qquad\qquad+w_L^*\sum_{i \in I_k}\left({\bm{\beta}}^{*^\top }{\bm{D}_i}\right)^2+2w_L^*\tilde{n}\mathbb{E}_{\bm{X}}\left[\left(\hat{{\bm{\beta}}}^{(-k)}-{\bm{\beta}}^*\right)^\top \bm{D}\bm{D}^\top {\bm{\beta}}^*\right]+o_p(w_L^*{n}^{\frac{1}{2}})\\
 &\qquad\qquad+\left(\hat{w}^{(k)}_U-w_U^*\right)\sum_{i \in J_k}\left(\hat{{\bm{\beta}}}^{(-k)^\top }\hat {\bm{D}}_i^{(k)}\right)^2+\left(\hat{w}^{(k)}_L-w_L^*\right)\sum_{i \in I_k}\left(\hat{{\bm{\beta}}}^{(-k)^\top }\hat {\bm{D}}_i^{(k)}\right)^2+o_p(n^{-1/2}),\label{66}
\end{align}
where (i) holds from \eqref{bound:ave-square}, \eqref{125}, and \eqref{bound:weights}, (ii) holds since \eqref{52'} and \eqref{52''}. Furthermore,
\begin{align}\label{63}
 &\notag\left(\hat{w}^{(k)}_U-w_U^*\right)\sum_{i \in J_k}\left(\hat{{\bm{\beta}}}^{(-k)^\top }\hat {\bm{D}}_i^{(k)}\right)^2\\
 &\qquad\notag\overset{(i)}{=}\left(\hat{w}^{(k)}_U-w_U^*\right)\left[\tilde{m}\mathbb{E}_{\bm{X}}\left(\hat{{\bm{\beta}}}^{(-k)^\top }\bm{D}\right)^2+O_p(N^{-1/2}m^{1/2}+m^{1/2})\right]\\
 &\qquad\overset{(ii)}{=}\tilde{m}\left(\hat{w}^{(k)}_U-w_U^*\right)\mathbb{E}_{\bm{X}}\left(\hat{{\bm{\beta}}}^{(-k)^\top }\bm{D}\right)^2+o_p(N^{-1/2}),
\end{align}
where (i) holds from \eqref{125} and \eqref{126}, (ii) holds by Lemma \ref{l9}. Similarly,
\begin{align}\label{65}
 &\notag\left(\hat{w}^{(k)}_L-w_L^*\right)\sum_{i \in I_k}\left(\hat{{\bm{\beta}}}^{(-k)^\top }\hat {\bm{D}}_i^{(k)}\right)^2\\
 &\qquad\notag\overset{(i)}{=}\left(\hat{w}^{(k)}_L-w_L^*\right)\left[\tilde{n}\mathbb{E}_{\bm{X}}\left(\hat{{\bm{\beta}}}^{(-k)^\top }\bm{D}\right)^2+O_p(N^{-1/2}n^{1/2}+n^{1/2})\right]\\
 &\qquad\overset{(ii)}{=}\tilde{n}\left(\hat{w}^{(k)}_L-w_L^*\right)\mathbb{E}_{\bm{X}}\left(\hat{{\bm{\beta}}}^{(-k)^\top }\bm{D}\right)^2+o_p(n^{-1/2}),
\end{align}
where (i) holds from \eqref{bound:ave-square} and \eqref{47}, (ii) holds by Lemma \ref{l9}. By construction, we also have
\begin{equation}\label{64}
 mw_U^*+nw_L^*=1, \ m\hat{w}^{(k)}_U+n\hat{w}^{(k)}_L=1. 
\end{equation}

Substituting \eqref{63}-\eqref{64} into \eqref{66} and taking the summation over $k\in\{1,\dots,K\}$,
\begin{align}\label{112}
 \notag&\hat{w}^{(k)}_U\sum_{k=1}^{K}\sum_{i \in J_k}\left(\hat{{\bm{\beta}}}^{(-k)^\top }\hat {\bm{D}}_i^{(k)}\right)^2+\hat{w}^{(k)}_L\sum_{k=1}^{K}\sum_{i \in I_k}\left(\hat{{\bm{\beta}}}^{(-k)^\top }\hat {\bm{D}}_i^{(k)}\right)^2
 \\
 \notag&\quad=w_U^*\sum_{i=n+1}^N\left({\bm{\beta}}^{*^\top }{\bm{D}_i}\right)^2+2Kw_U^*\tilde{m}\mathbb{E}_{\bm{X}}\left[\left(\hat{{\bm{\beta}}}^{(-k)}-{\bm{\beta}}^*\right)^\top \bm{D}\bm{D}^\top {\bm{\beta}}^*\right]+o_p(w_U^*{m}^{\frac{1}{2}})\\
 &\notag\qquad+w_L^*\sum_{i=1}^n\left({\bm{\beta}}^{*^\top }{\bm{D}_i}\right)^2+2Kw_L^*\tilde{n}\mathbb{E}_{\bm{X}}\left[\left(\hat{{\bm{\beta}}}^{(-k)}-{\bm{\beta}}^*\right)^\top \bm{D}\bm{D}^\top {\bm{\beta}}^*\right]+o_p(w_L^*{n}^{\frac{1}{2}})\\
 \notag&\qquad+K\left[\tilde{m}\left(\hat{w}^{(k)}_U-w_U^*\right)+\tilde{n}\left(\hat{w}^{(k)}_L-w_L^*\right)\right]\mathbb{E}_{\bm{X}}\left(\hat{{\bm{\beta}}}^{(-k)^\top }\bm{D}\right)^2+o_p(n^{-1/2})\\
 &\quad=w_U^*\sum_{i=n+1}^N\left({\bm{\beta}}^{*^\top }{\bm{D}_i}\right)^2+w_L^*\sum_{i=1}^n\left({\bm{\beta}}^{*^\top }{\bm{D}_i}\right)^2+2\mathbb{E}_{\bm{X}}\left[\left(\hat{{\bm{\beta}}}^{(-k)}-{\bm{\beta}}^*\right)^\top \bm{D}\bm{D}^\top {\bm{\beta}}^*\right]+o_p(n^{-1/2}).
\end{align}
Combining the results in \eqref{133}, \eqref{rate:debias}, and \eqref{112}, we have
\begin{align*}
 \notag\hat{\theta}_{\mathrm{OW}}&=\sum_{k=1}^K\notag\hat{\theta}^{(k)}_{\mathrm{OW}} \\
 &=\frac{2}{{n}}\sum_{i=1}^n\epsilon_i{\bm{\beta}}^{*^\top }{\bm{D}_i}[\varphi^*(Z_i)-\tau]+w_L^*\sum_{i=1}^n\left({\bm{\beta}}^{*^\top }{\bm{D}_i}\right)^2+w_U^*\sum_{i=n+1}^N\left({\bm{\beta}}^{*^\top }{\bm{D}_i}\right)^2+o_p(n^{-1/2}).
\end{align*}
Hence,
 \begin{align*}
 &\notag\hat{\theta}_{\mathrm{OW}}-\theta_{\mathrm{ETH}}\\
 &\quad\notag=\frac{1}{n}\sum_{i=1}^{N}\left[\mathbbm{1}_{i\leq n}\left[2\epsilon_i{\bm{\beta}}^{*^\top }{\bm{D}_i}+nw_L^*\left(\left({\bm{\beta}}^{*^\top }{\bm{D}_i}\right)^2-\theta_{\mathrm{ETH}}\right)\right]+\mathbbm{1}_{i>n}nw_U^*\left(\left({\bm{\beta}}^{*^\top }{\bm{D}_i}\right)^2-\theta_{\mathrm{ETH}}\right)\right]\\
 &\qquad+o_p\left(n^{-\frac{1}{2}}\right).
 \end{align*}
 By Minkovski's inequality, for any $\delta>0$,
\begin{align}\label{135}
 \notag&\left(\mathbb{E}_{\bm{X}}\left|\left[2\epsilon{\bm{\beta}}^{*^\top }\bm{D}+{n}w_L^*\left(\left({\bm{\beta}}^{*^\top }\bm{D}\right)^2-\theta_{\mathrm{ETH}}\right)\right]\right|^{2+\delta}\right)^{\frac{1}{2+\delta}}\\
 \notag&\quad\leq \left(\mathbb{E}_{\bm{X}}\left|2\epsilon{\bm{\beta}}^{*^\top }\bm{D}\right|^{2+\delta}\right)^{\frac{1}{2+\delta}}+\left(\mathbb{E}_{\bm{X}}\left|
 {n}w_L^*\left[\left({\bm{\beta}}^{*^\top }\bm{D}\right)^2-\theta_{\mathrm{ETH}}\right]\right|^{2+\delta}\right)^{\frac{1}{2+\delta}}\overset{(i)}{=}O(1),
\end{align}
where (i) holds by \eqref{109}, \eqref{18}, \eqref{57}, and Lemma \ref{l9}.
Moreover,
\begin{equation*}
 \mathbb{E}_{\bm{X}}\left|{n}w_U^*\left(\left({\bm{\beta}}^{*^\top }{\bm{D}_i}\right)^2-\theta_{\mathrm{ETH}}\right)\right|^{2+\delta}\overset{(i)}{=}O(1),
\end{equation*}
where (i) holds from \eqref{109}, \eqref{57}, and Lemma \eqref{l9}. Define
\begin{align*}
 C_{N}^2
 &:=\sum_{i=1}^{N}\Var\left[\mathbbm{1}_{i\leq n}\left[2\epsilon_i{\bm{\beta}}^{*^\top }{\bm{D}_i}+nw_L^*\left(\left({\bm{\beta}}^{*^\top }{\bm{D}_i}\right)^2-\theta_{\mathrm{ETH}}\right)\right]+\mathbbm{1}_{i>n}nw_U^*\left(\left({\bm{\beta}}^{*^\top }{\bm{D}_i}\right)^2-\theta_{\mathrm{ETH}}\right)\right]\\
 &=\sum_{i=1}^{n}\Var\left[2\epsilon_i{\bm{\beta}}^{*^\top }{\bm{D}_i}+nw_L^*\left(\left({\bm{\beta}}^{*^\top }{\bm{D}_i}\right)^2-\theta_{\mathrm{ETH}}\right)\right]+\sum_{i=n+1}^{N}\Var\left[nw_U^*\left(\left({\bm{\beta}}^{*^\top }{\bm{D}_i}\right)^2-\theta_{\mathrm{ETH}}\right)\right]\\
 &\geq \sum_{i=1}^{n}\Var\left[2\epsilon_i{\bm{\beta}}^{*^\top }{\bm{D}_i}+nw_L^*\left(\left({\bm{\beta}}^{*^\top }{\bm{D}_i}\right)^2-\theta_{\mathrm{ETH}}\right)\right]=n\sigma_{\mathrm{OW}}^2.
\end{align*}
When $\sigma_{\mathrm{OW}}^2>c$ with some constant $c>0$, we have
\begin{align*}
 &\frac{\sum_{i=1}^{N}\mathbb{E}_{\bm{X}}\left|\mathbbm{1}_{i\leq n}\left[2\epsilon_i{\bm{\beta}}^{*^\top }{\bm{D}_i}+nw_L^*\left(\left({\bm{\beta}}^{*^\top }{\bm{D}_i}\right)^2-\theta_{\mathrm{ETH}}\right)\right]+\mathbbm{1}_{i>n}nw_U^*\left(\left({\bm{\beta}}^{*^\top }{\bm{D}_i}\right)^2-\theta_{\mathrm{ETH}}\right)\right|^{2+\delta}}{(C_{N}^2)^{1+{\delta/2}}}\\
 &\qquad=O\left(\frac{n+m\left(\frac{n}{m+n}\right)^{2+\delta}}{(n\sigma_{\mathrm{OW}})^{1+{\delta/2}}}\right)\overset{(i)}{=}o(1),
\end{align*}
where (i) holds since $m({n/N})^{2+\delta}\leq m({n/N})^{2}\leq n$.

By Lindeberg-Feller central limit theorem and Slutsky's Theorem,
\begin{equation*}
 \frac{\sqrt{n}(\hat{\theta}_{\mathrm{OW}}-\theta_{\mathrm{ETH}})}{\sigma_{\mathrm{OW}}}\rightarrow N(0, 1).
\end{equation*}
Moreover, by Lemma \ref{l10}, $\hat{\sigma}_{\mathrm{OW}}^2 = \sigma_{\mathrm{OW}}^2 + o_p(1).$
\end{proof}

\begin{proof}[Proof of Theorem \ref{t5}]

We first characterize the convergence rate of the CATE estimate \(\hat{{\bm{\beta}}}^{(-k)}\) for any $k\leq K$. Consider the basic inequality \eqref{bound:basic}. Under Assumptions \ref{a2} and \ref{a7},
\begin{align*}
 \mathbb{E}_{\bm{X}}(\Delta_2^2)
 &\leq 2\mathbb{E}_{\bm{X}}\left|\left[Y(1)-\mu_1^*(\bm{X})\right]\left[\frac{A}{\hat{\pi}^{(-k,-k')}(\bm{X})}-\frac{A}{\pi^*(\bm{X})}\right]\right|^2\\
 &\qquad+2\mathbb{E}_{\bm{X}}\left|\left[Y(0)-\mu_0^*(\bm{X})\right]\left[\frac{1-A}{1-\hat{\pi}^{(-k,-k')}(\bm{X})}-\frac{1-A}{1-\pi^*(\bm{X})}\right]\right|^2\\
 &=O_p\left(\mathbb{E}_{\bm{X}}|\hat\pi^{(-k,-k')}(\bm{X}) - \pi(\bm{X})|^2\right)
\end{align*}
and
\begin{align*}
 \mathbb{E}_{\bm{X}}(\Delta_3^2)
 &\leq 2\mathbb{E}_{\bm{X}}\left|\left(\mu_1^*(\bm{X})-\hat{\mu}^{(-k,-k')}_1(\bm{X})\right)\left[\frac{A}{\pi^*(\bm{X})}-1\right]\right|^2\\
 &\quad+2\mathbb{E}_{\bm{X}}\left|\left(\mu_0^*(\bm{X})-\hat{\mu}^{(-k,-k')}_0(\bm{X})\right)\left[\frac{1-A}{1-\pi^*(\bm{X})}-1\right]\right|^2\\
 &=O_p\left(\max_{a\in\{0,1\}}\mathbb{E}_{\bm{X}}|\hat\mu_a^{(-k,-k')}(\bm{X}) - \mu_a^*(\bm{X})|^2\right).
\end{align*}
Repeating the steps in the proof of Theorem \ref{t1}, under Assumption \ref{a3},
\begin{align*}
&\sum_{l=1}^{3}\left\|\frac{1}{\sum_{k'\neq k}|I_{k'}|}
\sum_{k'\neq k}\sum_{i\in I_{k'}}\Delta_{li}{\bm{W}_i}\right\|_{\infty}\\
&\qquad=O_p\left(\sqrt{\frac{\log(p)}{{n}}}+\sqrt\frac{\mathbb{E}_{\bm{X}}(\Delta_2^2)\log(p)}{n}+\sqrt\frac{\mathbb{E}_{\bm{X}}(\Delta_3^2)\log(p)}{n}\right)=O_p\left(\sqrt{\frac{\log(p)}{{n}}}\right).
\end{align*}
This ensures the upper bound \eqref{rate:betahat-inter}.

Moreover, by \eqref{bound:Delta4} and \eqref{bound:Delta5}, we also have
\begin{align*}
\mathbb{E}_{\bm{X}}(\Delta_4^2)&\overset{(i)}{=}o_p(n^{-1}),\\
\mathbb{E}_{\bm{X}}(\Delta_5^2)&\overset{(ii)}{=}o_p(n^{-1}),\\
\Delta_6&\overset{(iii)}{=}0,
\end{align*}
where (i) holds when $\mathbb{E}_{\bm{X}}|\hat\pi^{(-k,-k')}(\bm{X}) - \pi(\bm{X})|^2|\hat\mu_a^{(-k,-k')}({\bm{X}}) - \mu_a({\bm{X}})|^2 = o_p(n^{-1})$ for each $a\in\{0,1\}$, (ii) holds since $\Delta_5=0$ when $\mu_a^*(\cdot)=\mu_a(\cdot)$ for each $a\in\{0,1\}$ and otherwise we have $\mathbb{E}_{\bm{X}}|\hat{\pi}^{(-k,-k')}(\bm{X})-\pi(\bm{X})|^2=o_p(n^{-1})$ under Assumption \ref{a3}, (iii) holds since $\pi^*(\cdot)=\pi(\cdot)$. Hence, repeating the remaining steps in the proof of Theorem \ref{t1} leads to
\begin{align*}
 &\frac{1}{\sum_{k'\neq k}|I_{k'}|}
\sum_{k'\neq k}\sum_{i\in I_{k'}}\sum_{l=4}^{6}\Delta_{li}^2=o_p\left(n^{-1/4}\right),
\end{align*} 
and consequently, when $s_{\beta}\log(p)=o(\sqrt{n})$, we have \eqref{39} holds. Lastly, the remaining results hold by repeating the proof of Theorem \ref{t3}.

\end{proof}

\end{document}